\documentclass{article}
\usepackage[utf8]{inputenc}
\usepackage{newcent}
\usepackage[totalwidth=469pt,totalheight=640pt]{geometry}
\usepackage{amsmath,amssymb,mathrsfs,physics,amsthm,tikz,enumerate}
\usepackage{pgfplots,multirow,pdfpages,mathtools,caption,subcaption,graphicx}
\pgfplotsset{compat=1.17} 
\usetikzlibrary{decorations.pathreplacing,calligraphy}
\usetikzlibrary{shapes.geometric}

\graphicspath{{Figures/}}

\newcount\tableauRow
\newcount\tableauCol

\newenvironment{Tableau}[1]{%
 \tikzpicture[scale=0.6
,draw/.append style={thick,black},
 baseline=(current bounding box.center)]
 \tableauRow=-1.5
 \foreach \Row in {#1} {
 \tableauCol=0.5
 \foreach\k in \Row {
 \draw[thin](\the\tableauCol,\the\tableauRow)rectangle++(1,1);
 \draw[thin](\the\tableauCol,\the\tableauRow)+(0.5,0.5)node{$\k$};
 \global\advance\tableauCol by 1
 }
 \global\advance\tableauRow by -1
 }
}{\endtikzpicture}

\tikzstyle{vertexf}=[auto=left,circle,fill=black,minimum size=10pt,
inner sep=0pt]

\tikzstyle{vertexe}=[auto=left,circle,draw=black,
fill=none,minimum size=10pt,inner sep=0pt, thick]

\def\bfrh{\bold{\rho}}
\newcommand{\Rho}{\mathrm{P}}

\def\fig#1{\includegraphics[width=.25\linewidth]{#1.pdf}}

\def\PI{\mbox{\rm P$_{\rm I}$}}
\def\PII{\mbox{\rm P$_{\rm II}$}}
\def\PIII{\mbox{\rm P$_{\rm III}$}}

\def\PIV{\mbox{\rm P$_{\rm IV}$}}
\def\PV{\mbox{\rm P$_{\rm V}$}}
\def\PVI{\mbox{\rm P$_{\rm VI}$}}

\def\sPV{\hbox{\rm S$_{\rm V}$}}

\def\HV{\mbox{$\mathcal{H}_{\rm V}$}}

\DeclareMathOperator{\Wr}{Wr}
\DeclareMathOperator{\He}{He}

\DeclareMathOperator{\Lag}{L}
\def\bold#1{\boldsymbol{#1}}

\newcommand{\fract}[2]{{\textstyle\frac{#1}{#2}}}

\newcommand{\ri}{\right}
\newcommand{\lf}{\left}
\newcommand\eq{\begin{equation}}
\newcommand\en{\end{equation}}
\newcommand\nn{\nonumber}

\newtheorem{theorem}{Theorem}[section]

\newtheorem{lemma}[theorem]{Lemma}
\newtheorem{corollary}[theorem]{Corollary}
\theoremstyle{definition}
\newtheorem{definition}[theorem]{Definition}
\newtheorem{example}[theorem]{Example}
\newtheorem{remark}[theorem]{Remark}
\newtheorem{remarks}[theorem]{Remarks}
\newtheorem{conjecture}[theorem]{Conjecture}
\numberwithin{figure}{section}
\numberwithin{equation}{section}
\numberwithin{table}{section}

\DeclarePairedDelimiter{\ceil}{\lceil}{\rceil}
\DeclarePairedDelimiter{\floor}{\lfloor}{\rfloor}

\def\d{{\rm d}}
\def\e{{\rm e}}
\def\i{{\rm i}}
\def\a{\alpha}
\def\b{\beta}
\def\ga{\gamma}
\def\de{\delta}
\def\ep{\varepsilon}
\def\ph{\varphi}
\def\k{\kappa}

\def\la{\lambda}
\def\th{\theta}

\def\p{Painlev\'e}

\def\peq{\p\ equation}
\def\peqs{\p\ equations}
\def\sch{Schr\"odinger}
\def\bk{B\"ack\-lund}

\def\bts{B\"ack\-lund transformations}

\newcommand{\D}{{\rm D}}

\newcommand\ZZ{\mathbb{Z}}
\newcommand\CC{\mathbb{C}}
\newcommand\NN{\mathbb{N}}

\newcommand{\Z}{\mathbb{Z}}

\newcommand\hf{\fract{1}{2}}
\newcommand{\gam}{\gamma}

\newcommand{\One}{{\hbox{{\rm 1{\hbox to 1.5pt{\hss\rm1}}}}}}
\renewcommand{\One}{\mathbb{1}}
\renewcommand{\One}{{\rm 1\!\!1}}

\def\bfla{\bold{\la}}
\def\bfLA{\bold{\Lambda}}
\def\bfnu{\bold{\nu}}

\newcommand{\blue}[1]{\textcolor{blue}{#1}}

\newcommand{\comment}[1]{}
\def\beq{\begin{equation}}
\def\eeq{\end{equation}}
\def\pms#1#2#3{\a=#1,\qquad \b=#2,\qquad \ga=#3}
\def\pmn#1#2#3{\a_{m,n}=#1,\qquad \b_{m,n}=#2,\qquad \ga_{m,n}=#3}
\def\pmm#1#2#3{\a_{m,0}=#1,\qquad \b_{m,0}=#2,\qquad \ga_{m,0}=#3}

\newcommand{\deriv}[3][]{\frac{\d^{#1}{#2}}{{\d{#3}}^{#1}}}
\newcommand{\pderiv}[3][]{\frac{\partial^{#1}{#2}}{{\partial{#3}}^{#1}}}
\def\wz{\deriv{w}{z}}
\def\wzz{\deriv[2]{w}{z}}
\def\luk{Lukashevich}
\def\etal{\textit{et al.}}

\newcommand{\KummerM}{M}
\newcommand{\KummerU}{U}
\def\Lag#1#2{L_{#1}^{(#2)}}
\def\Tmn#1#2{T_{#1}^{(#2)}}
\def\TTmn#1#2{\widehat{T}_{#1}^{(#2)}}
\def\ifrac#1#2{#1/#2}
\def\O{\mathcal{O}}
\def\W{\mathcal{W}}
\def\Vor{Vorob'ev}
\def\cdot{{\,\scriptstyle\bullet\,}}

\def\bn#1#2{\Big[\substack{#1\\[2pt]#2}\Big]}
\def\An{\mathcal{A}_n}
\def\Ann{\mathcal{A}_{n+1}}

\def\ds{\displaystyle}
\def\dz{\ds\deriv{}{z}}

\def\A#1{\widetilde{A}_{#1}^{(1)}}
\def\spiv{s\PIV}
\def\spv{s\PV}

\begin{document}
\title{Rational Solutions of the Fifth \p\ Equation.\\ Generalised Laguerre Polynomials}

\author{Peter A. Clarkson and Clare Dunning\\[2.5pt]
School of Mathematics, Statistics and Actuarial Science,\\ University of Kent, Canterbury, CT2 7NF, UK\\
Email: {P.A.Clarkson@kent.ac.uk}, {T.C.Dunning@kent.ac.uk}
}

\maketitle
\begin{abstract}
In this paper rational solutions of the fifth Painlev\'e equation are discussed. There are two classes of rational solutions of the fifth Painlev\'e equation, one expressed in terms of the generalised Laguerre polynomials, which are the main subject of this paper, and the other in terms of the generalised Umemura polynomials. Both the generalised Laguerre polynomials and the generalised Umemura polynomials can be expressed as Wronskians of Laguerre polynomials specified in terms of specific families of partitions. The properties of the generalised Laguerre polynomials are determined and various differential-difference and discrete equations found. The rational solutions of the fifth Painlev\'e equation, the associated $\sigma$-equation and the symmetric fifth Painlev\'e system are expressed in terms of generalised Laguerre polynomials. Non-uniqueness of the solutions in special cases is established and some applications are considered. In the second part of the paper, 
the structure of the roots of the polynomials are {investigated} for all values of the parameter. Interesting transitions between root structures through coalescences at the origin are discovered, with the allowed behaviours controlled by hook data associated with the partition. The discriminants of the generalised Laguerre polynomials are found and also shown to be expressible in terms of partition data.
Explicit expressions for the coefficients of a general Wronskian Laguerre polynomial defined in terms of a single partition are given. 
\end{abstract}

\begin{itemize}
 \item[]{\textbf{Keywords}: \p\ equation, rational solutions, Laguerre polynomials, discriminant, partition, Wronskian.}
\end{itemize}

\noindent
{Dedicated to Athanassios S.\ Fokas on the occasion of his 70th anniversary for his many contributions to studies of integrable nonlinear differential equations, including \p\ equations.}
 
\section{Introduction}
The fifth \p\ equation is given by
\beq\wzz=\left(\frac{1}{2w}+ \frac{1}{w-1}\right)\left(\wz\right)^{\!\!2} -
\frac{1}{z}\wz+ \frac{(w-1)^2(\a w^2+{\b})}{z^2w}+ \frac{\ga w}{z} 
+\frac{\de w(w+1)}{w-1},\label{eq:pv1}\eeq
with $\a$, $\b$, $\ga$ and $\de$ constants. In the generic case of \eqref{eq:pv1} when $\de\not=0$, then we set
$\de=-\tfrac12$, without loss of generality (by rescaling $z$ if necessary) and obtain
\beq\wzz=\left(\frac{1}{2w}+ \frac{1}{w-1}\right)\left(\wz\right)^{\!\!2} -
\frac{1}{z}\wz+ \frac{(w-1)^2(\a w^2+{\b})}{z^2w}+ \frac{\ga w}{z} 
-\frac{w(w+1)}{2(w-1)},\label{eq:pv}\eeq
which we will refer to as \PV. 

The six \peqs\ (\PI--\PVI), 
were discovered by \p, Gambier and their colleagues whilst studying second order ordinary differential equations of the form
\begin{equation} \label{gen-ode}
\deriv[2]{w}{z}=F\left(z,w,\deriv{w}{z}\right), 
\end{equation}
where $F$ is rational in $\d w/\d z$ and $w$ and analytic in $z$. The \p\ transcendents, i.e.\ the solutions of the \peqs, can be thought of as nonlinear analogues of the classical special functions. Iwasaki, Kimura, Shimomura and Yoshida \cite{refIKSY} characterize the six \peqs\ as ``the most important nonlinear ordinary differential equations" and state that ``many specialists believe that during the twenty-first century the \p\ functions will become new members of the community of special functions". Subsequently the \p\ transcendents are a chapter in the NIST \textit{Digital Library of Mathematical Functions}\ \cite[\S32]{refDLMF}.

The general solutions of the \peqs\ are transcendental in the sense that they cannot be expressed in terms of known elementary functions and so require the introduction of a new transcendental function to describe their solution. However, it is well known that all the \peqs, except \PI, possess rational solutions, algebraic solutions and solutions expressed in terms of the classical special functions --- Airy, Bessel, parabolic cylinder, Kummer and hypergeometric functions, respectively --- for special values of the parameters, see, e.g.\ \cite{refPAC06review,refFA82,refGLS}
and the references therein. These hierarchies are usually generated from ``seed solutions'' using the associated \bk\ transformations and frequently can be expressed in the form of determinants. 

\Vor\ \cite{refVor} and Yablonskii \cite{refYab59} expressed the rational solutions of \PII\ in terms of special polynomials, now known as the \textit{Yablonskii--\Vor\ polynomials}, which were defined through a second-order, bilinear differential-difference equation. Subsequently Kajiwara and Ohta \cite{refKO96} derived a determinantal {representation} of the polynomials, see also \cite{refKMi,refKMii}.
Okamoto \cite{refOkamotoP2P4} obtained special polynomials, analogous to the Yablonskii--\Vor\ polynomials, which are associated with some of the rational solutions of \PIV. Noumi and Yamada \cite{refNY99} generalized Okamoto's results and expressed all rational solutions of \PIV\ in terms of special polynomials, now known as the \textit{generalized Hermite polynomials} $H_{m,n}(z)$ and \textit{generalized Okamoto polynomials} $Q_{m,n}(z)$, both of which are determinants of sequences of Hermite polynomials; see also \cite{refKO98}. 

Umemura \cite{RefUm20} derived special polynomials associated with certain rational and algebraic solutions of 
\PIII\ and \PV, 
which are determinants of sequences of associated Laguerre polynomials. 
({The original manuscript was written by Umemura in 1996 for the proceedings of the conference ``\textit{Theory of nonlinear special functions: the \p\ transcendents}" in Montreal, which were not published; see \cite{refOkOh20}.})
Subsequently there have been further studies of rational and algebraic solutions of \PV\ 
\cite{refPACpv,refPAC23,refKLM,refMOK,refNY98ii,refOkamotoPV,refWat}.
Several of these papers are concerned with the combinatorial structure and determinant representation of the generalised Laguerre polynomials, often related to the Hamiltonian structure and affine Weyl symmetries of the \peqs. Additionally the coefficients of these special polynomials have some interesting combinatorial properties \cite{RefUm98,RefUm01,RefUm20}. See also \cite{refNoumi2} and 
results on the combinatorics of the coefficients of Wronskian Hermite polynomials \cite{refBDS} and Wronskian Appell polynomials \cite{refBon}.

We define generalised Laguerre polynomials as Wronskians of 
a sequence of associated Laguerre polynomials specified in terms of a partition of an integer. 
We give a short introduction to the combinatorial concepts in \S\ref{sec:parts} and record several equivalent definitions of a generalised Laguerre polynomial in \S\ref{sec:gen_lag}, where we also show that the polynomials satisfy various differential-difference equations and discrete equations. 
In \S\ref{sec:rats} we express a family of rational solution of \PV\ \eqref{eq:pv} in terms of the generalised Laguerre polynomials. For certain values of the parameter, we show that the solutions are not unique. Rational solutions of the \PV\ $\sigma$-equation, the second-order, second-degree differential equation associated with the Hamiltonian representation of \PV, are considered in \S\ref{sec:sigma}, which includes a discussion of some applications. 
In \S\ref{sec:spv} we describe rational solutions of the symmetric \PV\ system. 
Properties of generalised Laguerre polynomials are established in \S\ref{sec:lag} as well as an explicit description of all partitions with $2$-core of size $k$ and $2$-quotient $(\bfla,\bold{\emptyset})$ for all partitions $\bfla$. Then in \S\ref{sec:roots} we obtain the discriminants of the polynomials, describe the patterns of roots as a function of the parameter and explain how the roots move as the parameter varies. Finally, we show that many of the results in the last section can be expressed in terms of combinatorial properties of the underlying partition. We also obtain explicit expressions for the coefficients of Wronskian Laguerre polynomials that depend on a single partition using the hooks of the partition.

\section{Partitions}
\label{sec:parts}
Partitions will appear throughout this article. We give a brief description of the key ideas. Useful references include \cite{refMacd,refStanley}. 
A partition $\bfla=(\la_1, \la_2,\ldots, \la_r)$ 
is a sequence of non-increasing integers $\la_1 \ge \la_2 \ge \ldots \ge \la_r$. 
We sometimes set $r=\ell(\bfla)$. The partition $\bold {\emptyset}$ represents the unique partition of zero. 
We define $|\bfla|=\lambda_1+\lambda_2 + \dots + \lambda_r$. 
The associated {\it degree vector}
$\bold{h}_{\bfla}=(h_1, h_2,\ldots, h_r)$ is a sequence of distinct integers 
 $h_1> h_2> \ldots> h_r>0$ related to partition elements via 
\eq
\la_j = h_j-r+j,\qquad j=1,2,\ldots, r.
\label{degvec}
\en
We often write $\bold{h}$ rather than $\bold{h}_{\bfla}$. 
Define the Vandermonde determinant
$\Delta(\bold{h})$ as 
\eq
\Delta(\bold{h})=\prod_{1 \le j<k \le r}(h_k-h_j).
\en

Partitions are usefully 
represented as Young diagrams by stacking $r$ rows of boxes of decreasing length 
$\la_j$ for $j=1,2,\dots,r$ on top of each other. 
Reflecting a Young diagram in the main diagonal gives the 
diagram corresponding to the conjugate partition $\bfla^*$. 
Young's lattice is the lattice of all partitions partially ordered by inclusion of the corresponding Young diagrams. That is, $\widetilde{\bfla} \le \bfla$ if $\widetilde{\lambda}_i \le \lambda_i$ for $i=1,2,\dots,\ell(\widetilde{\bfla})$. 
We write $\widetilde \bfla <_j \bfla$ if $|\widetilde \bfla|+j =|\bfla|. $
Let $F_{\bfla}$ denote the number of paths in the Young lattice from $\bfla$ to $\bold{\emptyset}$, and $F_{\bfla / \widetilde{\bfla}}$ the number of paths from $\bfla$ to
$\widetilde{\bfla}$. Explicitly 
\[
F_{\bfla / \widetilde{\bfla}}= (|\bfla| - |\widetilde{\bfla}|)!
\text { det} \, \lf[ \frac{1}{( \lambda_j-\widetilde{\lambda}_k-j+k)!}\ri]_{j,k=1}^{\ell(\bfla)}.
\]
A hook length $h_{jk}$ is assigned 
to box $(j,k)$ in the Young diagram via 
\eq
h_{j,k} = \la_j+{\la}^*_k -j-k+1.
\label{hook_length}
\en
The hook length counts the number of boxes to the right of and below box $(j,k)$ plus one. Thus 
\[
F_{\bfla} = \frac{|\bfla|!} {\prod_{h \in \mathcal{H}_{\bfla}
 } h},
\]
where $\mathcal{H}_{\bfla}$ is the set of all hook lengths. 
The entries of the degree vector $\bold{h}_{\bfla}$ are 
 the hooks in the first column of the Young diagram. Examples of Young diagrams 
 and the corresponding hook lengths are given in Figure~\ref{fig:young1}. 

A partition can be represented as $p+1$ smaller partitions known as the $p$-core $\overline{\bfla}$ and $p$-quotient $(\bfnu_1, \dots, \bfnu_p)$. A partition is a $p$-core partition if it contains no hook lengths of size $p$. Therefore the example partition $(2,1)$ is a $2$-core and $\bfla=(4^2,2,1^3)$ is both a $6$- and $7$-core. 
We only consider $p=2$ here. 
The hooks of size $2$ are vertical or horizontal dominoes.
We note that all $2$-cores are staircase partitions $\overline{\bfla} =(k,k-1,\ldots, 1)$. 

The $2$-core of a partition is found by sequentially removing all hooks of size $2$ from the Young diagram such that at each step the diagram represents a partition. 
 The 
terminating Young diagram defines the $2$-core, which we denote $\overline{\bfla}$. It does not
depend on the order in which the hooks are removed. For example, the partition $(4^2,2,1^3)$ has $2$-core $\overline{\bfla}=(2,1)$. Figure~\ref{fig:young1}(a) {shows} that there are three choices of domino that may be removed at the first step. 
 The $2$-height $\text{ht}(\bfla)$
(or $2$-sign) of partition $\bfla$ is the (unique) number of vertical 
dominoes removed from $\bfla$ to obtain its $2$-core. Equivalently, the $2$-height is the number of vertical 
dominoes in any domino tiling of the Young diagram of 
$\bfla$.

The $2$-quotient records how the dominoes are removed from a partition to obtain its core. James' $p$-abacus \cite{refJKbk} is a useful tool to determine the quotient, and
 provides an alternative visual representation of a partition. 
A $2$-abacus consists of left and right vertical
 runners with bead positions labelled $0,2,4,\ldots$ (left) and $1,3,5,\ldots$ (right) from top to bottom. To represent a partition on the $2$-abacus, 
 place a bead at the points corresponding to each
 element of the degree vector $\bold{h}$. Since a partition can have 
 as many $0$'s as we like, 
we allow an abacus to have any number of 
initial beads and any number of empty beads after the last bead. There are, therefore, an infinite set of abaci associated to each partition, according to the location of the first unoccupied slot. We return to this point below. 
The parts of a partition are read from its abacus by counting the number of empty spaces before each bead. 

A bead with no bead directly above it on the same runner 
corresponds to a hook of length $2$ in the Young diagram. The $2$-core 
$\overline{\bfla}$ is found from the abacus by sliding all beads vertically up as far as possible and reading off the resulting partition. 
Figure~\ref{fig:young1} 
shows the Young diagram and hooklengths of $(4^2,2,1^3) $ in (a), 
 an abacus representation in (c), 
its $2$-core $\overline{\bfla}=(2,1)$ in (b) and the abacus corresponding 
to $\overline{\bfla}$ that is obtained from (c) by pushing up all beads. 
\begin{figure}[ht]
\begin{center}
 \begin{subfigure}[b]{0.24\textwidth}
\centering
\begin{Tableau}{{9,5,3,2},{8,4,2,1},{5,1},{3},{2},{1}}
\end{Tableau}
 \caption{${\bfla}$}
 \label{fig:young1a}
 \end{subfigure}
\hfill
 \begin{subfigure}[b]{0.24\textwidth}
\centering
\begin{Tableau}{{3,1},{ 1}}
\end{Tableau}
 \caption{$\overline{\bfla}$}
 \label{fig:young1b}
 \end{subfigure}
 \hfill
 \begin{subfigure}[b]{0.24\textwidth}
 \centering

\begin{tikzpicture}[scale=0.7]
 \node[vertexe] (n0) at (0,0) {};
 \node[vertexf] (n1) at (0,-1) {};
 \node[vertexe] (n2) at (0,-2) {};
 \node[vertexe] (n3) at (0,-3) {};
 \node[vertexf] (n3) at (0,-4) {};

 \node[vertexf] (n0) at (1,0) {};
 \node[vertexf] (n1) at (1,-1) {};
 \node[vertexf] (n2) at (1,-2) {};
 \node[vertexe] (n3) at (1,-3) {};
 \node[vertexf] (n3) at (1,-4) {};
\end{tikzpicture}

\caption{${\bfla}$}
 \label{fig:abac1}
 \end{subfigure}
 \hfill
 \begin{subfigure}[b]{0.24\textwidth}
 \centering

\begin{tikzpicture}[scale=0.7]
 \node[vertexf] (n0) at (0,0) {};
 \node[vertexf] (n1) at (0,-1) {};
 \node[vertexe] (n2) at (0,-2) {};
 \node[vertexe] (n3) at (0,-3) {};
 \node[vertexe] (n3) at (0,-4) {};

 \node[vertexf] (n0) at (1,0) {};
 \node[vertexf] (n1) at (1,-1) {};
 \node[vertexf] (n2) at (1,-2) {};
 \node[vertexf] (n3) at (1,-3) {};
 \node[vertexe] (n3) at (1,-4) {};
\end{tikzpicture}

 \caption{ $\overline{\bfla}$}
 \label{fig:abac2}
 \end{subfigure}

\end{center}
\caption{The Young diagrams including hook length 
corresponding to (a) $\bfla=(4^2,2,1^3)$ and its core (b)
$\overline{\bfla}=(2,1)$, and corresponding abacus diagrams (c) and (d). 
\label{fig:young1}}
\end{figure}
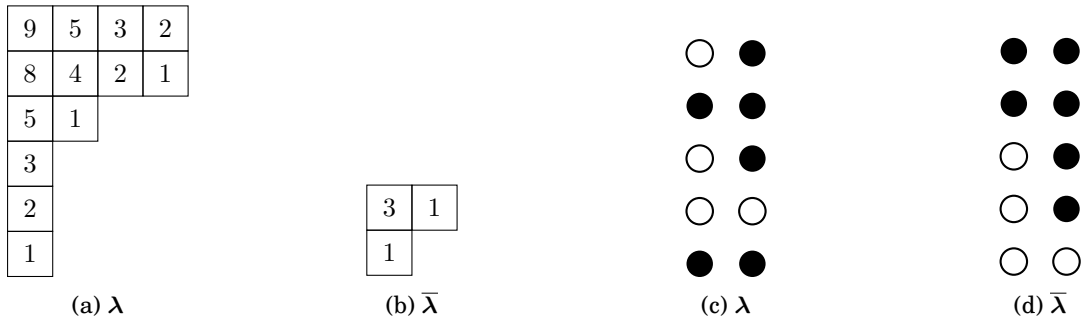

The $2$-quotient is an ordered pair of partitions $(\bfnu_1,\bfnu_2)$ that 
encodes how many places the beads on each runner 
are moved to obtain the $2$-core. 
The $2$-quotient ordering is specified by ensuring the $2$-core has at least as many
beads on the second runner as the first. One can always add a bead to 
the left runner of the partition abacus and shift all subsequent beads one place if this condition is not met \cite{refWildon}, swapping the order of the quotient partitions. Consequently, the relationship between a partition and 
its $2$-core of size $k$ and $2$-quotient $ (\bold{\nu}_1,\bold{\nu}_2)$ is bijective. 
In the running example, one bead on the left runner is moved one place and another bead is moved three places. This is recorded in the partition $\bold{\nu}_1=(3,1)$. 
Only one bead is moved on runner 2, by one space, and so $\bold{\nu}_2=(1)$. 
Therefore the $2$-core and $2$-quotient of ${\bfla}=(4^2,2,1^3)$ are $(2,1)$ 
and $( (3,1), (1) )$ respectively.

While we do not know of an explicit representation of the core and quotient for a generic partition, nor vice versa, the corresponding partitions can easily be found case by case and the bijection is known 
 in some special families of partitions. Partitions with $2$-core $k$ and $2$-quotient $(\bold{\nu},\bold{\emptyset})$ will be important 
in this article. For such partitions, we now determine the (unordered) first column hooks of 
the corresponding partition $\bold{\Lambda}(k, \bold{\nu})$. 
 Find the degree vector $\bold{h}_{\bold{\nu}}$ and 
place beads on the $2$-abacus in positions 
\eq
\lf\{ 2h_i\ri\}_{i=1}^{r}
\cup\{ 2j-1\}_{j=1}^{r+ k}.
\label{hks}
\en
We read off the corresponding partition $\bold{\Lambda}(k, \bold{\nu})$ from 
the position of the beads on the abacus. 
The first column hooks given by \eqref{hks} must be 
 ordered before using \eqref{degvec} to obtain the partition, which is why we cannot give an expression 
for $\bold{\Lambda}(k, \bold{\nu})$ for generic partitions $\bold{\nu}$. As an example take $k=3$ 
and $\bold{\nu}=(4,2,1)$. Then $\bold{h}_{\bold{\nu}}= ( 6,3,1)$. It follows from \eqref{hks} that the abacus of the partition $\bold{\Lambda}(3,(4,2,1)) $ has beads in places
$2,6,12$ and $1,3,5,7,9,11$. Therefore $\bold{h}_{\bold{\Lambda}}=(12,11,9,7,6,5,3,2,1)$ and thus $\bold{\Lambda}(3,(4,2,1)) =(4^2,3,2^3,1^3)$. 
In section~\ref{sec:lag}, we use the first column hook 
set (\ref{hks}) 
to determine 
an explicit formula for the family of partitions with 
$2$-core $k$ and $2$-quotient $( ((m+1)^n),\bold{\emptyset})$.

\section{Generalised Laguerre polynomials}
\label{sec:gen_lag}
\begin{definition}
The \textit{generalised Laguerre polynomial} $\Tmn{m,n}{\mu}(z)$, which is a polynomial of degree $(m+1)n$, is defined by 
\beq \label{def:Tmn}
\Tmn{m,n}{\mu}(z) = \det \lf[\frac{\d^{j+k} }{\d z^{j+k}} L_{m+n}^{(\mu+1)} (z) 
\ri]_{j,k=0}^{n-1}, \qquad m\ge 0, \quad n\ge 1,
\eeq
where $L_{n}^{(\a)}(z)$ is the associated Laguerre polynomial
\beq\label{eq:alp} \Lag{n}{\a}(z)=\frac{z^{-\a}\,\e^z}{n!}\deriv[n]{}{z}\left(z^{n+\a}\,\e^{-z}\right),\qquad n\geq0.\eeq
\end{definition}

\comment{\begin{remark}
 The generalised Laguerre polynomial $\Tmn{m,n}{\mu}(z)$ is a polynomial of degree $(m+1)n$. 
\end{remark}}
\begin{lemma}\label{lem:32}
The generalised Laguerre polynomial $\Tmn{m,n}{\mu}(z)$ 
can also be written as the Wronskian
\begin{align}
\Tmn{m,n}{\mu}(z) &=(-1)^{n(n-1)/2 } 
\Wr\left( L_{m+n}^{(n+\mu)}(z), L_{m+n-1}^{(n+\mu)}(z), 
\ldots, L_{m+1}^{(n+\mu)}(z)\right) \nn \\
&= \Wr\left( L_{m+1}^{(n+\mu)}(z), L_{m+2}^{(n+\mu)}(z), 
\ldots, L_{m+n}^{(n+\mu)}(z) \right). \label{def:Tmn2}
\end{align}
\end{lemma}
\noindent
\comment{I think this proof is wrong. Equation \eqref{eq:35} is incorrect and also leads to singular coefficients. Since
\[ \Lag{k}{\a+\b+1}{(x+y)}=\sum_{j=0}^k \Lag{j}{\b}(x) \Lag{k-j}{\a}(y),\qquad \Lag{k}{\a}(0)=\frac{(\a+1)_k}{k!} = \frac{\Gamma(\a+k+1)}{k!\,\Gamma(\a+1)},\]
then
\begin{align*}
 \Lag{k}{\a+\b+1}{(x)}&=\sum_{j=0}^k \Lag{j}{\b}(x) \Lag{k-j}{\a}(0)=\sum_{j=0}^k \frac{\Gamma(\a+1+k-j)}{(k-j)!\,\Gamma(\a+1)}\,\Lag{j}{\b}(x),
\end{align*}
and so
\[ \Lag{k}{\a}{(x)}=\sum_{j=0}^k \frac{\Gamma(\a-\b+k-j)}{(k-j)!\,\Gamma(\a-\b)}\,\Lag{j}{\b}(x),\]
which has singular coefficients if $\a-\b$ is a negative integer.}

\begin{proof}
We use 
\eq
\frac{\d^k}{\d z^k} L_{n}^{(\alpha)}(z) = \begin{cases}
(-1)^k L_{n-k} ^{(\alpha+k)}(z),\quad & k\le n, \\ 
0, & \text{otherwise},
\end{cases}
\label{lag1}
\en
cf.~\cite[equation (18.9.23)]{refDLMF},
to write the determinant form of $\Tmn{m,n}{\mu}(z)$ as a Wronskian
\eq \det \lf[\frac{\d^{j+k} }{\d z^{j+k}} L_{m+n}^{(\mu+1)} (z) 
\ri]_{j,k=0}^{n-1} 
= (-1)^{n(n-1)/2} \Wr\left( L_{m+n}^{(\mu+1)}(z), 
 L_{m+n-1}^{(\mu+2)}(z), \ldots, 
L_{m+1}^{(\mu+n)}(z) \right). \nn 
\en
\comment{Then we rewrite each term in the Wronskian using 
\eq
L_k^{(\alpha)} (x) = \sum_{j=0}^k \binom{ \alpha-\beta+k-j-1 }{ k-j} 
L_j^{(\beta)}(x).\label{eq:35}
 \en 
to obtain
\begin{align*}
\Tmn{m,n}{\mu}(z) &= (-1)^{n(n-1)/2}\nn\\
&\times\Wr\left( \sum_{j=0}^{m+n} \binom{m-j}{m+n-j} 
L_{j}^{(n+\mu)}(z), 
\sum_{j=0}^{m+n-1} \binom{m-j}{m+n-j-1} L_{j}^{(n+\mu)}(z),\ldots, 
L_{m+1}^{(n+\mu)}(z) \right). 
\end{align*}}%
{Using the result
\eq
\Lag{m}{\a }(z) = \Lag{m}{\a+1}(z)-\Lag{m -1}{\a+1}(z),
\label{eq:lag_rel}
\en
\cite[equation (18.9.13)]{refDLMF}, 
it can be shown using induction that
\[\Lag{m+k}{\a+1 -k }(z) =\Lag{m+k}{\a}(z)+\sum_{j=1}^{k-1} (-1)^{k-j}\binom{k-1}{j-1}\Lag{m+j}{\a}(z).\]
Hence setting $\a=\mu+n$ gives
\beq\Lag{m+k}{\mu+n+1 -k }(z) = \Lag{m+k }{\mu+n }(z)+ \sum_{j=1}^{k-1} (-1)^{k-j}\binom{k-1}{j-1}\Lag{m+j}{\mu+n}(z),\qquad k=1,2,\ldots,n,\eeq
and so we obtain
\begin{align*}
\Tmn{m,n}{\mu}(z) &= (-1)^{n(n-1)/2}\nn\\
&\times\Wr\left( {\Lag{m+n}{n+\mu}(z)+\sum_{j=1}^{n} (-1)^{n-j}\binom{n-1}{j-1}\Lag{m+j}{n+\mu},\ldots,\Lag{m+2}{n+\mu}(z)-\Lag{m+1}{n+\mu}(z), \Lag{m+1}{n+\mu}(z) }\right). 
\end{align*}}%
Since we can add a multiple of any column to any other column without 
changing the Wronskian determinant, we keep the last term in each sum:
\eq
\Tmn{m,n}{\mu}(z) = (-1)^{n(n-1)/2} \Wr\left( L_{m+n}^{(n+\mu)}(z), L_{m+n-1}^{(n+\mu)}(z), 
\ldots, 
L_{m+1}^{(n+\mu)}(z) \right). 
\en
On interchanging the $j^{\rm th}$ column with the $(n-j+1)^{\rm th}$ column, we 
find 
\eq
\Tmn{m,n}{\mu}(z) =
\Wr\left( L_{m+1}^{(n+\mu)}(z), L_{m+2}^{(n+\mu)}(z), 
\ldots, 
L_{m+n}^{(n+\mu)}(z) \right). 
\en
\end{proof}

We remark that 
\[
T_{0,m-1}^{(n-m+1)}(z)=\Wr\left( L_{1}^{(n)}(z), L_{2}^{(n)}(z), 
\ldots, L_{m-1}^{(n)}(z) \right)
= (-1)^{\floor{m/2}} L_{m-1}^{(-m-n )}(-z).
\]

\begin{definition}
Bonneux and Kuiljaars \cite{refBK18}, see also \cite{refDuran14b,refDP15,refGUGM18},
define a \textit{Wronskian of Laguerre polynomials}
\eq
\Omega_{\bfla}^{(\alpha)}(z) = \Wr\left( L_{h_1}^{(\alpha)}(z),
L_{h_2}^{(\alpha)}(z), 
\ldots, 
L_{h_r}^{(\alpha)}(z)\right),
\label{bk_omega}
\en
in terms of the degree vector 
$\bold{h}=(h_1, h_2, \ldots, h_r)$ of partition $\bfla=(\la_1,\la_2,\ldots,\la_{r})$.
Hence 
\eq
\Tmn{m,n}{\mu}(z) = (-1)^{n(n-1)/2} \, \Omega_{\bfla}^{(n+\mu)}(z),
\label{TOmega}
\en
where the partition is $\bfla=((m+1)^n)$.
\end{definition}


\begin{definition}
The \textit{elementary Schur polynomials} $p_j(\bold{t})$, 
for $j\in\ZZ$, in terms of the variables 
$\bold{t} =(t_1, t_2,\ldots)$, are defined by the generating function
\eq
\sum_{j=0}^\infty p_j(\bold{t}) \,x^j =\exp\lf(\sum_{j=1}^\infty t_j \,x^j\ri),\qquad p_j(\bold{t}) =0,\quad \text{for}\quad j<0,
\en
with $p_0(\bold{t}) =1$.
The \textit{Schur polynomial} $S_{\bfla}(\bold{t})$ for the partition $\bfla$ is given by
\eq
S_{\bfla} ( \bold{t} )
={\rm det} \begin{bmatrix}
p_{\la_j+ k-j} (\bold{t}) \end{bmatrix}_{j,k=1}^{r }.
\label{schur}
\en
\end{definition}

The generalised Laguerre polynomial $\Tmn{m,n}{\mu}(z) $ can be expressed as a Schur polynomial, as shown in the following Lemma.
\begin{lemma}\label{schur_defn}
The generalised Laguerre polynomial $\Tmn{m,n}{\mu}(z) $ is the Schur polynomial 
\eq
\Tmn{m,n}{\mu}(z) 
= (-1)^{n(n-1)/2} S_{\bfla}(\bold{t} ),
\label{Tschur}
\en
where 
$
\bfla=((m+1)^n) $
and 
\eq
 t_j = \frac{\mu+n+1}{j} -z,\qquad j=1,2,\ldots\ .
\label{tvals}
\en
\end{lemma}
\begin{proof}
Since 
$$
\frac{\partial^j p_m}{\partial t_1^j} =p_{m-j},
$$
the Schur polynomial (\ref{schur}) can be written as the Wronskian 
\eq
S_{\bfla}(\bold{t}) = \Wr\left(p_{\la_n}, {p_{\la_{n-1}+1}}, 
\ldots, p_{\la_1+n-1} \right),
\label{schur2}
\en
for any partition $\bfla$, 
 where the Wronskian is evaluated with respect to $t_1$. 
 The choice of $t_j$ defined in \eqref{Tschur} leads to 
\eq
p_j(\bold{t}) = L_j^{(\mu+n)}(-z),\qquad j=0,1,\ldots\ .
\en
Set $\bfla=((m+1)^n)$, then \eqref{Tschur} follows from 
(\ref{schur2}) by re-ordering rows and columns and letting $z \to -z$. 
\end{proof}


\begin{definition}
Define the polynomial $\TTmn{m,n}{\mu}(z)$
\eq \label{def:Tnhat}
\TTmn{m,n}{\mu}(z) = \det \lf[\frac{\d^{j+k} }{\d z^{j+k}} L_{m+n}^{(\mu+1)} 
(-z) 
\ri]_{j,k=0}^{n-1}, \quad m\ge 0, \quad n\ge 1,
\en
with $L_n^{(\a)}(z)$ the associated Laguerre polynomial.
\end{definition}
\begin{remark}{\label{rem:37}\rm 
We note that 
\beq\label{eq:TmnT}\Tmn{m,n}{\mu}(-z) = \TTmn{m,n}{\mu}(z).\eeq}\end{remark}

\begin{lemma}
The generalised Laguerre polynomial $\Tmn{m,n}{\mu}(z)$ has the discrete symmetry 
\eq 
\Tmn{m,n}{\mu}(z) = (-1)^{\floor{(m+n+1)/2}} \,\Tmn{n-1, m+1}{-\mu-2n-2m-2}(-z). 
\label{symmT}
\en
\end{lemma}
\begin{proof}
Apply the standard relation
\eq
S_{\bfla}(\bold{t} ) = S_{\bfla^*}(-\bold{t} ).
\en
with $\bfla^*=(n^{m+1})$ 
to the Schur form of the generalised Laguerre polynomial \eqref{schur_defn}. 
\end{proof}

\begin{lemma}
The generalised Laguerre polynomial $\Tmn{m,n}{\mu}(z)$ can also be written as the determinants
\begin{subequations}\begin{align} \label{def:Tmna}
 \Tmn{m,n}{\mu}(z)&=\det\left[\Lag{m+n}{\mu+j+k+1}(z)\right]_{j,k=0}^{n-1},&& m\geq0,\quad n\geq1,\\ \label{def:Tmnc}
 \Tmn{m,n}{\mu}(z)&=\det\left[\Lag{m+n-j-k}{\mu+2n-1}(z)\right]_{j,k=0}^{n-1},&& m\geq0,\quad n\geq1,\\ \label{def:Tmnd}
\Tmn{m,n}{\mu}(z)&=\det\left[\Lag{m+2-n+j+k}{\mu+2n-1}(z)\right]_{j,k=0}^{n-1},&& m\geq0,\quad n\geq1,\\ \label{def:Tmnb}
 \Tmn{m,n}{\mu}(z)&=(-1)^{\lfloor n/2\rfloor}\det\left[\Lag{m+j+1}{\mu+n+k}(z)\right]_{j,k=0}^{n-1},&& m\geq0,\quad n\geq1,\\ \label{def:Tmne}
\Tmn{m,n}{\mu}(z)&=(-1)^{\lfloor n/2\rfloor}\det\left[\Lag{m+1+j-k}{\mu+2n-1}(z)\right]_{j,k=0}^{n-1},&& m\geq0,\quad n\geq1,
\end{align}\end{subequations}
where $\Lag{n}{\a}(z)$ is the Laguerre polynomial with $\Lag{n}{\a}(z)=0$ if $n<0$.
\end{lemma}

\begin{proof}
These identities are easily proved using the well-known formulae 
 (\ref{lag1}) and (\ref{eq:lag_rel}), 
 and properties of Wronskians in either \eqref{def:Tmn} or \eqref{def:Tmn2}.

\end{proof}


\begin{lemma}\label{lem:38} The generalised Laguerre polynomial $\Tmn{m,n}{\mu}(z)$ satisfies the second-order, differential-difference equation
\beq\Tmn{m,n}{\mu}\deriv[2]{\Tmn{m,n}{\mu}}{z}-\left(\deriv{\Tmn{m,n}{\mu}}{z}\right)^{\!2}=\Tmn{m+1,n-1}{\mu}\Tmn{m-1,n+1}{\mu}.\label{eq:Tmn}\eeq
\end{lemma}

\begin{proof}According to Sylvester \cite{refSyl}, see also \cite{refMuir},
if $\mathcal{A}_n(\ph)$ is the double Wronskian given by
\[ \mathcal{A}_n(\ph) = \det\left[\deriv[j+k]{\ph}{z}\right]_{j,k=0}^{n-1}=\Wr\left(\ph,\deriv{\ph}{z},\ldots,\deriv[n-1]{\ph}{z}\right),\]
then $\mathcal{A}_n(\ph)$ satisfies the
\beq\label{eq:toda} \mathcal{A}_n\deriv[2]{\mathcal{A}_n}{z}-\left(\deriv{\mathcal{A}_n}{z}\right)^{\!2}=\mathcal{A}_{n+1}\mathcal{A}_{n-1},\eeq
which is now known as the Toda equation.
From \eqref{def:Tmn}
\[\Tmn{m,n}{\mu}= \det\left[\deriv[j+k]{\Lag{m+n}{\mu}}{z}\right]_{j,k=0}^{n-1}
=\Wr\left(\Lag{m+n}{\mu},\deriv{\Lag{m+n}{\mu}}{z},\ldots,\deriv[n-1]{\Lag{m+n}{\mu}}{z}\right).\]
If we let $\ph=\Lag{m+n}{\mu}$ and $\mathcal{A}_n\left(\Lag{m+n}{\mu}\right)=\Tmn{m,n}{\mu}$, then we need to show that
\[ \mathcal{A}_{n+1}\left(\Lag{m+n}{\mu}\right)=\Tmn{m-1,n+1}{\mu},\qquad\mathcal{A}_{n-1}\left(\Lag{m+n}{\mu}\right)=\Tmn{m+1,n-1}{\mu}.\]
By definition
\begin{align*}
\mathcal{A}_{n+1}\left(\Lag{m+n}{\mu}\right) &= \Wr\left(\Lag{m+n}{\mu},\deriv{\Lag{m+n}{\mu}}{z},\ldots,\deriv[n]{\Lag{m+n}{\mu}}{z}\right)=\Tmn{m-1,n+1}{\mu},\\
\mathcal{A}_{n-1}\left(\Lag{m+n}{\mu}\right)&= \Wr\left(\Lag{m+n}{\mu},\deriv{\Lag{m+n}{\mu}}{z},\ldots,\deriv[n-2]{\Lag{m+n}{\mu}}{z}\right)=\Tmn{m+1,n-1}{\mu},
\end{align*}
which proves the result. 
\end{proof}

\begin{remarks}{\rm
\begin{enumerate}[(i)]
\item[]
\item Lemma \ref{lem:38} can also be proved using the well-known \textit{Jacobi Identity} \cite{refDod}, sometimes known as the \textit{Lewis Carroll formula}, for the determinant $\mathcal{D}$
\beq \mathcal{D}\, \mathcal{D}\bn{i,k}{j,\ell}= \mathcal{D}\bn{i}{j}\mathcal{D}\bn{k}{\ell} - \mathcal{D}\bn{k}{j} \mathcal{D}\bn{i}{\ell}\label{JacobiId}\eeq
where $\mathcal{D}\bn{i}{j} $
is the determinant with the $i^{\rm th}$ row and the $j^{\rm th}$ column removed from $\mathcal{D}$. If 
\[\mathcal{D}=\Tmn{m-1,n+1}{\mu}= \det \lf[\frac{\d^{j+k} }{\d z^{j+k}} L_{m+n}^{(\mu+1)}\ri]_{j,k=0}^{n}
=\Wr\left(\Lag{m+n}{\mu+1},\deriv{\Lag{m+n}{\mu+1}}{z},\ldots,\deriv[n]{\Lag{m+n}{\mu+1}}{z}\right),\]
from \eqref{def:Tmn}, then
\comment{\begin{align*}
\mathcal{D}\bn{n,n+1}{n,n+1} &=\Tmn{m+1,n-1}{\mu},\qquad \mathcal{D}\bn{n}{n}= \deriv[2]{\Tmn{m,n}{\mu}}{z},\qquad
\mathcal{D}\bn{n+1}{n+1}=\Tmn{m,n}{\mu},\qquad \mathcal{D}\bn{n+1}{n}= \mathcal{D}\bn{n}{n+1}=\deriv{\Tmn{m,n}{\mu}}{z},
\end{align*}}%
\begin{align*}
\mathcal{D}\bn{n,n+1}{n,n+1} 
&=\Wr\left(\Lag{m+n}{\mu+1},\deriv{\Lag{m+n}{\mu+1}}{z},\ldots,\deriv[n-2]{\Lag{m+n}{\mu+1}}{z}\right)=
\Tmn{m+1,n-1}{\mu},\\
\mathcal{D}\bn{n+1}{n+1}&=\Wr\left(\Lag{m+n}{\mu+1},\deriv{\Lag{m+n}{\mu+1}}{z},\ldots,\deriv[n-1]{\Lag{m+n}{\mu+1}}{z}\right)
=\Tmn{m,n}{\mu},\\
\mathcal{D}\bn{n}{n+1}&=\mathcal{D}\bn{n+1}{n}=\Wr\left(\Lag{m+n}{\mu+1},\deriv{\Lag{m+n}{\mu+1}}{z},\ldots,\deriv[n-2]{\Lag{m+n}{\mu+1}}{z},\deriv[n]{\Lag{m+n}{\mu+1}}{z}\right)\\ &
=\deriv{}{z} \Wr\left(\Lag{m+n}{\mu+1},\deriv{\Lag{m+n}{\mu+1}}{z},\ldots,\deriv[n-2]{\Lag{m+n}{\mu+1}}{z}\right) = \deriv{\Tmn{m,n}{\mu}}{z},\\
\mathcal{D}\bn{n}{n}&=\deriv{}{z}\Wr\left(\Lag{m+n}{\mu+1},\deriv{\Lag{m+n}{\mu+1}}{z},\ldots,\deriv[n-2]{\Lag{m+n}{\mu+1}}{z},\deriv[n]{\Lag{m+n}{\mu+1}}{z}\right)
= \deriv[2]{\Tmn{m,n}{\mu}}{z},
\end{align*}
and so \eqref{eq:Tmn} follows from the Jacobi Identity \eqref{JacobiId} with $i=k=n$ and $j=\ell=n+1$.

\item We note that the generalised Hermite polynomial 
\[H_{m,n}(z)=\Wr\big(H_m(z),H_{m+1}(z),\ldots,H_{m+n-1}(z) \big),\]
with $H_k(z)$ the Hermite polynomial,
which arises in the description of rational solutions of \PIV, satisfies two second-order, differential-difference equations, see \cite[equation (4.19)]{refNY99}.
\end{enumerate}
}\end{remarks}

The generalised Laguerre polynomial $\Tmn{m,n}{\mu}(z)$ satisfies a number of discrete equations. In the following Lemma we prove two of these using Jacobi's Identity \eqref{JacobiId}.

\begin{lemma}\label{lem:312}The generalised Laguerre polynomial $\Tmn{m,n}{\mu}(z)$ satisfies the equations
\begin{align}\Tmn{m,n+1}{\mu-1}\, \Tmn{m,n-1}{\mu+1}&=\Tmn{m+1,n}{\mu-1}\Tmn{m-1,n}{\mu+1}-\left(\Tmn{m,n}{\mu}\right)^{\!2},\label{iden:Tmn1}\\
 \Tmn{m,n+1}{\mu-1} \,\Tmn{m+1,n-1}{\mu+1}&= {\Tmn{m+1,n}{\mu-1}\,\Tmn{m,n}{\mu+1}}-{\Tmn{m+1,n}{\mu}\,\Tmn{m,n}{\mu}}. \label{iden:Tmn2}
 \end{align}
\end{lemma}
\begin{proof}As the $n+1$-dimensional determinant in \eqref{iden:Tmn1} and \eqref{iden:Tmn2} is the same, then to apply Jacobi's Identity \eqref{JacobiId}, it'll be necessary to use two different representations of $\Tmn{m,n+1}{\mu-1}$.

To prove \eqref{iden:Tmn1}, we use $\Tmn{m,n}{\mu}$ as defined by \eqref{def:Tmn} and so we consider
\[\mathcal{A}=\Tmn{m,n+1}{\mu-1} 
=\Wr\left(\Lag{m+n+1}{\mu},\deriv{\Lag{m+n+1}{\mu}}{z},\ldots,\deriv[n]{\Lag{m+n+1}{\mu}}{z}\right),\]
then
\begin{align*}
\mathcal{A}\bn{1}{1} 
&=\Wr\left(\deriv[2]{\Lag{m+n+1}{\mu}}{z},\deriv[3]{\Lag{m+n+1}{\mu}}{z},\ldots,\deriv[n+1]{\Lag{m+n+1}{\mu}}{z}\right)\\
&=\Wr\left(\Lag{m+n-1}{\mu+2},\deriv{\Lag{m+n-1}{\mu+2}}{z},\ldots,\deriv[n-1]{\Lag{m+n-1}{\mu+2}}{z}\right)=\Tmn{m-1,n}{\mu+1},\\
\mathcal{A}\bn{n+1}{n+1} 
&=\Wr\left(\Lag{m+n+1}{\mu},\deriv{\Lag{m+n+1}{\mu}}{z},\ldots,\deriv[n-1]{\Lag{m+n+1}{\mu}}{z}\right)=\Tmn{m+1,n}{\mu-1},\\
\mathcal{A}\bn{1}{n+1} &=\mathcal{A}\bn{n+1}{1}
=\Wr\left(\deriv{\Lag{m+n+1}{\mu}}{z},\deriv[2]{\Lag{m+n+1}{\mu}}{z},\ldots,\deriv[n]{\Lag{m+n+1}{\mu}}{z}\right)\\
&=(-1)^n\Wr\left(\Lag{m+n}{\mu+1},\deriv{\Lag{m+n}{\mu+1}}{z},\ldots,\deriv[n-1]{\Lag{m+n}{\mu+1}}{z}\right)=(-1)^n\Tmn{m,n}{\mu},\\
\mathcal{A}\bn{1,n+1}{1,n+1} 
&=\Wr\left(\deriv[2]{\Lag{m+n+1}{\mu}}{z},\deriv[3]{\Lag{m+n+1}{\mu}}{z},\ldots,\deriv[n]{\Lag{m+n+1}{\mu}}{z}\right)\\
&=\Wr\left(\Lag{m+n-1}{\mu+2},\deriv{\Lag{m+n-1}{\mu+2}}{z},\ldots,\deriv[n-2]{\Lag{m+n-1}{\mu+2}}{z}\right)=\Tmn{m,n-1}{\mu+1},
\end{align*}

since
\[\deriv{}{z}\Lag{m}{\a}(z)=-\Lag{m-1}{\a+1}(z),\qquad \deriv[2]{}{z}\Lag{m}{\a}(z)=\Lag{m-2}{\a+2}(z).\]
Then using Jacobi's Identity \eqref{JacobiId} with $i=k=1$ and $j=\ell=n+1$,
we obtain \eqref{iden:Tmn1}
as required. 

To prove \eqref{iden:Tmn2}, we use the representation of $\Tmn{m,n}{\mu}$ given by \eqref{def:Tmn2}, so we consider
\[\mathcal{B}=\Tmn{m,n+1}{\mu-1} =\Wr\left( L_{m+1}^{(n+\mu)}, L_{m+2}^{(n+\mu)}, \ldots, L_{m+n}^{(n+\mu)}, L_{m+n+1}^{(n+\mu)} \right),\]
then
\[\begin{split} 
\mathcal{B}\bn{1}{1} &=\Wr\left( \deriv{}{z}L_{m+2}^{(n+\mu)}, \deriv{}{z}L_{m+3}^{(n+\mu)}, \ldots, \deriv{}{z}L_{m+n}^{(n+\mu)}, \deriv{}{z}L_{m+n+1}^{(n+\mu)} \right) \\
&=(-1)^n\Wr\left( L_{m+1}^{(n+\mu+1)} ,L_{m+2}^{(n+\mu+1)}, \ldots, L_{m+n-1}^{(n+\mu+1)},L_{m+n}^{(n+\mu+1)} \right) =(-1)^n \Tmn{m,n}{\mu+1}\\
\mathcal{B}\bn{n+1}{n+1} &
=\Wr\left( L_{m+1}^{(n+\mu)}, L_{m+2}^{(n+\mu)}, \ldots, L_{m+n}^{(n+\mu)}\right) =\Tmn{m,n}{\mu} \\
\mathcal{B}\bn{n+1}{1} &
=\Wr\left( L_{m+2}^{(n+\mu)}, L_{m+3}^{(n+\mu)}, \ldots, L_{m+n}^{(n+\mu)}, L_{m+n+1}^{(n+\mu)} \right)=\Tmn{m+1,n}{\mu}\\
\mathcal{B}\bn{1}{n+1} &=\Wr\left( \deriv{}{z}L_{m+1}^{(n+\mu)}, \deriv{}{z}L_{m+2}^{(n+\mu)}, \ldots, \deriv{}{z}L_{m+n-1}^{(n+\mu)}, \deriv{}{z}L_{m+n}^{(n+\mu)} \right) \\
&=(-1)^n\Wr\left( L_{m}^{(n+\mu+1)} ,L_{m+1}^{(n+\mu+1)}, \ldots, L_{m+n-2}^{(n+\mu+1)},L_{m+n-1}^{(n+\mu+1)} \right) =(-1)^n \Tmn{m-1,n}{\mu+1}\\
\mathcal{B}\bn{1,n+1}{1,n+1} &=\Wr\left( \deriv{}{z}L_{m+2}^{(n+\mu)}, \deriv{}{z}L_{m+3}^{(n+\mu)}, \ldots, \deriv{}{z}L_{m+n}^{(n+\mu)} \right) \\
&=(-1)^{n-1} \Wr\left( L_{m+1}^{(n+\mu+1)} ,L_{m+2}^{(n+\mu+1)}, \ldots, L_{m+n-1}^{(n+\mu+1)}\right) =(-1)^{n-1} \Tmn{m+1,n-1}{\mu+1}
 \end{split}\]
 and so using Jacobi's Identity with $i=k=1$ and $j=\ell=n+1$ gives \eqref{iden:Tmn2} as required.
\comment{\[\begin{split} 
&\mathcal{B}\,\mathcal{B}\bn{1,n+1}{1,n+1}=(-1)^{n-1} \Tmn{m,n+1}{\mu-1} \,\Tmn{m+1,n-1}{\mu+1},\\
&\mathcal{B}\bn{1}{1} \mathcal{B}\bn{n+1}{n+1} -\mathcal{B}\bn{1}{n+1} \mathcal{B}\bn{n+1}{1}
=(-1)^n\left\{ \Tmn{m,n}{\mu+1}\,\Tmn{m,n}{\mu}-\Tmn{m+1,n}{\mu}\, \Tmn{m-1,n}{\mu+1}\right\}.
 \end{split}\]
Hence the identity \eqref{iden:Tmn2} follows from Jacobi's Identity \eqref{JacobiId}.}
\end{proof}

The generalised Laguerre polynomial $\Tmn{m,n}{\mu}(z)$ satisfies a number of Hirota bilinear equations and discrete bilinear equations. 

\begin{lemma}\label{lem:313}The generalised Laguerre polynomial $\Tmn{m,n}{\mu}(z)$ satisfies the Hirota bilinear equations 
\begin{subequations}\label{eq324}\begin{align}
&\D_z\left(\Tmn{m,n-1}{\mu+1}\cdot \Tmn{m,n}{\mu}\right)=\Tmn{m+1,n-1}{\mu}\,\Tmn{m-1,n}{\mu+1},\label{324a}\\
&\D_z\left(\Tmn{m,n-1}{\mu+1}\cdot \Tmn{m+1,n}{\mu-1}\right)=\Tmn{m+1,n-1}{\mu}\,\Tmn{m,n}{\mu},\label{324b}\\
&\D_z\left(\Tmn{m,n-1}{\mu+1}\cdot \Tmn{m,n}{\mu-1}\right)=\Tmn{m+1,n-1}{\mu}\,\Tmn{m-1,n}{\mu},\label{324c}\\
&\D_z\left(\Tmn{m+1,n}{\mu}\cdot \Tmn{m,n}{\mu+1}\right)=\Tmn{m+1,n-1}{\mu+1}\,\Tmn{m,n+1}{\mu},\label{324d}\\
&\D_z\left(\Tmn{m,n}{\mu}\cdot \Tmn{m,n}{\mu+1}\right)=\Tmn{m+1,n-1}{\mu+1}\,\Tmn{m-1,n+1}{\mu},\label{324e}\\
&\D_z\left(\Tmn{m+1,n}{\mu}\cdot \Tmn{m,n}{\mu}\right)=\Tmn{m+1,n-1}{\mu+1}\,\Tmn{m,n+1}{\mu-1},\label{324f}
\end{align}\end{subequations}
where $\D_z$ is the Hirota bilinear operator
\beq {\D_z(f\cdot g)=\deriv{f}{z}g-f\deriv{g}{z}},\label{Hirota}\eeq
and the discrete bilinear equation
\beq
\Tmn{m,n}{\mu}\,\Tmn{m,n-1}{\mu}-\Tmn{m-1,n}{\mu}\,\Tmn{m+1,n-1}{\mu}=\Tmn{m,n}{\mu-1}\,\Tmn{m,n-1}{\mu+1}.\eeq
\end{lemma}

\comment{\noindent\blue{Note: Perhaps we should introduce the notation
\[\W_{n}(\ph) = \det \left[\frac{\d^{j+k}\ph}{\d z^{j+k}}
\right]_{j,k=0}^{n-1} = \Wr\left(\ph,\deriv{\ph}{z},\ldots,\deriv[n-1]{\ph}{z}\right),\]
which simplifies some of the above formulae.}}

\begin{proof}
In \cite[Theorem 3.6]{refVeinDale}, Vein and Dale prove three variants of the Jacobi Identity \eqref{JacobiId}. To prove some to the results in this Lemma, we use, 
\beq\An\bn{1}{1} \Ann\bn{n}{1}- \An\bn{n}{1} \Ann\bn{1}{1}=\Ann\bn{n+1}{1} \Ann\bn{1,n}{1,n+1}, \label{VD36C}\eeq
which is identity (C) in \cite[Theorem 3.6]{refVeinDale} with $r=1$.
For \eqref{324a}, consider the determinants
\[\An =\W_{n}\left(\Lag{m+n+1}{\mu}\right)=\Tmn{m+1,n}{\mu-1},\qquad
\Ann=\W_{n+1}\left(\Lag{m+n+1}{\mu}\right)=\Tmn{m,n+1}{\mu-1},\]
where $\W_{n}(\ph)$ is defined by
\[\W_{n}(\ph) = \det \left[\frac{\d^{j+k}\ph}{\d z^{j+k}}
\right]_{j,k=0}^{n-1} = \Wr\left(\ph,\deriv{\ph}{z},\ldots,\deriv[n-1]{\ph}{z}\right),\]
then
\begin{align*}
\An\bn{1}{1} &=\W_{n-1}\left(\deriv[2]{\Lag{m+n+1}{\mu}}{z}\right)=\W_{n-1}\left(\Lag{m+n-1}{\mu+2}\right)=\Tmn{m,n-1}{\mu+1},\\
\An\bn{n}{1}&=\W_{n-1}\left(\ds\deriv{\Lag{m+n+1}{\mu}}{z}\right) =(-1)^{n-1}\W_{n-1}\left(\Lag{m+n}{\mu+1}\right) =(-1)^{n-1}\,\Tmn{m+1,n-1}{\mu},\\
\Ann\bn{1}{1} &=\W_{n}\left(\deriv[2]{\Lag{m+n+1}{\mu}}{z}\right)=\W_n\left(\Lag{m+n-1}{\mu+2}\right)=\Tmn{m-1,n}{\mu+1},\\
\Ann\bn{n}{1} &=\dz\W_{n}\left(\deriv{\Lag{m+n+1}{\mu}}{z}\right)=(-1)^n\dz\W_{n}\left(\deriv{\Lag{m+n}{\mu+1}}{z}\right)=(-1)^n\dz\Tmn{m,n}{\mu},\\
\Ann\bn{n+1}{1}& =\W_{n}\left(\deriv{\Lag{m+n+1}{\mu}}{z}\right)=(-1)^n\W_{n}\left(\Lag{m+n}{\mu+1}\right)=(-1)^n\Tmn{m,n}{\mu},\\
\Ann\bn{1,n}{1,n+1} &=\dz\W_{n-1}\left(\deriv[2]{\Lag{m+n+1}{\mu}}{z}\right) =\dz\W_{n-1}\left(\Lag{m+n-1}{\mu+2}\right)
=\dz\Tmn{m,n-1}{\mu+1},
\end{align*}
and so
\[\Tmn{m,n-1}{\mu+1}\,\dz\Tmn{m,n}{\mu}+\Tmn{m+1,n-1}{\mu}\,\Tmn{m-1,n}{\mu+1}=\Tmn{m,n}{\mu}\dz\Tmn{m,n-1}{\mu+1},\]
which proves the result.

To prove \eqref{324b}, we use \eqref{VD36C}
with
\[\begin{split} \An& =\Wr\left( L_{m+1}^{(n+\mu-1)}, L_{m+2}^{(n+\mu-1)}, \ldots, L_{m+n}^{(n+\mu-1)} \right)=\Tmn{m,n}{\mu-1},\\
\Ann&=\Wr\left( L_{m+1}^{(n+\mu-1)}, L_{m+2}^{(n+\mu-1)}, \ldots, L_{m+n+1}^{(n+\mu-1)} \right)=\Tmn{m,n+1}{\mu-2} . \end{split}\]
then
\[\begin{split} 
\An\bn{1}{1} &=\Wr\left(\dz L_{m+2}^{(n+\mu-1)}, \dz L_{m+3}^{(n+\mu-1)}, \ldots, \dz L_{m+n}^{(n+\mu-1)} \right)\\
&=(-1)^{n-1}\Wr\left( L_{m+1}^{(n+\mu)}, L_{m+2}^{(n+\mu)}, \ldots, L_{m+n-1}^{(n+\mu)} \right)
=(-1)^{n-1}\,\Tmn{m,n-1}{\mu+1},\\
\An\bn{n}{1} &=\Wr\left( L_{m+2}^{(n+\mu-1)}, L_{m+2}^{(n+\mu-1)}, \ldots, L_{m+n}^{(n+\mu-1)} \right)=\Tmn{m+1,n-1}{\mu},\\
\Ann\bn{n}{1} &=\dz\Wr\left( L_{m+2}^{(n+\mu-1)}, L_{m+3}^{(n+\mu-1)}, \ldots, L_{m+n+1}^{(n+\mu-1)} \right)=\dz\Tmn{m+1,n}{\mu-1},\\
\Ann\bn{1}{1} &=\Wr\left(\dz L_{m+2}^{(n+\mu-1)},\dz L_{m+3}^{(n+\mu-1)}, \ldots, \dz L_{m+n+1}^{(n+\mu-1)} \right)\\
&=(-1)^n \Wr\left(L_{m+1}^{(n+\mu)}, L_{m+2}^{(n+\mu)}, \ldots, L_{m+n}^{(n+\mu)} \right)=(-1)^n\Tmn{m,n}{\mu},\\
\Ann\bn{n+1}{1} &=\Wr\left( L_{m+2}^{(n+\mu-1)}, L_{m+3}^{(n+\mu-1)}, \ldots, L_{m+n+1}^{(n+\mu-1)} \right)=\Tmn{m+1,n}{\mu-1},\\
\Ann\bn{1,n}{1,n+1} &=\Wr\left(\dz L_{m+2}^{(n+\mu-1)}, \dz L_{m+3}^{(n+\mu-1)}, \ldots, \dz L_{m+n}^{(n+\mu-1)} \right),\\
&=(-1)^{n-1}\Wr\left( L_{m+1}^{(n+\mu)}, L_{m+2}^{(n+\mu)}, \ldots, L_{m+n-1}^{(n+\mu)} \right)=(-1)^{n-1}\dz \Tmn{m,n-1}{\mu+1}, 
 \end{split}\]
 and so
 \[\Tmn{m,n-1}{\mu+1}\,\dz\Tmn{m+1,n}{\mu-1} -\Tmn{m+1,n-1}{\mu}\,\Tmn{m,n}{\mu} =\Tmn{m+1,n}{\mu-1} \dz \Tmn{m,n-1}{\mu+1},\]
which proves the result. {The other results  \eqref{324c}--\eqref{324f} are proved in a similar way.}
\end{proof}
\section{Rational solutions of \PV}
\label{sec:rats}
\subsection{Classification of rational solutions of \PV}
Rational solutions of \PV\ \eqref{eq:pv} are classified in the following Theorem. 

\begin{theorem}{\label{thm:PT.RS.KLM}
Equation \eqref{eq:pv} has a rational solution if and only if one of the
following holds:
\begin{enumerate}[(i)]
\item
$\a=\tfrac12 m^2$, $\b=-\tfrac12(m+2n+1+\mu)^2$, $\ga=\mu$, for $m\geq1$;
\item
$\a=\tfrac12(m+\mu)^2$, $\b=-\tfrac12(n+\ep\mu)^2$, $\ga=m+\ep n$, with $\ep=\pm1$, provided that $m\not=0$ or $n\not=0$;
\item
$\a=\tfrac12(m+\tfrac12)^2$, $\b=-\tfrac12(n+\tfrac12)^2$, $\ga=
\mu$, provided that $m\not=0$ or $n\not=0$,
\end{enumerate}
where $m,n\in\Z$ and $\mu$ is an arbitrary constant,
together with the solutions obtained through the symmetries 
\begin{align}
\label{PV:S1}\mathcal{S}_1:\qquad &{w_1}({z};{\a_1}, {\b_1}, {\ga_1},-\tfrac12)
={w(-z;\a,\b,\ga,-\tfrac12)},&&({\a_1}, {\b_1}, {\ga_1},-\tfrac12)=(\a,\b,-\ga,-\tfrac12),\\
\label{PV:S2}\mathcal{S}_2:\qquad &{w_2}({z};{\a_2}, {\b_2}, {\ga_3},-\tfrac12)
=\frac{1}{w(z;\a,\b,\ga,-\tfrac12)},&&({\a_2}, {\b_2}, {\ga_3},-\tfrac12)=(-\b,-\a,-\ga,-\tfrac12),
\end{align}
where $w(z;\a,\b,\ga,-\tfrac12)$ is a solution of \eqref{eq:pv}.
}\end{theorem}

\begin{proof}See Kitaev, Law and McLeod \cite{refKLM}; also \cite[Theorem 40.3]{refGLS}.\end{proof}

\begin{remark}{\rm Kitaev, Law and McLeod \cite[Theorem 1.1]{refKLM} give four cases, though their cases (I) and (II) are related by the symmetry \eqref{PV:S2}.
Kitaev, Law and McLeod \cite{refKLM} also state that $\mu\not\in\Z$ in case (iii), but this does not seem necessary, except for uniqueness as discussed in \S\ref{PVrata_nonuni}. 
}\end{remark}

\comment{\[ \a=\tfrac12a^2,\qquad \b=-\tfrac12b^2,\qquad \ga=-1-2\th.\]
\[\begin{array}{|c@{\quad}c@{\quad}c@{\quad}c|} \hline 
& a & b & \th\\ \hline
& & & \\[-10pt]
\text{(i)} & m & m+2n+1+ \mu &-\tfrac12(1+\mu)\\[2.5pt]
\text{(ii)} & m+\mu & n+\ep \mu & -\tfrac12(1+m+\ep n)\\[2.5pt]
\text{(iii)} & m+\tfrac12 & n+\tfrac12 &-\tfrac12(1+\mu)\\[2.5pt]\hline
\end{array}\]}

Rational solutions in case (i) of Theorem\ \ref{thm:PT.RS.KLM} are expressed in terms of \textit{generalised Laguerre polynomials}, which are written in terms of a determinant of Laguerre polynomials and are our main concern in this manuscript. 

\comment{We remark that solutions of \PV\ \eqref{eq:pv} have the symmetry 
\beq\label{PV:S2}\qquad \widetilde{w}({z};\widetilde{A}, \widetilde{B}, \widetilde{C},-\tfrac12)
=\frac{1}{w(z;\a,\b,\ga,-\tfrac12)},\qquad(\widetilde{A}, \widetilde{B}, \widetilde{C},-\tfrac12)
=(-B,-A,-C,-\tfrac12).\eeq}
Rational solutions in cases (ii) and (iii) of Theorem \ref{thm:PT.RS.KLM} are expressed in terms of \textit{generalised Umemura polynomials}.
As mentioned above, Umemura \cite{RefUm20} defined some polynomials through a differential-difference equation to describe rational solutions of \PV\ \eqref{eq:pv}; see also
\cite{refPACpv,refNY98ii,refYamada}.
Subsequently these were generalised by Masuda, Ohta and Kajiwara \cite{refMOK}, who defined the {generalised Umemura polynomial} $U_{m,n}^{(\a)}(z)$ through a coupled differential-difference equations and also gave a representation as a determinant. Our study of the generalised Umemura polynomials is currently under investigation and we do not pursue this further here.

Rational solutions in case (i) of Theorem \ref{thm:PT.RS.KLM} are special cases of the solutions of \PV\ \eqref{eq:pv} expressible in terms of Kummer functions $\KummerM(a,b,z)$ and $\KummerU(a,b,z)$, or equivalently 
the confluent hypergeometric function ${}_1F_1(a;c;z)$. 
Specifically 
\beq U\left(-n,\a+1,z\right)=(-1)^{n}{\left(\a+1\right)_{n}}M\left(-n,\a+1,z\right)=(-1)^{n}n!L^{(\a)}_{n}(z),\label{Kummer:Lag}\eeq
with $\Lag{n}{\a} (z)$ the associated Laguerre polynomial, cf.~\cite[equation (13.6.19)]{refDLMF}.

Determinantal representations of these rational solutions are given in the following Theorem.

\begin{theorem}{\label{thm:PVi}Define the polynomial $\tau_{m,n}^{(\mu)}(z)$
\beq \label{def:taumn}
\tau_{m,n}^{(\mu)}(z)=\det\left[\left(z\deriv{}{z}\right)^{j+k} \Lag{m+n}{n+\mu}(z)\right]_{j,k=0}^{n-1},
\eeq
with $L_n^{(\alpha)}(z)$ the associated Laguerre polynomial \eqref{eq:alp},
then
\begin{subequations}\begin{align}
{w}_{m,n}(z;\mu)&=\left(\frac{m+\mu+ 2n}{m+\mu+ 2n+1}\right)^{\!n}
\frac{\tau_{m-1,n}^{(\mu)}(z)\,\tau_{m-1,n+1}^{(\mu)}(z)}{\tau_{m,n}^{(\mu)}(z)\,\tau_{m-2,n+1}^{(\mu)}(z)},\qquad m, n\geq1,
\end{align}\end{subequations}
is a rational solution of \PV\ \eqref{eq:pv} for the parameters
\begin{subequations}\begin{align}
\pmn{\tfrac12m^2}{-\tfrac12(m+ 2n+1+\mu)^2}{\mu}
\end{align}\end{subequations}}
\end{theorem}
\begin{proof}This result can be derived from the determinantal representation of the special function solutions of \PV\ \eqref{eq:pv} given by Masuda \cite[Theorem 2.2]{refMasuda04}
.\end{proof}
\begin{remark}{\rm The polynomial $\tau_{m,n}^{(\mu)}(z)$ 
has degree $\tfrac12(2m+n+1)n$.
}\end{remark}
\begin{lemma}\label{lem:54}
The polynomials $\tau_{m,n}^{(\mu)}(z)$ and $\Tmn{m,n}{\mu}(z)$ are related as follows
\[ \tau_{m,n}^{(\mu)}(z)=a_{m,n}z^{n(n-1)/2}\Tmn{m,n}{\mu}(z),\qquad a_{m,n}=\prod^{n}_{j =1}(m+n+j+\mu)^{j -1}.\]
\end{lemma}
\begin{proof}
From \eqref{def:taumn}, by definition
\[
\tau_{m,n}^{(\mu)}(z)=\det\left[\left(z\deriv{}{z}\right)^{\!(j+k)} \Lag{m+n}{n+\mu}(z) \right]_{j,k=0}^{n-1}. 
\]
Now we use the identity
\begin{subequations}\label{iden1} 
\beq \det \left[\left(z\deriv{}{z}\right)^{\!j} f_k(z)\right]_{j,k=0}^{n-1}= z^{n(n-1)/2}\Wr\big(f_0(z),f_1(z),\ldots,f_{n-1}(z)\big), \eeq
with
\beq f_0(z)=\Lag{m+n}{n+\mu}(z),\qquad
 f_k(z)=\left(z\deriv{}{z}\right)^{\!k}\Lag{m+n}{n+\mu}(z),\quad k=1,2,\ldots,n-1.\eeq \end{subequations}
Using the recurrence relation
\[z\deriv{}{z}\Lag{n}{\a}(z)=n\Lag{n}{\a}(z)-(n+\mu)\Lag{n-1}{\a}(z),\]
cf.~\cite[equations (18.9.14), (18.9.23)]{refDLMF},
it is straightforward to show by induction that
\beq\label{iden2}\left(z\deriv{}{z}\right)^{\!k}\Lag{n}{\a}(z)=
\sum_{j=0}^{k-1}b_{j,k}^{(n,\mu)} \Lag{n-j}{\a}(z)+(-1)^k b_{k,k}^{(n,\mu)}\Lag{n-k}{\a}(z),\eeq
where $b_{j,k}^{(n,\mu)}$, $j=0,1,\ldots,k$, are constants, with
\beq\label{def:bkk} b_{k,k}^{(n,\mu)}=\prod_{j=0}^{k-1}(n-j+\mu).\eeq (It is not necessary to know what 
the constants $b_{j,k}^{(n,\mu)}$, $j=0,1,\ldots,k-1$ are.) Therefore, using \eqref{iden1} and \eqref{iden2}, we have
\begin{align*}
\tau_{m,n}^{(\mu)}(z) &= z^{n(n-1)/2}\Wr\left(\Lag{m+n}{n+\mu}(z), z\deriv{}{z}\Lag{m+n}{n+\mu}(z),\ldots,\left(z\deriv{}{z}\right)^{n-1}\Lag{m+n}{n+\mu}(z)\right)\\
&= z^{n(n-1)/2} \Wr\left( L_{m+n}^{(n+\mu)}(z), -(m+2n+\mu)L_{m+n-1}^{(n+\mu)}(z), \ldots, (-1)^{(n-1)} b_{n-1,n-1}^{(m+n,n+\mu)}L_{m+1}^{(n+\mu)}(z)\right),\end{align*}
since, as in the proof of Lemma \ref{lem:32}, we need only keep the last term due to properties of Wronskians. Consequently from \eqref{def:Tmn2} we have
\begin{align*}
\tau_{m,n}^{(\mu)}(z) &= z^{n(n-1)/2} \left(\prod_{k=0}^{n-1}b_{k,k}^{(m+n,n+\mu)}\right)\Wr\left( L_{m+1}^{(n+\mu)}(z), L_{m+2}^{(n+\mu)}(z), \ldots, L_{m+n}^{(n+\mu)}(z)\right)\\
&= a_{m,n} z^{n(n-1)/2} \,\Tmn{m,n}{\mu}(z),\end{align*}
where, using \eqref{def:bkk}
\begin{align*} a_{m,n}&=\prod_{k=1}^{n-1}b_{k,k}^{(m+n,n+\mu)}
=\prod_{k=1}^{n-1}\prod_{j=0}^{k-1}(m+2n-j+\mu) 
=\prod^{n}_{j =1}(m+n+j+\mu)^{j -1},\end{align*}
as required.
\end{proof}

\begin{theorem}{\label{thm:PVia}Given the generalised Laguerre polynomial $\Tmn{m,n}{\mu}(z)$ given by \eqref{def:Tmn}, then \begin{subequations}\label{sol:rat1b}\begin{align}
{w}_{m,n}(z;\mu)=\frac{T_{m-1,n}^{(\mu)}(z)\,T_{m-1,n+1}^{(\mu)}(z)}{\Tmn{m,n}{\mu}(z)\,T_{m-2,n+1}^{(\mu)}(z)},\qquad m, n\geq1,
\end{align}
is a rational solution of \PV\ \eqref{eq:pv} for the parameters
\begin{align}
\pmn{\tfrac12m^2}{-\tfrac12(m+2n+1+\mu)^2}{\mu}.
\end{align}\end{subequations}}
In the case when $n=0$ then 
\begin{subequations}\label{sol:rat1b0}
\beq
{w}_{m,0}(z;\mu)=\frac{T_{m-1,1}^{(\mu)}(z)}{T_{m-2,1}^{(\mu)}(z)} = \frac{\Lag{m}{\mu+1}(z)}{\Lag{m-1}{\mu+1}(z)},\qquad m\geq1,\eeq
is a rational solution of \PV\ \eqref{eq:pv} for the parameters
\beq\pmm{\tfrac12m^2}{-\tfrac12(m+1+\mu)^2}{\mu}.\eeq
\end{subequations}
\end{theorem}
\begin{proof}The result follows from Theorem \ref{thm:PVi} and Lemma \ref{lem:54}.
\end{proof}

\begin{corollary}\label{coro:57}{The rational solutions related through the symmetry $\mathcal{S}_1$ \eqref{PV:S1} are given by
\begin{subequations}\label{sol:rat1hat}
\beq\widehat{w}_{m,n}(z;\mu)=\frac{\widehat{T}_{m-1,n}^{(\mu)}(z)\,\widehat{T}_{m-1,n+1}^{(\mu)}(z)}{\TTmn{m,n}{\mu}(z)\,\widehat{T}_{m-2,n+1}^{(\mu)}(z)},\qquad m, n\geq1,\eeq
with $\TTmn{m,n}{\mu}(z)$ the polynomial given by \eqref{def:Tnhat},
which is a rational solution of \PV\ \eqref{eq:pv} for the parameters
\beq
\pmn{\tfrac12m^2}{-\tfrac12(m+2n+1+\mu)^2}{-\mu}.
\eeq \end{subequations}
In the case when $n=0$ then 
\begin{subequations}\label{sol:rat1hat0}
\beq\widehat{w}_{m,0}(z;\mu)=\frac{\widehat{T}_{m-1,1}^{(\mu)}(z)}{\widehat{T}_{m-2,1}^{(\mu)}(z)} = \frac{\Lag{m}{\mu+1}(-z)}{\Lag{m-1}{\mu+1}(-z)},\qquad m\geq1,\eeq
is a rational solution of \PV\ \eqref{eq:pv} for the parameters
\beq\pmm{\tfrac12m^2}{-\tfrac12(m+1+\mu)^2}{-\mu}.\eeq\end{subequations}
}\end{corollary}
\begin{proof}Since $\Tmn{m,n}{\mu}(-z)=\TTmn{m,n}{\mu}(z)$, recall \eqref{eq:TmnT}, then $w_{m,n}(-z;\mu)=\widehat{w}_{m,n}(z;\mu)$ and so the result follows immediately.
\end{proof}

It is known that rational solutions of \PIII\ can be expressed either in terms of four special polynomials or in terms of the logarithmic derivative of the ratio of two special polynomials \cite[Theorem 2.4]{refPACpiii}. Hence it might be expected that the rational solutions of \PV\ discussed here can also be written in terms of the logarithmic derivative of the ratio of two generalised Laguerre polynomials.

\begin{remark}{\rm Using computer algebra we have verified for several small values of $m$ and $n$ that alternative forms of the rational solutions \eqref{sol:rat1b} and \eqref{sol:rat1hat} are given by
\begin{align}\label{altsol1} w_{m,n}(z;\mu)&= \frac{z}{m} \deriv{}{z}\left\{\ln\frac{\Tmn{m-2,n+1}{\mu}(z)}{\Tmn{m,n}{\mu}(z)}\right\}-\frac{z-m-2n-1-\mu}{m} ,\\ \label{altsol2}
\widehat{w}_{m,n}(z;\mu)&= \frac{z}{m} \deriv{}{z}\left\{\ln\frac{\TTmn{m-2,n+1}{\mu}(z)}{\TTmn{m,n}{\mu}(z)}\right\}+\frac{z+m+2n+1+\mu}{m},
\end{align}
respectively.
Consequently, by comparing the solutions we {expect} the relations
\begin{subequations}\label{eq416}\begin{align} &z\D_z\left({\Tmn{m-1,n+1}{\mu}\cdot\Tmn{m+1,n}{\mu}}\right) =(z-m-2n-2-\mu)\Tmn{m-1,n+1}{\mu}\,\Tmn{m+1,n}{\mu}
+(m+1)\Tmn{m,n}{\mu}\Tmn{m,n+1}{\mu},\\
&z\D_z\left({\TTmn{m-1,n+1}{\mu}\cdot\TTmn{m+1,n}{\mu}}\right) =-(z+m+2n+2+\mu)\TTmn{m-1,n+1}{\mu}\,\TTmn{m+1,n}{\mu}
+(m+1)\TTmn{m,n}{\mu}\TTmn{m,n+1}{\mu},
\end{align}\end{subequations}
where $\D_z$ is the Hirota bilinear operator \eqref{Hirota}. {We envisage that the relations \eqref{eq416} can be proved using the Jacobi identity \eqref{JacobiId} or a variant thereof, 
though we don't pursue this further here.}
}\end{remark}

Setting $n=0$ in \eqref{altsol1} gives
\begin{align*}w_{m,0}(z;\mu)&= \frac{z}{m} \deriv{}{z}\left\{\ln\Tmn{m-2,1}{\mu}(z)\right\}-\frac{z-m-1-\mu}{m}\\&=\frac{z}{m} \deriv{}{z}\ln \left\{ L_{m-1}^{(\mu+1)}(z)\right\} -\frac{z-m-1-\mu}{m}
=\frac{\Lag{m}{\mu+1}(z)}{\Lag{m-1}{\mu+1}(z)},
\end{align*}
which is \eqref{sol:rat1b0},
since
\[z\deriv{}{z}L_{m-1}^{(\mu+1)}(z)=(m-1)L_{m-1}^{(\mu+1)}(z)-(m+\mu)L_{m-2}^{(\mu+1)}(z).\]
The solutions \eqref{sol:rat1hat0} and \eqref{altsol2} in the case when $n=0$ can be shown to be the same in a similar way.

\begin{remark}
From Theorem~\ref{thm:PVia} we note that 
$
w_{m,n}(z;-m-n-j) $ and $ w_{m,j-1}(z;-m-n-j)
$
are both rational solutions for 
\[ \alpha_{m,n} = \fract{1}{2} m^2 ,\quad \beta_{m,n} = -\hf (n+1-j)^2 , \quad \gam_{m,n}=-m-n-j,\quad j=1,\dots,n. \]
The equality of the solutions follows 
from lemma~\ref{lemma:multzeros} and the definition of $w_{m,n}(z;\mu)$ in the form (\ref{altsol1}). We add that 
\[
m \,w_{m,n}(z;-m-n) =-(n+1) \widehat w_{n+1,0}(z;-m-n-2).
\]
\end{remark}

\subsection{\label{PVrata_nonuni}Non-uniqueness of rational solutions of \PV}
Kitaev, Law and McLeod \cite[Theorem 1.2]{refKLM} state that rational solutions of \PV\ \eqref{eq:pv} are unique when the parameter $\mu\not\in\Z$. 
In the following Lemma we illustrate that when $\mu\in\Z$ then non-uniqueness of rational solutions of \PV\ \eqref{eq:pv} can occur, that is for certain parameter values there is more than one rational function. 
\begin{lemma}
Consider the rational solutions of \PV\ \eqref{eq:pv} given by 
\begin{align}\label{sol:rat1a}
{w}_{m,n}(z;\mu)&=\frac{\Tmn{m-1,n}{\mu}(z)\,\Tmn{m-1,n+1}{\mu}(z)}{\Tmn{m,n}{\mu}(z)\,\Tmn{m-2,n+1}{\mu}(z)},\qquad 
\widehat{w}_{m,n}(z;\mu)=\frac{\TTmn{m-1,n}{\mu}(z)\,\TTmn{m-1,n+1}{\mu}(z)}{\TTmn{m,n}{\mu}(z)\,\TTmn{m-2,n+1}{\mu}(z)}.
\end{align}
If $\mu\in\Z$ and $\mu\ge -n$ then there are two distinct rational solutions of \PV\ \eqref{eq:pv} for the \textit{same} parameters.
\end{lemma}
\begin{proof}
If $\mu=k$, with $k\in\Z$ and $k\geq -n$, then from Theorem \ref{thm:PVia} and Corollary \ref{coro:57},
${w}_{m,n}(z;k)$ and $\widehat{w}_{m,n+k}(z;-k)$ both satisfy \PV\ \eqref{eq:pv} for the parameters
\[\pms{\tfrac12m^2}{-\tfrac12(m+2n+k+1)^2}{k}.\]
\end{proof}

\begin{example}{\rm
The rational functions
\begin{align*}
w_{1,1}(z;1)&=-\frac{(z-3)(z^{2}-8z+20)}{(z-2)(z-6)},\qquad
\widehat{w}_{1,2}(z;-1)=\frac{(z^{2}+4 z+6)(z^{3}+9 z^{2}+36 z+60)}{z^{4}+12 z^{3}+54 z^{2}+96 z+72}, 
\end{align*}
are both solutions of \PV\ \eqref{eq:pv} with parameters
\[ \pms{\ifrac12}{-\ifrac{25}{2}}{1}.\]
Also the rational functions
\begin{align*}
w_{1,2}(z;-1)&=-\frac{(z^{2}-4 z+6)(z^{3}+9 z^{2}-36 z+60)}{z^{4}-12 z^{3}+54 z^{2}-96 z+72},\qquad
\widehat{w}_{1,1}(z;1)=\frac{(z+3)(z^{2}+8z+20)}{(z+2)(z+6)},
\end{align*}
are both solutions of \PV\ \eqref{eq:pv} with parameters
\[ \pms{\ifrac12}{-\ifrac{25}{2}}{-1}.\]
We note that
\[ w_{1,1}(-z;1)=\widehat{w}_{1,1}(z;-1),\qquad {w_{1,2}(-z;-1)=\widehat{w}_{1,2}(z;1)}.\]
The solutions $w_{1,1}(z;1)$ and $\widehat{w}_{1,2}(z;-1)$ have different expansions about both $z=0$ and $z=\infty$, which are singular points of \PV. As $z\to0$
\begin{align*}
w_{1,1}(z;1)&=5-\frac{1}{3} z+\frac{5}{18} z^{2}+\frac{7}{54} z^{3}+\frac{41}{648} z^{4}+\frac{61}{1944} z^{5}+\O(z^6),\\
\widehat{w}_{1,2}(z;-1)&=5-\frac{1}{3} z+\frac{5}{18} z^{2}+\frac{7}{54} z^{3}-\frac{139}{648} z^{4}+\frac{313}{1944} z^{5}+\O(z^6),
\end{align*}
and as $z\to\infty$
\begin{align*}
w_{1,1}(z;1)&=-z+3-\frac{8}{z}-\frac{40}{z^{2}}-\frac{224}{z^{3}}-\frac{1312}{z^{4}}-\frac{7808}{z^{5}}+\O(z^{-6}),\\
\widehat{w}_{1,2}(z;-1)&=z+1+\frac{12}{z}-\frac{36}{z^{2}}+\frac{72}{z^{3}}+\frac{216}{z^{4}}-\frac{3888}{z^{5}}+\O(z^{-6}).
\end{align*}
}\end{example}

{\begin{remark}{\rm
Recently Aratyn \etal\ \cite{refAGLZ2,refAGLZ3} also discuss non-uniqueness of solutions of \PV\ \eqref{eq:pv}.
}\end{remark}}

\section{Rational solutions of the \PV\ $\sigma$-equation}
\label{sec:sigma}
\subsection{Hamiltonian structure}
Each of the \peqs\ \PI--\PVI\ can be written as a (non-autonomous) Hamiltonian system
\beq\label{eq:HamV}
z\deriv{q}{z}=\pderiv{\mathcal{H}_{\rm J}}{p},\qquad 
z\deriv{p}{z}=-\pderiv{\mathcal{H}_{\rm J}}{q},\qquad {\rm J} = {\rm I, II, \ldots, VI},
\eeq
for a suitable Hamiltonian function $\mathcal{H}_{\rm J}=\mathcal{H}_{\rm J}(q,p,z)$.
Further, 
there is a second-order, second-degree equation, often called the \textit{\p\ $\sigma$-equation} or \textit{Jimbo-Miwa-Okamoto equation}, whose solution is expressible in terms of the solution of the associated \peq\ \cite{refJM81,refOkamoto80a}. 

For \PV\ \eqref{eq:pv} the Hamiltonian is
\beq\label{HJM:47}
z\HV(q,p,z) = q(q-1)^2p^2 -\left\{\nu_1(q-1)^2-(\nu_1-\nu_2-\nu_3)q(q-1)+zq\right\}p+\nu_2\nu_3q,\eeq
with $\nu_1$, $\nu_2$ and $\nu_3$ parameters \cite{refJM81,refOkamoto80a,refOkamotoPV}. 
Substituting \eqref{HJM:47} into \eqref{eq:HamV} gives
\begin{subequations}\label{HJM:4950}
\begin{align}\label{HJM:49}
z\deriv{q}{z}&=2q(q-1)^2p-\nu_1(q-1)^2+(\nu_1-\nu_2-\nu_3)q(q-1)-zq,\\
\label{HJM:50}
z\deriv{p}{z}&=-(3q-1)(q-1)p^2-2(\nu_2+\nu_3)qp+(z-\nu_1-\nu_2-\nu_3)p-\nu_2\nu_3.
\end{align}\end{subequations}
Eliminating $p$ then $q=w$ satisfies \PV\ \eqref{eq:pv}
with \[ 
\pms{\tfrac12(\nu_2-\nu_3)^2}{-\tfrac12\nu_1^2}{\nu_1-\nu_2-\nu_3-1}.\]
The function $\sigma(z)=z\HV(q,p,z)$ defined by \eqref{HJM:47} satisfies the second-order, second-degree equation
\beq\label{eq:S5JM}
\left(z\deriv[2]{\sigma}{z}\right)^{\!\!2}= \left[ 2\left(\deriv{\sigma}{z}\right)^{\!\!2}+(\nu_1+\nu_2+\nu_3-z)\deriv{\sigma}{z}+\sigma\right]^2
-4\deriv{\sigma}{z} 
\prod_{j=1}^3\left(\deriv{\sigma}{z}+\nu_{j}\right),\eeq
cf.~\cite[equation (C.45)]{refJM81}; the \PV\ $\sigma$-equation derived by Okamoto \cite{refOkamoto80a,refOkamotoPV} is equation \eqref{eq:S5Ok} below. 
Conversely, if $\sigma(z)$ is a solution of equation \eqref{eq:S5JM}, then the solutions of equation \eqref{HJM:4950} are
\begin{align*}
q(z)&=\frac{z\sigma''+2(\sigma')^2+(\nu_1+\nu_2+\nu_3-z)\sigma'+\sigma}{2(\sigma'+\nu_2)(\sigma'+\nu_3)},\\
p(z)&=\frac{z\sigma''-2(\sigma')^2-(\nu_1+\nu_2+\nu_3-z)\sigma'-\sigma}{2(\sigma'+\nu_1)}.
\end{align*}
Henceforth we shall refer to equation \eqref{eq:S5JM} as the \sPV\ equation.

The \PV\ $\sigma$-equation derived by Okamoto \cite{refOkamoto80a,refOkamotoPV} is
\beq\label{eq:S5Ok}
\left(z\deriv[2]{h}{z}\right)^{\!\!2}= \left[ 2\left(\deriv{h}{z}\right)^{\!\!2}-z\deriv{h}{z}+h\right]^2
-4 \prod_{j=0}^3\left(\deriv{h}{z}+\k_{j}\right),
\eeq
with $\k_0$, $\k_1$, $\k_2$ and $\k_3$ parameters such that
$\k_0+\k_1+\k_2+\k_3=0$.
Equation \eqref{eq:S5Ok} is equivalent to \sPV\ \eqref{eq:S5JM}, since these
are related by the transformation
\begin{subequations}\beq 
\sigma(z;\bfnu) = h(z;\bold{\k})+ \k_0z+2\k_0^2,\qquad \nu_j=\k_j-\k_0,\quad j=1,2,3,\eeq
where $\bfnu=(\nu_1,\nu_2,\nu_3)$ and $\bold{\k}=(\k_0,\k_1,\k_2,\k_3)$, with
\beq \k_0=-(\k_1+\k_2+\k_3)=-\tfrac14(\nu_1+\nu_2+\nu_3),
\eeq\end{subequations}
as is easily verified.

There is a simple symmetry for solutions of \sPV\ \eqref{eq:S5JM} given in the following Lemma.
\begin{lemma}\label{lemma61}
Making the transformation 
\begin{subequations}\label{SV:BT}\beq \sigma(z;\bfnu) = \widetilde{\sigma}(z;\bfla)-\nu_1 z+(\nu_2+\nu_3-\nu_1)\nu_1,\eeq
with 
\beq \bfla=(\la_1,\la_2,\la_3)=(-\nu_1,\nu_2+\nu_1,\nu_3+\nu_1),\eeq
\end{subequations}
in \sPV\ \eqref{eq:S5JM} yields
\[\left(z\deriv[2]{\widetilde{\sigma}}{z}\right)^{\!\!2}= \left[ 2\left(\deriv{\widetilde{\sigma}}{z}\right)^{\!\!2}+(\la_1+\la_2+\la_3-z)\deriv{\widetilde{\sigma}}{z}+\widetilde{\sigma}\right]^2
-4\deriv{\widetilde{\sigma}}{z} 
\prod_{j=1}^3\left(\deriv{\widetilde{\sigma}}{z}+\la_{j}\right).\]
\end{lemma}
\begin{proof}
This is easily verified by substituting \eqref{SV:BT} in \eqref{eq:S5JM}.
\end{proof}

\subsection{Classification of rational solutions of \sPV}
There are two classes of rational solutions of \sPV\ 
\eqref{eq:S5JM}, 
one expressed in terms of the generalised Laguerre polynomial $\Tmn{m,n}{\mu}(z)$, 
which we discuss in the following theorem, and a second in terms of the generalised Umemura polynomial $U_{m,n}^{(\a)}(z)$. 

\begin{theorem}\label{SVrats}
The rational solution of \sPV\ \eqref{eq:S5JM} in terms of the generalised Laguerre polynomial $\Tmn{m,n}{\mu}$ 
is
\beq\label{Tmn:S5JMa}
\sigma_{m,n}(z;\bfnu)=z\deriv{}{z}\ln \left\{\Tmn{m,n}{\mu}(z)\right\} -(m+1)n,\qquad m\geq0,\quad n\geq1,
\eeq for the parameters
\beq\label{Tmn:S5JMb} \bfnu=(m+1,-n,m+n+\mu+1).
\eeq
\comment{\begin{subequations}\label{Tmn:S5JMa}\beq 
\sigma_{m,n}(z;\bfnu)=z\deriv{}{z}\ln \left\{\Tmn{m,n}{\mu}\right\}+(\mu+m+n+1)z+(\mu+m+2n+1)(\mu+n),
\eeq for the parameters
\beq \bfnu=(-\mu-n,-\mu-m-n-1,-\mu-m-2n-1).
\eeq\end{subequations}}
\comment{The rational solution of \eqref{eq:S5Ok} in terms of the generalised Laguerre polynomial $\Tmn{m,n}{\mu}(z)$ 
is
\begin{subequations}\label{Tmn:S5Ok}\beq 
h_{m,n}(z;\mu)=z\deriv{}{z}\ln \left\{\Tmn{m,n}{\mu}(z)\right\}+\tfrac14(\mu+2m+2)z
-\tfrac18(\mu+2m+2)^2-(m+1)n,
\eeq for the parameters
\beq \begin{array}{l@{\qquad}l}
\k_{0} = \tfrac{3}{4} \mu+\tfrac12m+n+\tfrac12, & \k_{1} =-\tfrac14\mu+\tfrac12m+\tfrac12,\\[5pt] 
\k_2=-\tfrac14\mu-\tfrac12m-n -\tfrac12, & \k_{3} =-\tfrac14\mu-\tfrac12m-\tfrac12.
\end{array}
\eeq\end{subequations}}
\end{theorem}
\begin{proof}
This result can be inferred from the work of Forrester and Witte \cite{refFW02} and Okamoto \cite{refOkamotoPV} on special function solutions of \sPV, together with the relationship between Kummer functions and associated Laguerre polynomials \eqref{Kummer:Lag}. We have used Lemma \ref{lemma61} as a normalisation.
\end{proof}

\begin{corollary}
The rational solution of \sPV\ \eqref{eq:S5JM} in terms of the generalised Laguerre polynomial $\TTmn{m,n}{\mu}(z)$ 
is
\beq\label{TTmn:S5JMa}
\widehat{\sigma}_{m,n}(z;\bfnu)=z\deriv{}{z}\ln \left\{\TTmn{m,n}{\mu}(z)\right\} -(m+1)n,\qquad m\geq0,\quad n\geq1,
\eeq for the parameters
\beq\label{TTmn:S5JMb} \bfnu=(-m-1,n,-m-n-\mu-1).
\eeq
\comment{\begin{subequations}\label{Tmn:S5JMa}\beq 
\widehat{\sigma}_{m,n}(z;\bfnu)=z\deriv{}{z}\ln \left\{\TTmn{m,n}{\mu}(z)\right\} -(\mu+m+n+1)z+(\mu+m+2n+1)(\mu+n),
\eeq for the parameters
\beq \bfnu=(\mu+n,\mu+m+n+1,\mu+m+2n+1).
\eeq\end{subequations}}
\comment{The rational solution of \eqref{eq:S5Ok} in terms of the polynomial $\TTmn{m,n}{\mu}(z)$ 
is
\begin{subequations}\label{TTmn:S5Ok}\beq 
\widehat{h}_{m,n}(z;\mu)=z\deriv{}{z}\ln \left\{\TTmn{m,n}{\mu}(z)\right\} -\tfrac14(\mu+2m+2)z
-\tfrac18(\mu+2m+2)^2-(m+1)n,
\eeq for the parameters
\beq \begin{array}{l@{\qquad}l}
\k_{0} = -\tfrac{3}{4} \mu-\tfrac12m+n -\tfrac12, & \k_{1} =\tfrac14\mu-\tfrac12m-\tfrac12,\\[5pt] 
\k_2=\tfrac14\mu+\tfrac12m+n+\tfrac12, & \k_{3} =\tfrac14\mu+\tfrac12m+\tfrac12.
\end{array}
\eeq\end{subequations}}
\end{corollary}
\begin{proof}
Since $\TTmn{m,n}{\mu}(z)=\Tmn{m,n}{\mu}(-z)$ then
$\widehat{\sigma}_{m,n}(z;\bfnu)={\sigma}_{m,n}(-z;-\bfnu)$.
\end{proof}

\begin{remark}
 We note that 
 \begin{align*}
 &\sigma_{m,n}(z; m+1,-n,m+1-j) = \sigma_{m-j,n}(z; m+1-j,-n, m+1),\quad &j&=1,\ldots, m, \\
&\sigma_{m,n}(z; m+1,-n,0) = 0, \\
&\sigma_{m,n}(z; m,-n,1-j) = \sigma_{m,j-1}(z; m+1,1-j,-n),\quad &j&=2,\ldots ,n.
 \end{align*}
This result follow from the factorisation 
given in Lemma~\ref{lemma:multzeros} 
of the $T_{m,n}^{(\mu)}(z)$ at certain negative integer values of $\mu$.
The third case also follows from the invariance of the Hamiltonian $\HV(q,p,z)$ under 
the interchange of $\nu_2$ and $\nu_3$.
\end{remark}

\subsection{Non-uniqueness of rational solutions of \sPV}
In \S\ref{PVrata_nonuni} it was shown that there was non-uniqueness of rational solutions of \PV\ \eqref{eq:pv} in case (i) in terms of the generalised Laguerre polynomial $\Tmn{m,n}{\mu}(z)$ when $\mu$ is an integer. An analogous situation arises for rational solutions of 
\sPV\ \eqref{eq:S5JM}. 

\begin{lemma}
If $\mu\in\Z$ and $\mu \ge - n$ then there are two distinct rational solutions of \sPV\ \eqref{eq:S5JM} 
for the \textit{same} parameters.
\end{lemma}
\begin{proof}
If $\mu=k\in\Z$ and $k\geq-n$ then a second rational solution for the parameters \eqref{Tmn:S5JMb} is
\beq \label{TTmn:S5JMc}
\widehat{\sigma}_{m,n}(z;m+1,-n,m+n+k+1)
=z\deriv{}{z}\ln \left\{\TTmn{m,n+k}{-k}(z)\right\} -(m+1)z -(m+1)n.
\eeq
If $\mu=k\in\Z$ and $k\geq-n$ then a second rational solution for the parameters \eqref{TTmn:S5JMb} is
\beq\label{Tmn:S5JMc}
\sigma_{m,n}(z;m-1,n,-m-n-k-1)
=z\deriv{}{z}\ln \left\{\Tmn{m,n+k}{-k}(z)\right\}+(m+1)z-(m+1)n.
\eeq
\end{proof}

\subsection{Applications}
\def\sig{S}
\subsubsection{Probability density functions associated with the Laguerre unitary ensemble}
In their study of probability density functions associated with Laguerre unitary ensemble (LUE), 
Forrester and Witte \cite{refFW02} were interested in solutions of
\begin{align}
\left(z\deriv[2]{\sig}{z} \right)^{\!2}
&=\left[2 \left(\deriv{\sig}{z} \right)^{\!2}+
\left(2M+\ell-\mu -z \right) \left(\deriv{\sig}{z}\right)+\sig \right]^{2\!} \nn\\ &\qquad-4 \deriv{\sig}{z}
\left(\deriv{\sig}{z} -\mu \right) \left(\deriv{\sig}{z}+M \right) \left(\deriv{\sig}{z}+M+\ell \right),
\end{align}
where $M\ge 0$, $\ell \in \NN$ and $\mu$ is a parameter, which is \sPV\ \eqref{eq:S5JM} with parameters
$\bfnu=(-\mu,M,M+\ell)$. Forrester and Witte 
\cite[Proposition 3.6]{refFW02} define the solution 
\eq
\sig(z;-\mu,M,M+\ell) = -\mu M -Mz+ z\deriv{}{z} \ln 
\det\lf[\deriv[j]{}{z} 
L_{M+k}^{(\mu)}(-z) \ri]_{j,k=0}^{a-1},
\en
which behaves as 
\eq 
\sig(z;-\mu,M,M+\ell) = -\mu M - \frac{\mu M}{\mu+\ell} z+ \mathcal{O} (z^2 ), \qquad \text{as}\qquad z\to 0. 
\en
In terms of the generalised Laguerre polynomial $\Tmn{m,n}{\mu}(z)$, we have 
\eq
\sig(z;-\mu,M,M+\ell) = - \mu M -M z+ z\deriv{}{z} \ln \Tmn{M-1,\ell}{\mu-\ell}(-z).
\en
Explicitly, we have 
\begin{align}
\det\lf[\deriv[j]{}{z} 
L_{M+k}^{(\mu)}(-z) \ri]_{j,k=0}^{\ell-1} 
&= (-1)^{\floor{\ell/2}}\,\Tmn{M-1,\ell}{\mu-\ell}(-z) \\ 
&= (-1)^{\floor{\ell/2}+\floor{(M+\ell)/2}} \Tmn{\ell-1, M}{-\mu-\ell-2M}(z). 
\end{align} 

\subsubsection{Joint moments of the characteristic polynomial of CUE random matrices}
In their study of joint moments of the characteristic polynomial of 
CUE random matrices, Basor \etal~\cite[equation (3.85)]{refBBBGIIK} were interested in solutions of the equation
\begin{subequations}\label{eq:Basor}
\begin{align}
\left(z\deriv[2]{\sig_k}{z}\right)^{\!2}&= \left[ 2\left(\deriv{\sig_k}{z}\right)^{\!2}-(2N+z)\deriv{\sig_k}{z}+\sig_k\right]^2 \nn\\ &\qquad
-4\deriv{\sig_k}{z} \left(\deriv{\sig_k}{z}+k\right)\left(\deriv{\sig_k}{z}-N\right)\left(\deriv{\sig_k}{z}-k-N\right)\!,
\label{eq:Basor1}\end{align}
where $N,k\in\Z$ with $n\geq k>1$,
which is \sPV\ \eqref{eq:S5JM} with parameters $\bfnu=(k,-N,-k-N)$,
satisfying the initial condition
\beq\label{eq:Basor2}
\sig_k(z)=-kN+\tfrac12Nz+\O(z^2),\qquad\text{as}\quad z\to0.\eeq
\end{subequations}
Basor \etal\ derive the solution of \eqref{eq:Basor}, see \cite[equation (4.23)]{refBBBGIIK}, given by
\beq\label{sol:Basor}
\sig_k(z) = -kN+z\deriv{}{z}\ln{B_k(z)},\eeq
where $B_k(z)$ is the determinant
\eq
B_k(z) = \det\lf[ L_{N+k+1-i-j}^{(2k-1)}(-z) \ri]_{i,j=1}^k
,\qquad N\ge k> 1
\label{refBasor3}
\en
with $L_n^{(\a)}(z)$ the associated Laguerre polynomial.
Basor \etal~\cite{refBBBGIIK} remark that equation
\eqref{eq:Basor1} is degenerate at $z = 0$, which is a singular point of the equation, and so the {Cauchy-Kovalevskaya} theorem is not applicable to the initial value problem \eqref{eq:Basor}. 

From \eqref{def:Tmnd}, we have 
\begin{align}
B_k(z) &= \widehat{T}_{N-1,k}^{(0)}(z) 
= (-1)^{\floor{((N+k)/2)}}\Tmn{k-1,N}{-2(k+N)}(z),
\end{align}
where the second equality follows from \eqref{symmT}. 
In terms of the generalised Laguerre polynomial $\Tmn{m,n}{\mu}(z)$, a solution of \eqref{eq:Basor} is given by
\beq\label{sol:PAC}
\sigma(z;k,-N,-k-N) = -kN+Nz+z\deriv{}{z}\ln \{\Tmn{N-1,k}{0}(z)\},\qquad N\geq1, \quad k\geq1.
\eeq 
Alternatively, in terms of the polynomial $\TTmn{m,n}{\mu}(z)$, a solution of \eqref{eq:Basor} is given by
\[\widehat{\sigma}(z;k,-N,-k-N) = -kN+z\deriv{}{z}\ln \TTmn{N-1,k}{0}(z),\qquad N\geq1, \quad k\geq1,\]
which is the same solution as \eqref{sol:Basor}, though without the constraint $N\geq k$. Therefore we have two \textit{different} solutions of the initial value problem \eqref{eq:Basor}. 
The solutions \eqref{sol:Basor} and \eqref{sol:PAC} are related by 
\[ \sig_k(z) ={\sigma(z;k,-N,-k-N)-Nz},\]
since equation \eqref{eq:Basor} is invariant under the tranformation
\[ \sigma(z)\to \sigma(z)-Nz,\qquad z\to -z.\]

For example, suppose that $N=2$ and $k=2$, then from \eqref{sol:Basor}
\[\begin{split}
\sig_2(z)&= -\frac{16z^3+192z^2+720z+960}{z^4+16z^3+96z^2+240z+240} 
= -4+z-\frac{z^2}{5}+\frac{3\,z^4}{100}+\frac{z^5}{45}+\O(z^6).
\end{split}\]
 and from \eqref{sol:PAC} 
\[\begin{split}
\sigma(z;2,-2,-4)&= 2z+\frac{16z^3-192z^2+720z-960}{z^4-16z^3+96z^2-240z+240} 
= -4+z-\frac{z^2}{5}+\frac{3\,z^4}{100}-\frac{z^5}{45}+\O(z^6).
\end{split}\]

If we seek a series solution of \eqref{eq:Basor} in the form
\[ \sigma(z)=-Nk+\tfrac12Nz+ \sum_{j=2}^\infty a_{j}z^j,\]
then $a_{2j}$ are uniquely determined with
\[ a_2=\frac{(N+2k)N}{4(4k^2-1)},\qquad a_4 =\frac{(N+2k+1)(N+2k)(N+2k-1)N}{16(4k^2-1)^2(4k^2-1)}+\frac{36(4k^2-1)(k^2-1)}{N(N+2k)(4k^2-9)}a_3^2,
\quad\ldots\ ,\]
and $a_{2j+1}=0$ unless $k$ is an integer. If $k$ is an integer then $a_{2j+1}=0$ for $j<k$, $a_{2k+1}
$ {is arbitrary}, and $a_{2j+1}$ uniquely determined for $j>k$, as discussed in \cite{refBBBGIIK}. For example, when $N=2$ and $k=2$ then
\[ \sigma(z;k,-N,-k-N)= -4+z-\frac{z^2}{5}+\frac{3\,z^4}{100}+a_5z^5+\frac{29\,z^6}{3000}+\frac{4a_5\,z^7}{25}+\frac{263\,z^8}{360000}-\frac{13a_5\,z^9}{6000}+
\O(z^{10}), \]
with $a_5$ arbitrary.

The solutions $\sig_2(z)$ and $\sigma(z;2,-2,-4)$ have completely different asymptotics as $z\to\infty$, namely
\begin{align*}
\sig_2(z)&=-\frac{16}{z}+\frac{64}{z^2}+\frac{208}{z^3}+\frac{64}{z^4}-\frac{7424}{z^5}+\O(z^{-6}),\\
\sigma(z;2,-2,-4)&= 2z+\frac{16}{z}+\frac{64}{z^2}-\frac{208}{z^3}+\frac{64}{z^4}+\frac{7424}{z^5}+\O(z^{-6}).
\end{align*}

\section{Rational solutions of the symmetric \PV\ system}
\label{sec:spv}
From the works of Okamoto
\cite{refOkamotoP2P4,refOkamotoPV,refOkamotoPVI,refOkamotoPIII}, it is known
that the parameter spaces of \PII--\PVI\ all admit the action of an
extended affine Weyl group; the group acts as a group of {\bts}. In a series
of papers, Noumi and Yamada \cite{refNoumi,refNY98i,refNY98iii,refNY04}
have implemented this idea to derive a hierarchy of dynamical systems
associated to the affine Weyl group of type $\A{N}$, which are now known as
``\emph{symmetric forms of the \peqs}". The behaviour of each dynamical
system varies depending on whether $N$ is even or odd. 

The first member of the $\A{2n}$ hierarchy, i.e.\ $\A{2}$, usually known as \spiv, is equivalent to \PIV\ and given by
\begin{subequations}\label{eq:spiv}\begin{align}
&\deriv{f_1}{z}=f_1(f_2-f_3)+\k_1,\\
&\deriv{f_2}{z}=f_2(f_3-f_1)+\k_2,\\
&\deriv{f_3}{z}=f_3(f_1-f_2)+\k_3,
\end{align}
with constraints
\begin{equation}\k_1+\k_2+\k_3=1,\qquad f_1+f_2+f_3=z.\end{equation}
\end{subequations}
The first member of the $\A{2n+1}$ hierarchy, i.e.\ $\A{3}$, usually known as \spv, is equivalent to \PV\ \eqref{eq:pv}, as shown below, and given by
\begin{subequations}\label{eq:spv}\begin{align}
 z\deriv{f_1}{z}&=f_1f_3(f_2-f_4)+(\tfrac12-\k_3)f_1+\k_1f_3,\\
z\deriv{f_2}{z}&=f_2f_4(f_3-f_1)+(\tfrac12-\k_4)f_2+\k_2f_4,\\
z\deriv{f_3}{z}&=f_3f_1(f_4-f_2)+(\tfrac12-\k_1)f_3+\k_3f_1,\\
z\deriv{f_4}{z}&=f_4f_2(f_1-f_3)+(\tfrac12-\k_2)f_4+\k_4f_2,
\end{align}
with the normalisations
\beq\label{eq:sPVa}
f_1(z)+f_3(z)=\sqrt{z},\qquad f_2(z)+f_4(z)=\sqrt{z}
\eeq\end{subequations}
and $\k_1$, $\k_2$, $\k_3$ and $\k_4$ are constants such that
\beq\label{eq:sPVk}
\k_1+\k_2+\k_3+\k_4=1.
\eeq
The symmetric systems \spiv\ \eqref{eq:spiv} and \spv\ \eqref{eq:spv} were found by Adler \cite{refAdler} in the context of periodic
chains of B\"acklund transformations, see also \cite{refVesSh}. The symmetric
systems \spiv\ \eqref{eq:spiv} and \spv\ \eqref{eq:spv} have applications in random matrix theory,
see, for example, 
\cite{refFW01,refFW02}.

Setting $f_1(z)=\sqrt{z}\,u(z)$ and $f_2(z)=\sqrt{z}\,v(z)$, in \spv\ \eqref{eq:spv} 
gives the system
\begin{subequations}\label{eq:sPV2}\begin{align}
z\deriv{{u}}{z}&=z(2{v}-1){u}^2-(2z{v}-z+\k_1+\k_3){u}+\k_1,\label{eq:sPV2a}\\
z\deriv{{v}}{z}&=z(1-2{u}){v}^2+(2z{u}-z-\k_2-\k_4){v}+\k_2.\label{eq:sPV2b}
\end{align}
\end{subequations}
Solving \eqref{eq:sPV2a} for ${v}$, substituting in \eqref{eq:sPV2b}
gives
\begin{align} \deriv[2]{{u}}{z}=\frac12&\left(\frac{1}{{u}}+\frac{1}{{u}-1}\right)\left(\deriv{{u}}{z}\right)^{\!2}-\frac{1}{z}\deriv{{u}}{z}
+\frac{({u}-1)^2\k_1^2-{u}^2\k_3^2}{2z^2{u}({u}-1)} \nn\\ &
+ \frac{(\k_2-\k_4){u}({u}-1)}{z}+\frac{{u}({u}-1)(2{u}-1)}{2}. \label{sol:unm}
\end{align}
Making the transformation ${u}=1/(1-w)$ in \eqref{sol:unm} yields
\begin{subequations}\begin{align} \deriv[2]{w}{z}&=\left(\frac{1}{2w}+\frac{1}{w-1}\right)\left(\deriv{w}{z}\right)^{\!2}-\frac{1}{z}\deriv{w}{z}
+\frac{(w-1)^2(w^2\k_1^2-\k_3^2)}{2z^2w} +\frac{(\k_2-\k_4)w}{z}-\frac{w(w+1)}{2w-1)},
\end{align}
which is \PV\ \eqref{eq:pv} with parameters
\beq\pms{\tfrac12\k_1^2}{-\tfrac12\k_3^2}{\k_2-\k_4}.\label{usol-params}\eeq\end{subequations}
{Analogously solving \eqref{eq:sPV2b} for ${u}$, substituting in \eqref{eq:sPV2a} 
gives
\begin{align*} 
\deriv[2]{{v}}{z}=\frac12&\left(\frac{1}{{v}}+\frac{1}{{v}-1}\right)\left(\deriv{{v}}{z}\right)^{\!2}-\frac{1}{z}\deriv{{v}}{z}
+\frac{({v}-1)^2\k_2^2-{v}^2\k_4^2}{2z^2{v}({v}-1)} \nn\\ & 
+\frac{(\k_3-\k_1){v}({v}-1)}{z}+\frac{{v}({v}-1)(2{v}-1)}{2}.\end{align*}
Then making the transformation ${v}=1/(1-w)$ gives \PV\ \eqref{eq:pv} with parameters
\[\pms{\tfrac12\k_2^2}{-\tfrac12\k_4^2}{\k_3-\k_1}.\]}

As shown above, \PV\ \eqref{eq:pv} has the rational solution in terms of the generalised Laguerre polynomial $\Tmn{m,n}{\mu}(z)$ given by
\begin{subequations}\beq {w}_{m,n}(z;\mu)=\frac{\Tmn{m-1,n}{\mu}(z)\,\Tmn{m-1,n+1}{\mu}(z)}{\Tmn{m,n}{\mu}(z)\,\Tmn{m-2,n+1}{\mu}(z)},\eeq
for the parameters
\beq\pms{\tfrac12m^2}{-\tfrac12(m+2n+\mu+1)^2}{\mu},\label{Lagrat-params}\eeq\end{subequations}
and so
\beq u_{m,n}(z;\mu)=\frac{1}{1-{w}_{m,n}(z;\mu)}=\frac{\Tmn{m,n}{\mu}(z)\,\Tmn{m-2,n+1}{\mu}(z)}{\Tmn{m,n}{\mu}(z)\,\Tmn{m-2,n+1}{\mu}(z)-
\Tmn{m-1,n}{\mu}(z)\,\Tmn{m-1,n+1}{\mu}(z)}. \eeq
{From 
equations \eqref{iden:Tmn2} in Lemma \ref{lem:312} and \eqref{324c} in Lemma \ref{lem:313}, with $n\to n+1$, we have}
\begin{align}
&\Tmn{m,n}{\mu}\,\Tmn{m,n+1}{\mu}-\Tmn{m,n+1}{\mu-1}\,\Tmn{m,n}{\mu+1}=\Tmn{m+1,n}{\mu}\,\Tmn{m-1,n+1}{\mu},\\
&\D_z\left(\Tmn{m,n}{\mu+1}\cdot\Tmn{m,n+1}{\mu-1}\right)=\Tmn{m+1,n}{\mu}\,\Tmn{m-1,n+1}{\mu},
\end{align}
with $\D_z$ the Hirota operator \eqref{Hirota}, and so the solution of equation \eqref{sol:unm} is given by
\beq u_{m,n}(z;\mu) =-\frac{\Tmn{m,n}{\mu}(z)\,\Tmn{m-2,n+1}{\mu}(z)}{\Tmn{m-1,n}{\mu+1}(z)\,\Tmn{m-1,n+1}{\mu-1}(z)} 
=\deriv{}{z}\ln\frac{\Tmn{m-1,n+1}{\mu-1}(z)}{\Tmn{m-1,n}{\mu+1}(z)},\qquad m\geq1,\quad n\geq1.\eeq
In the case when $n=0$ then
\beq u_{m,0}(z;\mu) =-\frac{\Tmn{m-2,1}{\mu}(z)}{\Tmn{m-1,1}{\mu-1}(z)} 
=\deriv{}{z}\ln\Tmn{m-1,1}{\mu-1}(z),\qquad m\geq1. \eeq
We note that
\[u_{m,0}(z;\mu) =-\frac{\Lag{m}{\mu+1}(z)}{\Lag{m+1}{\mu}(z)}=\deriv{}{z}\ln \Lag{m}{\mu}(z).
\]

\def\kk#1#2#3#4{\boldsymbol{\k} = (#1,#4,#3,#2)}
\def\K#1#2#3#4{\boldsymbol{\k} = (#1,#4,#3,#2)}
From equation \eqref{eq:sPV2a}, we obtain
\beq {v}=\frac{1}{2z{u}({u}-1)}\left\{z\deriv{{u}}{z}+z{u}^2-(z-\k_1-\k_3){u}-\k_1\right\}.\eeq
Depending on the choice of $\k_1$ and $\k_3$, there is a different solution for ${v}$.
From \eqref{eq:sPVk}, \eqref{usol-params} and \eqref{Lagrat-params} we obtain
\[\k_1^2=m^2,\qquad\k_3^2=(m+2n+\mu+1)^2,\qquad \k_2-\k_4=\mu,\qquad \k_1+\k_2+\k_3+\k_4=1,\]
which gives four solutions
\begin{align*}
&\kk{m}{ -m -n -\mu}{\mu +m +2 n +1}{-m -n},\label{kusol}\\
&\kk{m}{n+1}{-\mu -m -2 n -1}{\mu+n+1},\\
&\kk{-m}{ -n -\mu}{\mu +m +2 n +1}{ -n},\\
&\kk{-m}{m+n+1}{-\mu -m -2 n -1}{\mu+m+n+1}.
\end{align*}
Each of these gives a different solution $v_{m,n}(z)$ which we will discuss in turn.

\begin{enumerate}[(i)] 
\item For the parameters $\K{m}{ -m -n -\mu}{\mu +m +2 n +1}{-m -n}$, the solution is
\begin{subequations}
\begin{align}
v_{m,n}^{\rm(i)}(z;\mu) &=-\frac{m+n}{z}\,
\frac{\Tmn{m-1,n+1}{\mu-1}(z)\,\Tmn{m-2,n}{\mu+1}(z)}{\Tmn{m-1,n}{\mu}(z)\,\Tmn{m-2,n+1}{\mu}(z)} \nn \\ &
=1-\frac{\mu+2n+1}{z}+\deriv{}{z}\ln\frac{\Tmn{m-1,n}{\mu}(z)}{\Tmn{m-2,n+1}{\mu}(z)},\qquad m\geq1,\quad n\geq1,\\
v_{m,0}^{\rm(i)}(z;\mu) &=-\frac{m}{z}\,
\frac{\Tmn{m-1,1}{\mu-1}(z)}{\Tmn{m-2,1}{\mu}(z)}=1-\frac{\mu+1}{z}-\deriv{}{z}\ln \Tmn{m-2,1}{\mu}(z),\qquad m\geq1.
\end{align}\end{subequations}

\item For the parameters $\K{m}{n+1}{-\mu -m -2 n -1}{\mu+n+1}$, the solution is
\begin{align}
v_{m,n}^{\rm(ii)}(z;\mu) &=\frac{\Tmn{m-1,n+1}{\mu-1}(z)\,\Tmn{m-2,n+1}{\mu+1}(z)}{\Tmn{m-1,n+1}{\mu}(z)\,\Tmn{m-2,n+1}{\mu}(z)} 
=1+\deriv{}{z}\ln\frac{\Tmn{m-1,n+1}{\mu}(z)}{\Tmn{m-2,n+1}{\mu}(z)},\qquad m\geq1,\quad n\geq0.
\end{align}

\item For the parameters $\K{-m}{ -n -\mu}{\mu +m +2 n +1}{ -n}$, the solution is
\begin{align}
v_{m,n}^{\rm(iii)}(z;\mu) &=-\frac{\Tmn{m,n - 1}{\mu + 1}(z)\,\Tmn{m - 1, n + 1}{\mu - 1}(z)}{\Tmn{m - 1, n}{\mu}(z)\,\Tmn{m, n}{\mu}(z)}
=\deriv{}{z}\ln\frac{\Tmn{m-1,n}{\mu}(z)}{\Tmn{m,n}{\mu}(z)},\qquad m\geq1,\quad n\geq1,
\end{align}
and $v_{m,0}^{\rm(iii)}(z;\mu)=0$.

\item For the parameters $\K{-m}{m+n+1}{-\mu -m -2 n -1}{\mu+m+n+1}$, the solution is
\begin{subequations}
\begin{align}
v_{m,n}^{\rm(iv)}(z;\mu) &=\frac{\mu+m + n + 1}{z}\,\frac{\Tmn{m, n}{\mu + 1}\,\Tmn{m - 1, n + 1}{\mu - 1}}{\Tmn{m, n}{\mu}\,\Tmn{m - 1, n + 1}{\mu}} \nn\\
&=\frac{\mu+2n+1}{z}+\deriv{}{z}\ln \frac{\Tmn{m -1, n +1}{\mu}(z)}{\Tmn{m,n}{\mu}(z)},\qquad m\geq1,\quad n\geq1,\\
v_{m,0}^{\rm(iv)}(z;\mu) &=\frac{\mu+m + 1}{z}\,\frac{\Tmn{m - 1,1}{\mu - 1}}{\Tmn{m - 1,1}{\mu}}= \frac{\mu+1}{z}+\deriv{}{z}\ln \Tmn{m-1,1}{\mu}(z),\qquad m\geq1.
\end{align}
\end{subequations}
\end{enumerate}

\begin{remarks}{\rm 
\begin{enumerate}[(i)]\item[]
\item Analogous rational solutions of \spv\ \eqref{eq:spv} can be derived in terms of the polynomial $\TTmn{m,n}{\mu}(z)=\Tmn{m,n}{\mu}(-z)$ given by
\[ \widehat{u}_{m,n}(z;\mu)=u_{m,n}(-z;\mu),\qquad
\widehat{v}_{m,n}(z;\mu)=v_{m,n}(-z;\mu).\]
\item Some rational solutions of \spv\ \eqref{eq:spv} are given in \cite{refAGLZ,refGGM20,refGGLM}, where a different normalisation of the symmetric system is used.
\end{enumerate}}\end{remarks}

\subsection{Non-uniqueness of rational solutions of \spv}
As was the case for \PV\ \eqref{eq:pv} and \sPV\ \eqref{eq:S5JM}, there is non-uniqueness for some rational solutions of the symmetric system \spv\ \eqref{eq:spv}. We illustrate this with an example.
\begin{example}{\rm
The sets of functions 
\[ u_{1,1}(z;1)=\frac{(z-2)(z-6)}{(z-4)(z^2-6z+12)},\qquad v_{1,1}^{\rm(i)}(z;1)=\frac{z^2-6z+12}{z (z-3)}, \]
and
\[ \widehat{u}_{1,2} (z;-1)=-\frac{z^{4}+12 z^{3}+54 z^{2}+96 z +72}{(z^{2}+6 z +12) (z^{3}+6 z^{2}+18 z +24)},
\qquad \widehat{v}_{1,2}^{\rm(i)} (z;-1)=-\frac{2( z^{2}+6 z +12)}{z (z^{2}+4 z +6)}, \]
are both solutions of the system \eqref{eq:sPV2} for the parameters
\[\boldsymbol{\k} =(1,-2,5,-3).\]
Hence the associated solutions of \spv\ \eqref{eq:spv} are
\begin{align*} 
&f_1(z)=\frac{\sqrt{z}\,(z-2)(z-6)}{(z-4)(z^2-6z+12)},&& f_2(z)=\frac{\sqrt{z}\,(z^2-6z+12)}{z (z-3)}, \\
& f_3(z)=\frac{\sqrt{z}\,(z-3)(z^2-8z+20)}{(z-4)(z^2-6z+12)}, &&f_4(z)=\frac{3\sqrt{z}\,(z-4)}{z(z-3)}, 
\end{align*}
and
\begin{align*} & \widehat{f}_1(z)=-\frac{\sqrt{z}\,(z^{4}+12 z^{3}+54 z^{2}+96 z +72)}{(z^{2}+6 z +12) (z^{3}+6 z^{2}+18 z +24)},
&&\widehat{f}_2(z) =-\frac{2\sqrt{z}\,( z^{2}+6 z +12)}{z (z^{2}+4 z +6)}, \\
& \widehat{f}_3(z)=\frac{\sqrt{z}\,(z^{2}+4 z +6) (z^{3}+9 z^{2}+36 z +60)}{(z^{2}+6 z +12)(z^{3}+6 z^{2}+18 z +24)},
&&\widehat{f}_4 (z)=\frac{\sqrt{z}\,(z^{3}+6 z^{2}+18 z +24)}{z (z^{2}+4 z +6)}.
\end{align*}

}\end{example}

\section{Properties of generalised Laguerre polynomials} 
\label{sec:lag}
\begin{remark}
 The generalised Laguerre polynomial $\Tmn{m,n}{\mu}(z)$ is such that 
\begin{align}
\Tmn{m,n}{\mu}(z) = c_{m,n} &\lf \{ z^{(m+1)n} 
- n \bigl(m{+}1\bigr) (m{+}n{+}1{+}\mu ) z^{(m+1)n-1} \nn\right. \\ 
&\quad + \tfrac12n(m+1)(m+n+1+\mu)[(m +1) (m n +n^{2}+n -2)+(m n +n -1) \mu ]
 z^{(m+1)n-2} 
\nn \\ &\qquad \left. +\ldots+ (-1)^{n(m+n) } d_{m,n}\right\}
\label{Texp}
\end{align}
where 
\eq
c_{m,n}=(-1)^{n(2m+1+n)/2} \prod_{j=1}^n \frac{(j-1)! }{(m+j)!}, \label{def:cmn}
\en
which follows from Lemma 1 in \cite{refBK18}, and 
\eq
d_{m,n} =
\prod_{ j=1}^{{\rm min} (m+1,n)-1} (\mu+n+j)^j 
\prod_{ {\rm min}(m+1,n) }^{{\rm max} (m+1,n)} 
(\mu+n+j)^{{\rm min}(m+1,n)}
\prod_{{\rm max}(m+1,n)+1 }^{m+n} 
 (\mu+n+j)^{m+n+1-j}.
\en
Therefore 
\eq
\Tmn{m,n}{-n-j}(0)=0,\qquad j=1,2,\ldots, m+n.
\label{T0}
\en
\end{remark}

\begin{lemma}\label{lemma:multzeros}
The generalised Laguerre polynomials have multiple 
roots at the origin when 
\eq
\mu=-n-j,\qquad j=1,2,\ldots, m+n.
\label{mult}
\en
 Moreover at such values of $\mu$ the polynomials $\Tmn{m,n}{\mu}(z)$ factorise 
as 
\begin{align} 
\Tmn{m,n}{-n-j}(z)& = \frac{c_{m,n}}{c_{m-j,n}} z^{nj} 
\,\Tmn{m-j, n}{j-n}(z), & j&=1,2,\ldots, m, \\ 
\Tmn{m,n}{-m-n-1}(z) &= c_{m,n} \, z^{n(m+1)}, & \\ 
\Tmn{m,n}{-m-n-j}(z) &= \frac{c_{m,n}}{c_{m,j-1}} 
\,z^{(m+1)(n+1-j)} \, \Tmn{m,j-1}{-m-n-j}(z), & j&=2,\ldots, n, 
\label{Tzero3}
\end{align}
where 
$$ 
\Tmn{m-j, n}{j-n}(0) \ne 0, \qquad \Tmn{m,j-1}{-m-n-j}(0)\ne 0. 
$$
\end{lemma}
\begin{proof}
The fact that the generalised Laguerre polynomials have 
multiple roots at the points \eqref{mult} follows from 
the discriminant, and 
that these roots are always at the origin is a consequence of \eqref{T0}. 
We use the standard property of Wronskians 
\eq
\Wr\big(c_1 g(x) f_1(x), \ldots, c_r g(x)f_r(x)\big) = 
\lf( \prod_{i=1}^r c_i \ri ) [g(x)]^r \Wr\big(f_1(x),\ldots, f_r(x)\big) 
,\qquad c_1, \ldots, c_r \in \CC,
\label{stan_wron}
\en
and the property (see, for example, \cite{refKMcL})
\eq
L_n^{(\alpha)}(z) = \frac{(n+\alpha)!}{n!}(-z)^{-\alpha} L_{n+\alpha}^{(-\alpha)}(z),\qquad \alpha \in \{-n,-n+1, \ldots, -1\},
\label{newL}
\en
to rewrite 
\eq
\Tmn{m,n}{-m-n-1}(z) =
\Wr\left( L_{m+1}^{(-m-1)}(z), L_{m+2}^{(-m-1)}(z), 
\ldots, 
 L_{m+n}^{(-m-1)}(z) \right),
\en
as 
\eq
\Tmn{m,n}{-m-n-1}(z) =
(-z)^{n(m+1)} \prod_{j=0}^{n-1} \frac{j!}{(m+j+1)!} 
\Wr\left( L_{0}^{(m+1)}(z), L_{1}^{(m+1)}(z), 
\ldots, 
 L_{n-1}^{(m+1)}(z)\right). 
\en
Since $ L_{0}^{(m+1)}(z)=1$ and 
\eq
\Wr\big( 1, f_1(x),f_2(x),\ldots, f_r(x)\big) = 
\Wr\big( f_1'(x),f_2'(x),\ldots, f_r'(x)\big),
\label{wron_df}
\en
we repeatedly use (\ref{lag1}) and (\ref{wron_df}) to show that 
\eq
\Wr\left( L_{0}^{(m+1)}(z), L_{1}^{(m+1)}(z), 
\ldots, 
 L_{n-1}^{(m+1)}(z) \right) = \prod_{j=0}^{n-1} (-1)^j.
\en
Hence we obtain 
\eq 
\Tmn{n,m}{-m-n-1}(z)= (-z)^{n(m+1)} \prod_{j=0}^{n-1} \frac{ (-1)^j j!}{(m+j+1)!} 
=c_{m,n} \, z^{n(m+1)}. 
\en

When $\alpha=-n-j$ for $ j=1,2,\ldots,m$, we 
again use (\ref{newL}) and (\ref{stan_wron}) 
to obtain 
\begin{align}
\Tmn{m,n}{-n-j}(z) &= \Wr\left( L_{m+1}^{(-j)}(z), L_{m+2}^{(-j)}(z), 
\ldots, 
L_{m+n}^{(-j)}(z)\right)
\nn \\ 
&= z^{nj}\, (-1)^{nj} 
\prod_{i=1}^n \frac{(m-j+i)! }{ (m+i)!}
\Wr\left( L_{m+1-j}^{(j)}(z), L_{m+2-j}^{(j)}(z), 
\ldots, 
 L_{m+n-j}^{(j)}(z)\right)\nn \\
&= \frac{c_{m,n}}{c_{m-j,n}} z^{nj}\,\Tmn{m-j, n }{j-n}(z).
\end{align}

The final case of $\alpha=-m-n-j$ 
for $j=2,3,\ldots,n$ follows similarly, except that we first apply 
 the symmetry \eqref{symmT} in order to use (\ref{newL}). 
 Specifically, we have 
\begin{align}
\Tmn{m,n}{-m-n-j}(z) &= (-1)^{\floor{(m+n+1)/2}} \,
\TTmn{n-1, m+1}{-m-n+j-2}(z)\nn \\ 
&=(-1)^{\floor{(m+n+1)/2}} 
z^{(m+1)(n-j+1)} \prod_{i=0}^m 
\frac{ (j+i-1)!}{(n+i)!} \nn \\ 
& \quad \quad 
\times \Wr\left( L_{j-1}^{(n+1-j)}(-z), L_{j}^{(n+1-j)}(-z), 
\ldots, 
L_{j+m-1}^{(n+1-j)}(-z)\right)
\nn \\ 
&= (-1)^{\floor{(m+n+1)/2}} 
z^{(m+1)(n-j+1)} \prod_{i=0}^m 
\frac{ (j+i-1)!}{(n+i)!}
 \, \TTmn{j-2,m+1}{n-m-j}(z).\nn 
\end{align}
Applying the symmetry \eqref{symmT} yields (\ref{Tzero3}).
Finally, 
\[
\Tmn{m-j,n}{j-n}(0) \ne 0,\qquad j=1,2,\ldots,m,
\]
and 
\[
\Tmn{m,j-1}{-m-n-j}(0)\ne 0,\qquad j=2,\ldots,n,
\]
follow from Lemma 2 in \cite{refBK18}.
\end{proof}

\begin{remark}
 The Young diagrams of the polynomials on the right-hand side of (\ref{Tzero3}) are 
 found from the Young diagram of $\bfla=((m+1)^n) $ for $j=1,2,\ldots,m+1$ by removing the right-most $j$ columns. 
When $j=2 ,3,\ldots, n$ the Young diagrams are those such that the bottom $n-j+1$ rows have been removed from $\bfla$.
\end{remark}


\begin{definition} A \textit{Wronskian Hermite polynomial} 
$H_{\bfla}(z)$, labelled by partition $\bfla$, is a 
Wronskian of probabilists' Hermite polynomials $\He_n(z)$ given by 
\eq \label{eq:WHP}
H_{\bfla}(z)= \frac{\Wr\left(\He_{h_1}(z),\He_{h_2}(z),
 \ldots,\He_{h_r}(z)\right)}{\Delta(\bold{h}_{\bfla})}.
\en
The scaling by the Vandermonde determinant
$\Delta(\bold{h}_{\bfla})$ ensures the polynomials are monic.
\end{definition}

\begin{remark}
The well-known identities relating Hermite polynomials and Laguerre polynomials
\[
\He_{2n}(z) = (-1)^n 2^n n! \,L_n^{(-1/2)}(\tfrac12z^2) ,\qquad
\He_{2n+1}(z) = (-1)^n 2^n n! \,z L_n^{(1/2)}(\tfrac12z^2),
\]
cf.~\cite[\S18.7]{refDLMF}, mean that generalised Laguerre polynomials evaluated at negative half-integers are 
related to Wronskian Hermite polynomials. We specialise 
Corollary 4 in \cite{refBDS} to the generalised Laguerre polynomials 
$\Omega_{\bold{\nu}}^{(\alpha)}(z)$. 
Suppose partition $\bold{\Lambda}=\bold{\Lambda}(k, \bold{\nu})$ has $2$-core $k$ and $2$-quotient $(\bold{\nu},\bold{\emptyset})$. Set $\alpha_k =-\tfrac12-\ell(\bfnu) -k$. Then 
\eq
H_{\bold{\Lambda}(k, \bold{\nu})}(z) = 2^{|\bold{\nu}|} 
 z^{k(k-1)/2} \frac{\prod_{j=1}^{\ell(\bfnu)} 
(-1)^{h_j} \,h_j!}{\Delta(\bold{h}_{\bold{\nu}})} 
\, \Omega_{\bold{\nu}}^{(\alpha_k)} 
\lf (\tfrac12 z^2\ri),
\label{HL}
\en
where $\bold{h}_{\bold{\nu}}=(h_1,\ldots,h_r)$ is the degree vector of partition $\bold{\nu}$.
\end{remark}





\begin{lemma}
Set $\alpha_k =-2n-k-\tfrac12$ for $k=0,1,\ldots$\ . Then 
\eq
\Tmn{m,n}{-2n-k-1/2}(\fract{1}{2}z^2) =2^{-n(m+1)} c_{m,n} \, 
z^{-k(k+1)/2} \, H_{\bfLA_{k,m,n}}(z),
\label{lagH}
\en
where the partition $\bfLA_{k,m,n}$ 
is 
%
\eq
\bfLA_{k,m,n} =
\begin{cases}
\lf ( \{2m-j-k+1\}_{j=0}^{n-1}, \{n+k- j\}_{j=0}^{n+k-1} \ri ), & k<m-n+2, \\[3pt]
\lf( \{2m-j-k+1\}_{j=0}^{m-k}, \{m+1\}_{j=1}^{2(n-m+k-1)} 
, \{m+1-j\}_{j=0}^{m} \ri ),
& m-n+2 \le k< m+1, \\[3pt] 
\lf (\{k-j\}_{j=0}^{k-m-1}, \{m+1\}_{j=0}^{2n-2} 
, \{m+1-j\}_{j=0}^{m} \ri ),
& k\ge m+1. 
\end{cases} 
\label{Lam}
\en
We can equivalently write 
\eq
\Tmn{m,n}{-2n-k-1/2}(\fract{1}{2}z^2) = b_{k,m,n}\,
 z^{-k(k+1)/2} \, 
\Wr\lf ( \{\He_{1+2j}\}_{j=0}^{n+k-1}, 
 \{\He_{2(m+1+j)}\}_{j=0}^{n-1} 
 \ri ),
 \label{TWHP2}
\en
where 
\eq
b_{k,m,n}=\frac{ 2^{-n(m+1)} c_{m,n} }
{ \Delta\lf ( \{1+2j\}_{j=0}^{n+k-1}, 
 \{2(m+1+j)\}_{j=0}^{n-1} \ri) }.
\label{dconst} 
\en
We also find 
\eq
\Tmn{m,n}{-2n-k-1/2}(\fract{1}{2}z^2) = (-1)^{n(m+1)}
2^{-n(m+1)} c_{m,n} \, 
z^{-k(k+1)/2} \, H_{\bfLA_{k,m,n}^*}(z),
\label{THcon}
\en
where $\bfLA_{k,m,n}^*$ denotes the conjugate partition to 
$\bfLA_{k,m,n}$ and $c_{m,n}$ is given by \eqref{def:cmn}.
\end{lemma}
\begin{proof}
Set $\mu= \mu_k = -2n-k-\tfrac12$ in (\ref{TOmega}) then
 \begin{align}
 \Tmn{m,n}{\nu}(\fract{1}{2}z^2) &= (-1)^{n(n-1)/2} \,
 \Omega_{\bfla}^{(-n-k-1/2)}(\fract{1}{2}z^2) \nn \\
 &= \frac{(-1)^{n(n-1)/2} 2^{n(m+1)} \Delta(\bold{h}_{\bfla}) }
 {\prod_{m=1}^n (-1)^{m+1} (m+1)! }z^{-k(k+1)/2} H_{
\bfLA_{k,m,n}}(z),
 \end{align}
 using (\ref{HL}) with 
 $\bold{\nu}=\bfla=((m+1)^n) $ and 
 $\alpha_k=n+\mu_k$. We denote by 
 $\bfLA_{k,m,n}$ the partition that has
 $2$-core $k$ and 
 $2$-quotient $(\bfla, \bold{\emptyset})$. 
 Simplifying the constant term, we obtain \eqref{lagH}. 
Moreover (\ref{THcon}) follows from \eqref{lagH} by replacing $z$ with $ \i z$ 
 and using the well-known relation 
 $$
 H_{\bfrh}(\i z) = \i^{|\bfrh|} H_{\bfrh^*}(z).
 $$

We determine the degree vector of partition $\bfLA_{k,m,n}$
from the degree vector 
$$
\bold{h}_{\bfla} =( m+1, m+3, \ldots, m+n),
$$
using \eqref{hks}.
Put beads in positions $2(m+1)$ to $2(m+n)$ on the left runner 
and in positions $1$ to $2(n+k-1)+1$ on the right runner. The 
 components of the degree vector of $\bfLA_{k,m,n} $ correspond 
to the positions of the beads: 
\eq
 \{2(m+1+j)\}_{j=0}^{n-1} \cup \{2j-1\}_{j=1}^{n+k}.
\label{fch}
\en
Writing the Wronskian Hermite polynomial explicitly in terms of \eqref{fch} gives \eqref{TWHP2}, where
the Vandermonde determinant in the denominator of the constant \eqref{dconst}
arises because the components of the degree vector as given in \eqref{fch} 
are not ordered. 

The degree vector $\bold{h}_{\bfLA_{k,m,n}}$ is obtained by ordering \eqref{fch} from 
largest value to smallest value.
Depending on $k,m,n$, there are three possibilities corresponding 
to the three abaci in
Figure \ref{fig:abaci}. 
We deduce from the abaci that the degree vector is
\eq
 \bold{h}_{\bfLA_{k,m,n}} =
 \begin{cases}\lf (\{ 2(m{+}n{-}j) \}_{j=0}^{n-1} 
, \{2(n{+}k{-}j ){-}1\}_{j=0}^{n+k-1} \ri),
& k< m{-}n{+}2,\\[3pt]
 \lf (\{ 2(m{+}n {-}j) \}_{j=0}^{m-k} 
, \{2(n{+}k) {-}1{-}j\}_{j=0}^{2(n+k-m)-3} 
, \{ 2(m{-}j){+}1 \}_{j=0}^m 
\ri ),
& m{-}n{+}2 \le k < m+1, \\[3pt]
 \lf (\{ 2(n{+}k {-}j) {-}1\}_{j=0}^{k-1-m} 
, \{2(m{+}n) {-}j\}_{j=0}^{2(n-2)} 
, \{ 2(m{-}j){+}1 \}_{j=0}^m 
\ri ),
& k\ge m{+}1. 
\end{cases}
\nn\en
The description of the partition $\bfLA_{k,m,n}$ in \eqref{Lam} follows from the degree vector using \eqref{degvec} with $r=2n+k$. 
\end{proof}
\begin{remark}
In \eqref{Lam} we have explicitly described the partition $\bfLA_{k,m,n}$ with $2$-core $k$ and $2$-quotient $((m+1)^n, \bold{\emptyset})$. This result may be of independent interest to those who work in combinatorics. 
\end{remark}
\begin{remark}
Wronskian Hermite polynomials of the type $ H_{\bfLA_{K,m,n}}(z)$ appear in \cite{refGGM20} in their classification of solutions to \PV\ at half-integer values of the associated Laguerre parameter using Maya diagrams. Such diagrams also represent partitions and there is straightforward connection between their results and the ones in this article. 
The $ H_{\bfLA_{K,m,n}}(z)$ are related to the $k=2$ cases studied in \S6 of \cite{refGGM20}; the $k=3$ case therein relates to solutions of generalised Umemura polynomials at half-integer values of the parameter.
\end{remark}
\begin{figure}
 \centering
 \begin{subfigure}[b]{0.3\textwidth}
 \centering

\begin{tikzpicture}[scale=0.7]
 \node[vertexe] (n0) at (0,0) {};
 \node[vertexf] (n1) at (1,0) {};
 \node[vertexe] (n2) at (0,-2) {};
 \node[vertexf] (n3) at (1,-2) {};
 \path (n0) -- (n2) node [black, font=\LARGE, midway, sloped] {$\ldots$}; 
 \path (n1) -- (n3) node [black, font=\LARGE, midway, sloped] {$\ldots$};

 \node[vertexe] (n4) at (0,-3) {};
 \node[vertexe] (n5) at (1,-3) {};
 \node[vertexe] (n6) at (0,-5) {};
 \node[vertexe] (n7) at (1,-5) {};
 \path (n4) -- (n6) node [black, font=\LARGE, midway, sloped] {$\ldots$}; 
 \path (n5) -- (n7) node [black, font=\LARGE, midway, sloped] {$\ldots$};

\node[vertexf] (n8) at (0,-6) {};
 \node[vertexe] (n9) at (1,-6) {};
 \node[vertexf] (n10) at (0,-8) {};
 \node[vertexe] (n11) at (1,-8) {};
 \path (n8) -- (n10) node [black, font=\LARGE, midway, sloped] {$\ldots$}; 
 \path (n9) -- (n11) node [black, font=\LARGE, midway, sloped] {$\ldots$}; 

\draw [decorate,decoration={brace,amplitude=10pt,raise=5pt},xshift=0pt,yshift=0pt]
(n1) -- (n3) node [black,midway,xshift=1.2cm]{\small $n+k$}; 

\draw [decorate,decoration={brace,amplitude=10pt,raise=5pt},xshift=0pt,yshift=0pt]
(n5) -- (n7) node [black,midway,xshift=1.8cm]{\small $m-n-k+1$};

\draw [decorate,decoration={brace,amplitude=10pt,raise=5pt},xshift=0pt,yshift=0pt]
(n10) -- (n8) node [black,midway,xshift=-0.8cm]{\small $n$};

\end{tikzpicture}

 \caption{$k< m-n+2$}
 \label{fig:abac11}
 \end{subfigure}
 \hfill
 \begin{subfigure}[b]{0.3\textwidth}
 \centering

\begin{tikzpicture}[scale=0.7]
 \node[vertexe] (n0) at (0,0) {};
 \node[vertexf] (n1) at (1,0) {};
 \node[vertexe] (n2) at (0,-2) {};
 \node[vertexf] (n3) at (1,-2) {};
 \path (n0) -- (n2) node [black, font=\LARGE, midway, sloped] {$\ldots$}; 
 \path (n1) -- (n3) node [black, font=\LARGE, midway, sloped] {$\ldots$};

 \node[vertexf] (n4) at (0,-3) {};
 \node[vertexf] (n5) at (1,-3) {};
 \node[vertexf] (n6) at (0,-5) {};
 \node[vertexf] (n7) at (1,-5) {};
 \path (n4) -- (n6) node [black, font=\LARGE, midway, sloped] {$\ldots$}; 
 \path (n5) -- (n7) node [black, font=\LARGE, midway, sloped] {$\ldots$};

\node[vertexf] (n8) at (0,-6) {};
 \node[vertexe] (n9) at (1,-6) {};
 \node[vertexf] (n10) at (0,-8) {};
 \node[vertexe] (n11) at (1,-8) {};
 \path (n8) -- (n10) node [black, font=\LARGE, midway, sloped] {$\ldots$}; 
 \path (n9) -- (n11) node [black, font=\LARGE, midway, sloped] {$\ldots$}; 

\draw [decorate,decoration={brace,amplitude=10pt,raise=5pt},xshift=0pt,yshift=0pt]
(n1) -- (n7) node [black,midway,xshift=1.2cm]{\small $n+k$};

\draw [decorate,decoration={brace,amplitude=10pt,raise=5pt},xshift=0pt,yshift=0pt]
(n10) -- (n4) node [black,midway,xshift=-0.8cm]{\small $n$};

\end{tikzpicture}

 \caption{ $m-n+2 \le k <m$}
 \label{fig:abac22}
 \end{subfigure}
 \hfill
 \begin{subfigure}[b]{0.3\textwidth}
 \centering
\begin{tikzpicture}[scale=0.7]
 \node[vertexe] (n0) at (0,0) {};
 \node[vertexf] (n1) at (1,0) {};
 \node[vertexe] (n2) at (0,-2) {};
 \node[vertexf] (n3) at (1,-2) {};
 \path (n0) -- (n2) node [black, font=\LARGE, midway, sloped] {$\ldots$}; 
 \path (n1) -- (n3) node [black, font=\LARGE, midway, sloped] {$\ldots$};

 \node[vertexf] (n4) at (0,-3) {};
 \node[vertexf] (n5) at (1,-3) {};
 \node[vertexf] (n6) at (0,-5) {};
 \node[vertexf] (n7) at (1,-5) {};
 \path (n4) -- (n6) node [black, font=\LARGE, midway, sloped] {$\ldots$}; 
 \path (n5) -- (n7) node [black, font=\LARGE, midway, sloped] {$\ldots$};

\node[vertexe] (n8) at (0,-6) {};
 \node[vertexf] (n9) at (1,-6) {};
 \node[vertexe] (n10) at (0,-8) {};
 \node[vertexf] (n11) at (1,-8) {};
 \path (n8) -- (n10) node [black, font=\LARGE, midway, sloped] {$\ldots$}; 
 \path (n9) -- (n11) node [black, font=\LARGE, midway, sloped] {$\ldots$}; 

\draw [decorate,decoration={brace,amplitude=10pt,raise=5pt},xshift=0pt,yshift=0pt]
(n1) -- (n11) node [black,midway,xshift=1.2cm]{\small $n+k$};

\draw [decorate,decoration={brace,amplitude=10pt,raise=5pt},xshift=0pt,yshift=0pt]
(n6) -- (n4) node [black,midway,xshift=-0.8cm]{\small $n$};

\end{tikzpicture}

\caption{ $k\ge m+1$ }
 \label{fig:abac33}
 \end{subfigure}
 \caption{The abaci of $\bfla_{k,m,n}$.}
 \label{fig:abaci}
\end{figure}
%




\section{Discriminants, root patterns and partitions}
\label{sec:roots}
In this section we give an expression for the discriminant of the generalised Laguerre polynomials and obtain several results {and conjectures} concerning the pattern of 
roots of the generalised Laguerre polynomials in the complex plane. We finish by noting that several of the results can be reframed using partition data.

\subsection{Discriminant of $\Tmn{m,n}{\mu}(z)$}

Recall that a monic polynomial $f(x)$ 
\eq
f(x) = x^d+ a_{d-1} x^{d-1}+ \ldots+ a_1 x+ a_0,
\en
with roots $\a_1,\alpha_2,\ldots, \alpha_{d} \in \CC$ 
has discriminant 
\eq
\text{Dis}(f)=\prod_{1\le j<k \le d} (\alpha_j-\alpha_k)^2.
\en
The discriminants $\text{Dis}_{m,n}(\mu)$ of several $\Tmn{m,n}{\mu}(z)$
are given in 
 Table~\ref{tab:discrim}.

\begin{table}
 \begin{align*}
 \text{Dis}_{1,1}(\mu) &= (\mu+3) \\
 \text{Dis}_{1,2}(\mu) &= (\mu+3)(\mu+4)^4 (\mu+5) /2^4 3^3 \\
 \text{Dis}_{1,3}(\mu) &= (\mu+4)^2(\mu+5)^8 (\mu+6)^4(\mu+7) / 2^{24} 3^8 \\
 \text{Dis}_{2,1}(\mu) &= (\mu+3)(\mu+4)^2/ 2^2 3 \\
 \text{Dis}_{2,2}(\mu) &= -(\mu+3)(\mu+4)^4 (\mu+5)^8(\mu+6)^2 / 2^{24} 3^8 \\
 \text{Dis}_{2,3}(\mu) &= -(\mu+4)^2(\mu+5)^8 (\mu+6)^{16}(\mu+7)^8(\mu+8)^2 / 2^{60} 3^{21} 5^{11} 
 \end{align*} 
\caption{\label{tab:discrim}
 Some discriminants of $\Tmn{m,n}{\mu}(z)$.}
\end{table}

\begin{conjecture}
The discriminant of $\Tmn{m,n}{\mu}(z)$ when $n>m$ is 
\begin{align}
\text{Dis}_{m,n}(\mu)&=
(-1)^{(m+1) \floor{n/2} } c_{m,n}^{2((m+1)n -1)}\prod_{j=1}^m j^{j^3}
\prod_{j=m+1}^n j^{j(m+1)^2} 
\prod_{j=n+1}^{m+n} j^{ j (m+n-j+1)^2} \nn \\
&\qquad \times \prod_{j=1}^{m} j^{2j(n-j)(j-1-m)} 
\prod_{j=1}^{m} (\mu+n+j)^{f(n-1,j) }
\nn \\
&\qquad \times 
\prod_{j=m+1}^{n} (\mu+n+j)^{f(m+n-j, m+1)} 
\prod_{j=n+1}^{m+n} (\mu+n+j)^{f(m,m+n+1-j )},
\label{dis}
\end{align}
and when $n\le m$
\begin{align}
\text{Dis}_{m,n}(\mu)&=
(-1)^{(m+1) \floor{n/2} } c_{m,n}^{2((m+1)n -1)}\prod_{j=1}^n j^{j^3}
\prod_{j=n+1}^m j^{jn^2} 
\prod_{j=m+1}^{m+n} j^{ j (m+n-j+1)^2} \nn \\
&\qquad \times \prod_{j=1}^{n} j^{2j(n-j)(j-1-m)} 
\prod_{j=1}^{n} (\mu+n+j)^{f(n-1,j)}
\nn \\
&\qquad \times 
\prod_{j=n+1}^{m} (\mu+n+j)^{f(j-1,n)} 
\prod_{j=m+1}^{m+n} (\mu+n+j)^{f(m,m+n+1-j) }
\label{dis2}
\end{align}
where 
\eq
f(j,p) =jp^2 -p(p-1)(p-2)/3. 
\en 
\end{conjecture}{Roberts}~\cite{refRoberts} derived formulae 
for the discriminants of the 
Yablonskii-\Vor\ polynomials, 
 the generalised Hermite polynomials and 
the generalised Okamoto polynomials starting from suitable sets of 
differential-difference equations. 
Amdeberhan~\cite{refAmd} applied similar ideas to the Umemura polynomials 
associated with rational solutions of \PIII. It would be interesting to see 
if {Roberts' approach} can be adapted to prove the generalised Laguerre discriminants, possibly starting from the differential-difference equations found in section~\ref{sec:gen_lag}.

\subsection{Roots in the complex plane}

In this section we classify the allowed configuration of roots of $\Tmn{m,n}{\mu}(z)$ in the $z^2$-plane as a function of $\mu$. 
Given the symmetry
 \eqref{symmT}, the root 
plot of $\Tmn{m,n}{\mu}$ when 
 $\mu \in (-m-n-1,\ldots, \infty)$
follows from that of 
$\Tmn{n-1,m+1}{-\mu-2n-2m-2}(\tfrac12z^2)$ rotated by $\tfrac12\pi$.

\begin{example} 
Figure~\ref{fig:T64} shows the roots of $T_{6,4}^{(\mu)}(\tfrac12z^2)$ in the complex plane 
for various $\mu$. 
For $\mu=-35/2$ and $\mu=-6$ 
the non-zero roots form a pair of approximate rectangles of size $5 \times 6$. 
When $\mu=-14$ and $\mu=-8$, there are $24 $ roots at the origin and 
two rectangles of roots of size 
$ 3 \times 6$. 
At $\mu=-17/2$ the roots form two rectangles of size $2\times 6$ (or possibly $3 \times 6$), two 
approximate trapezoids of short base 4 and long base $5$ (or $6$) centered on the real axis and two triangles of size $2$ centred on the imaginary axis. At $\mu=-25/2$ there are four $4$-triangles and two $5 \times 2 $ rectangles.

\begin{figure}\[ \begin{array}{c@{\quad}c@{\quad}c}
\fig{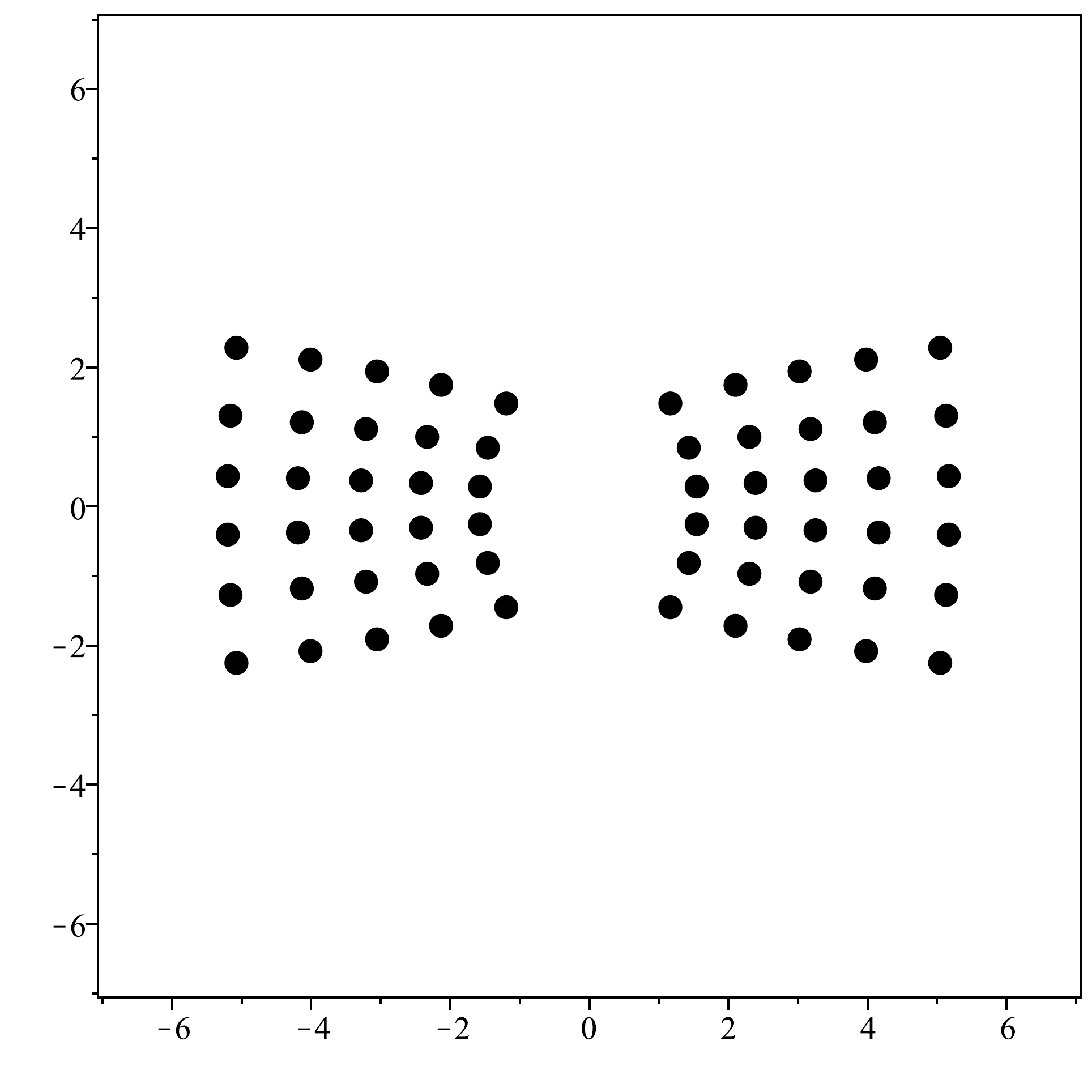} & \fig{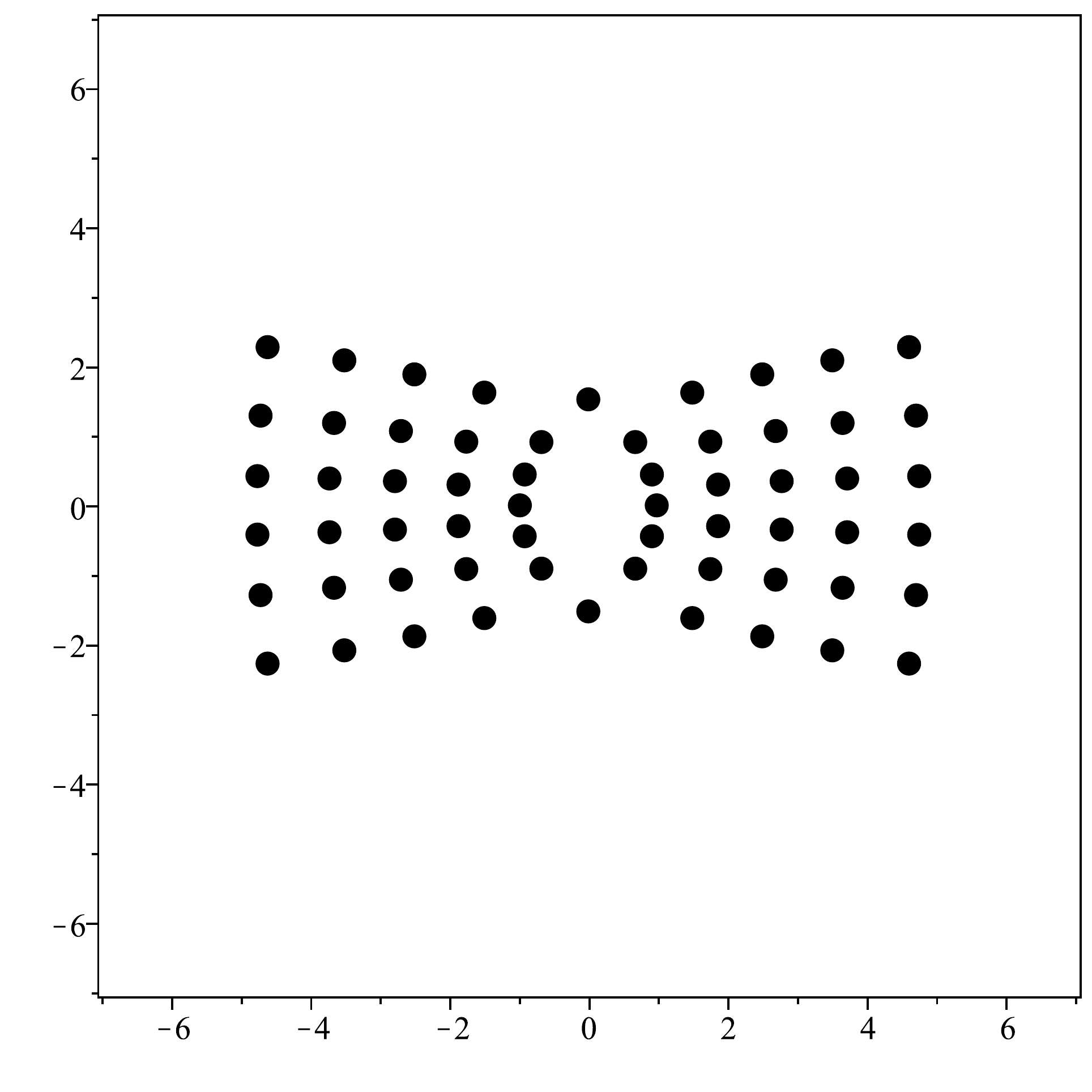} &\fig{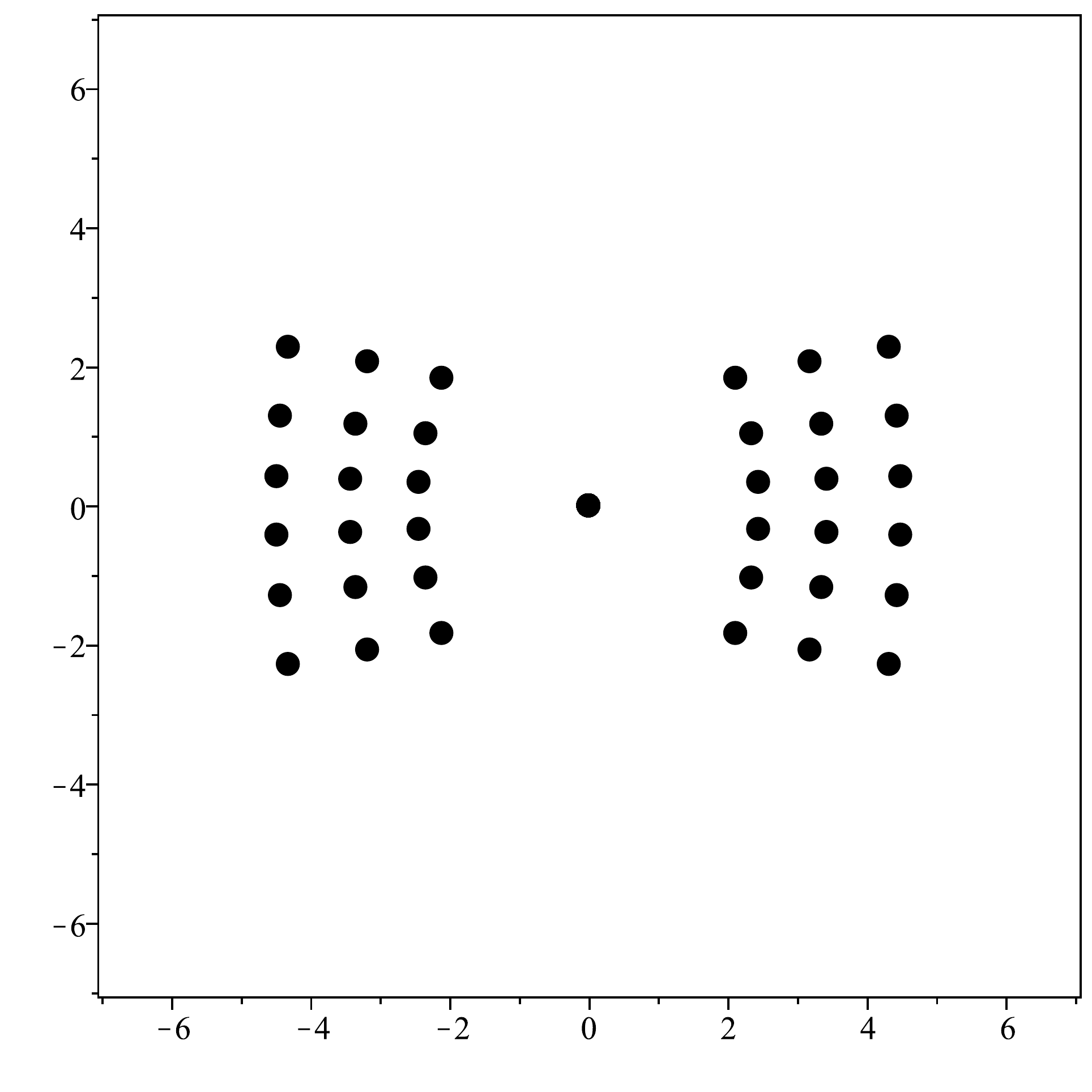}\\
\mu=-6 & \mu=-15/2& \mu=-8\\ 
\fig{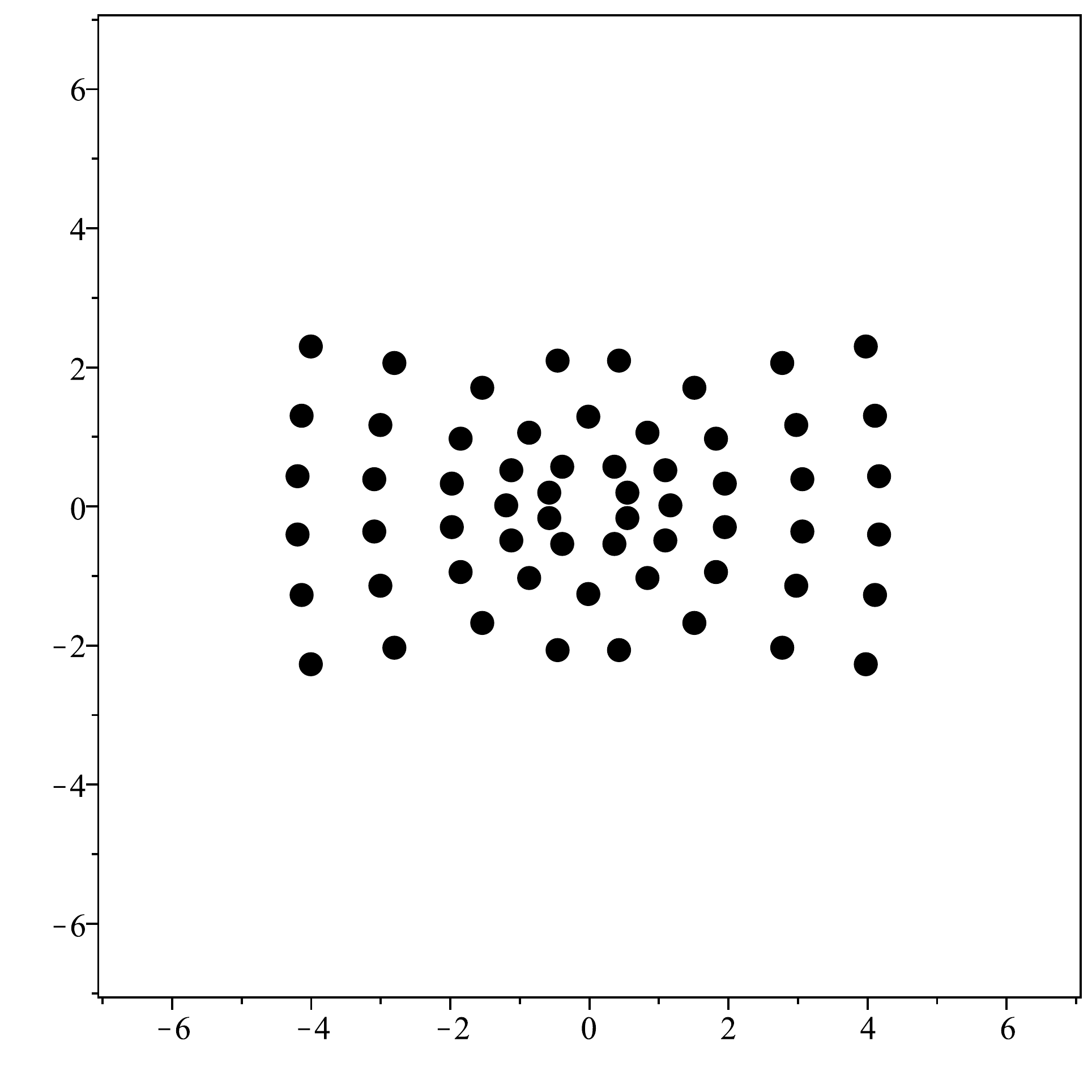} & \fig{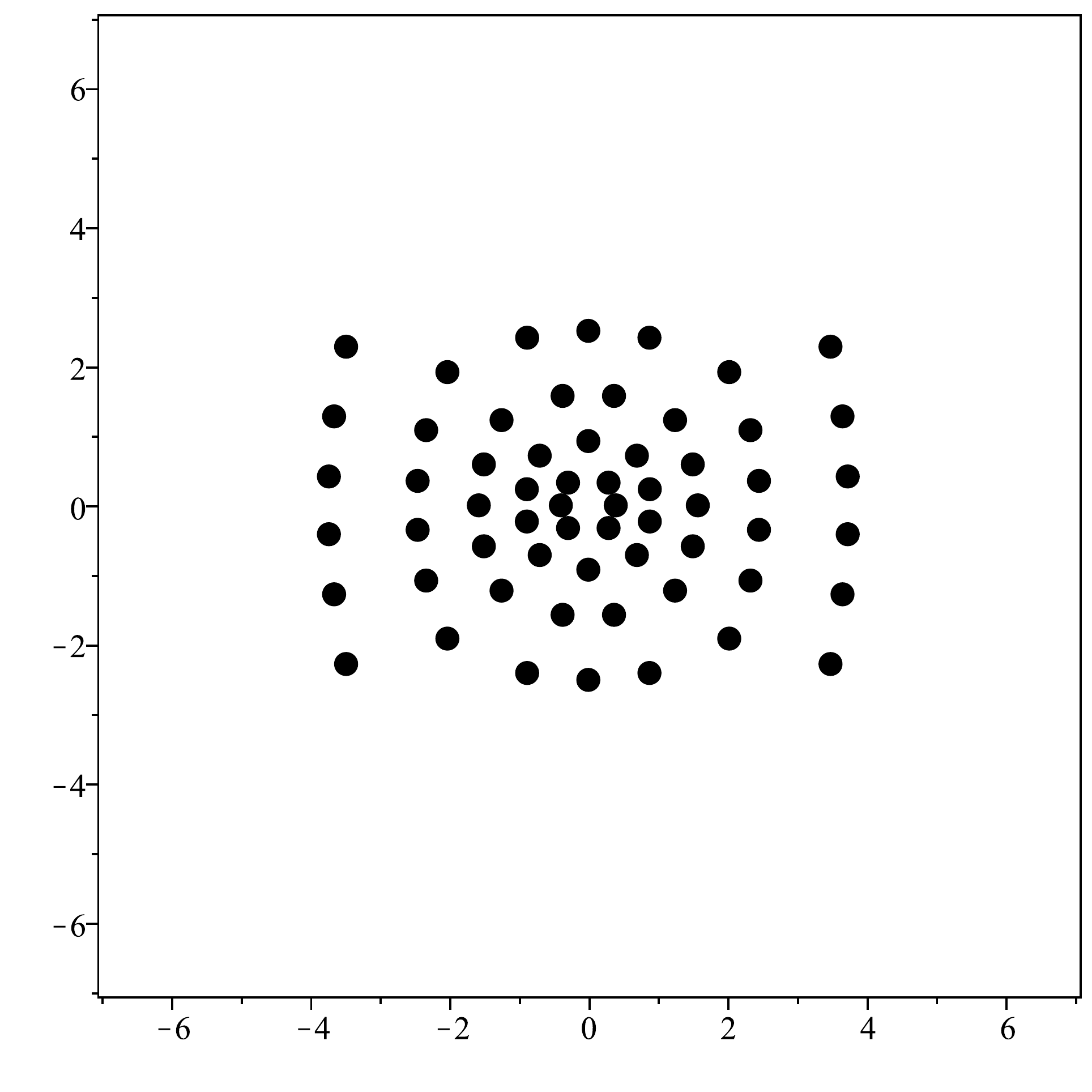} &\fig{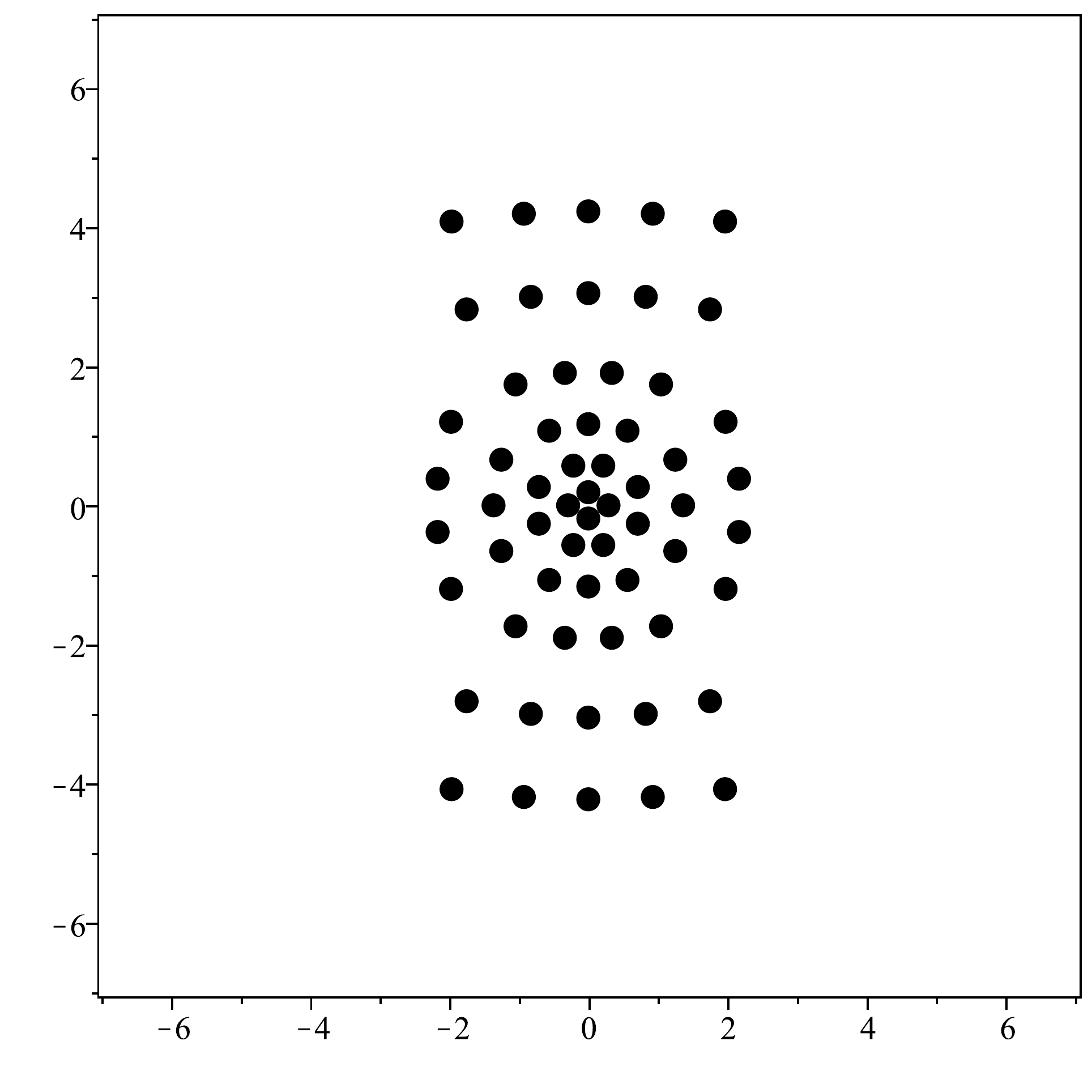}\\
\mu=-17/2 & \mu=-19/2 & \mu= -25/2\\ 
\fig{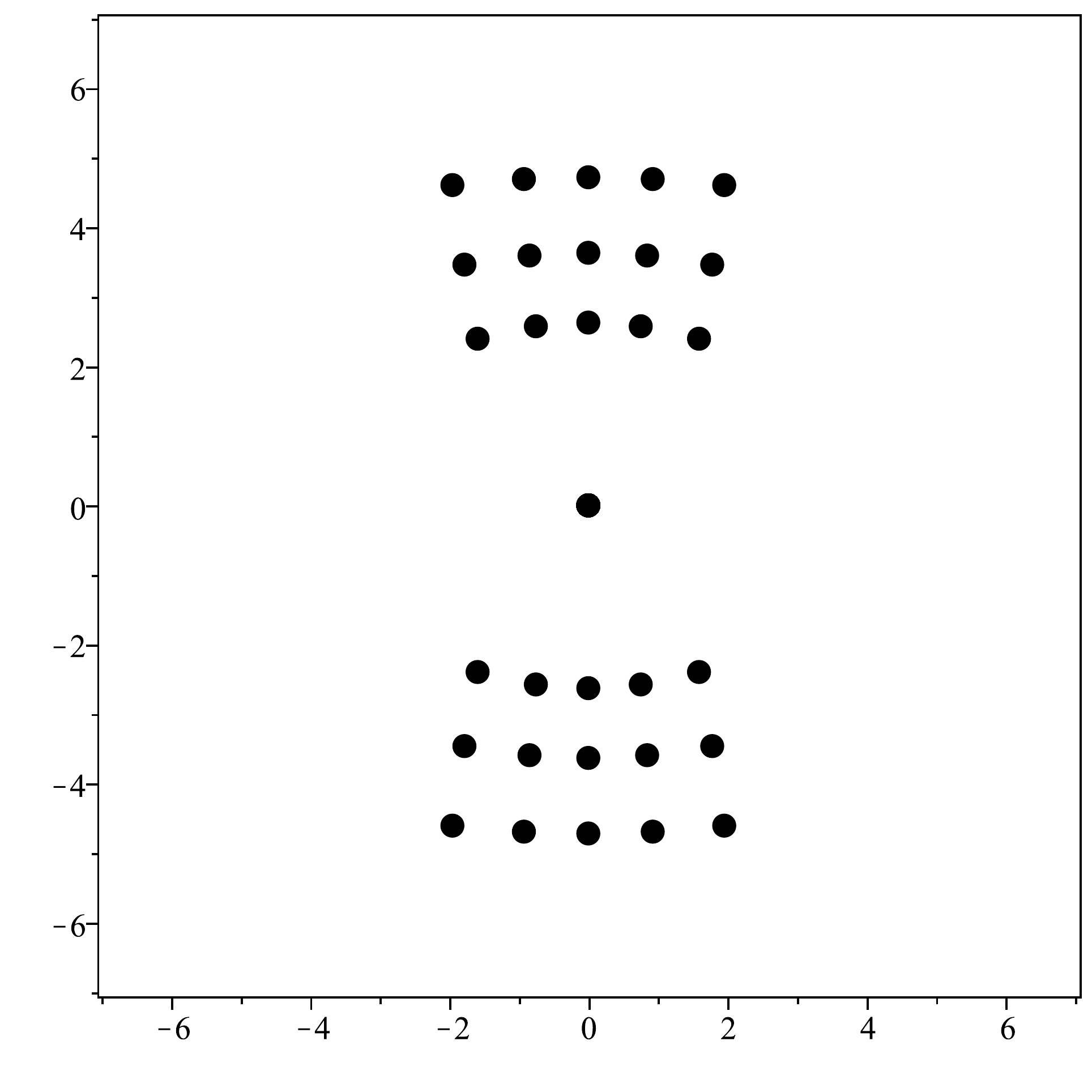} & \fig{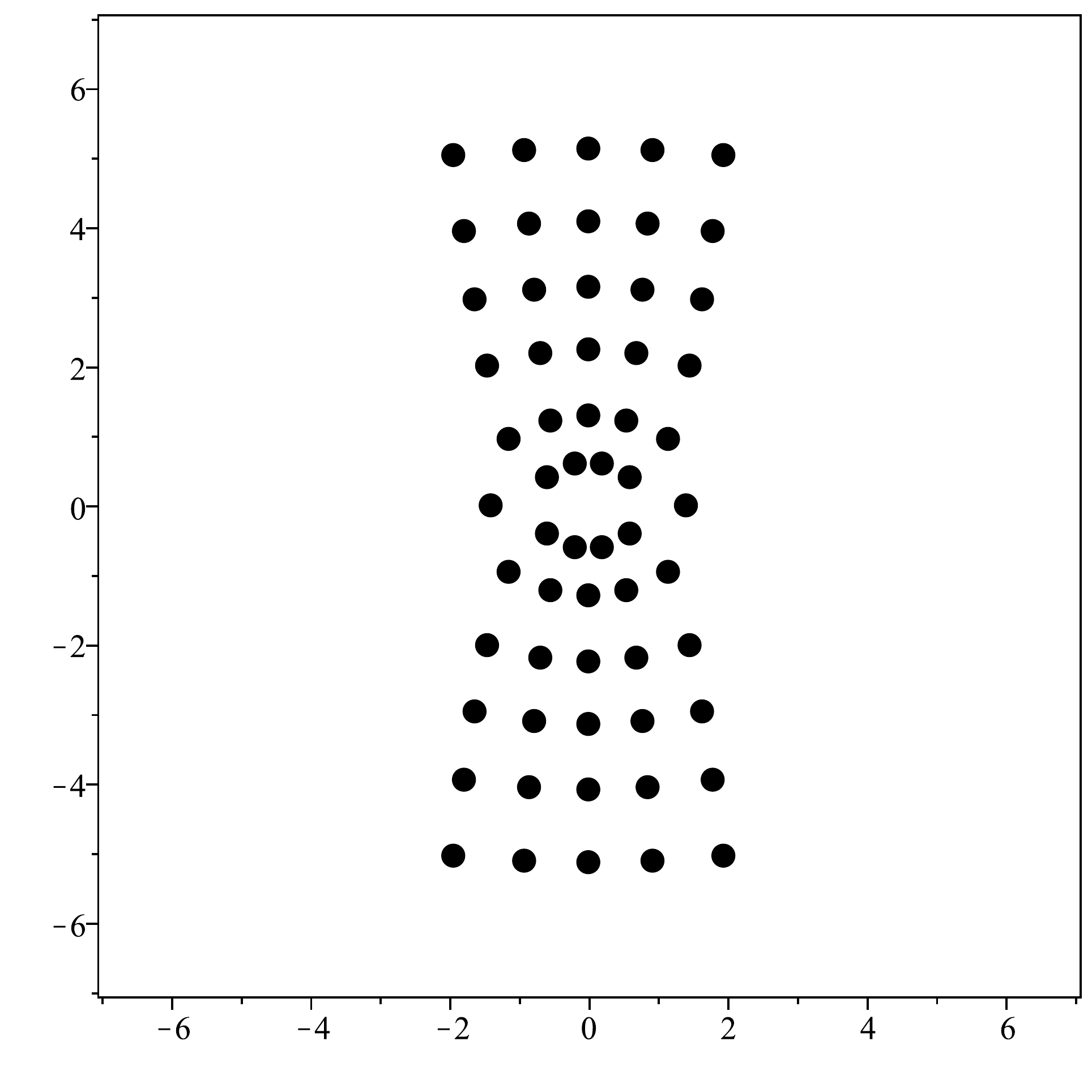} &\fig{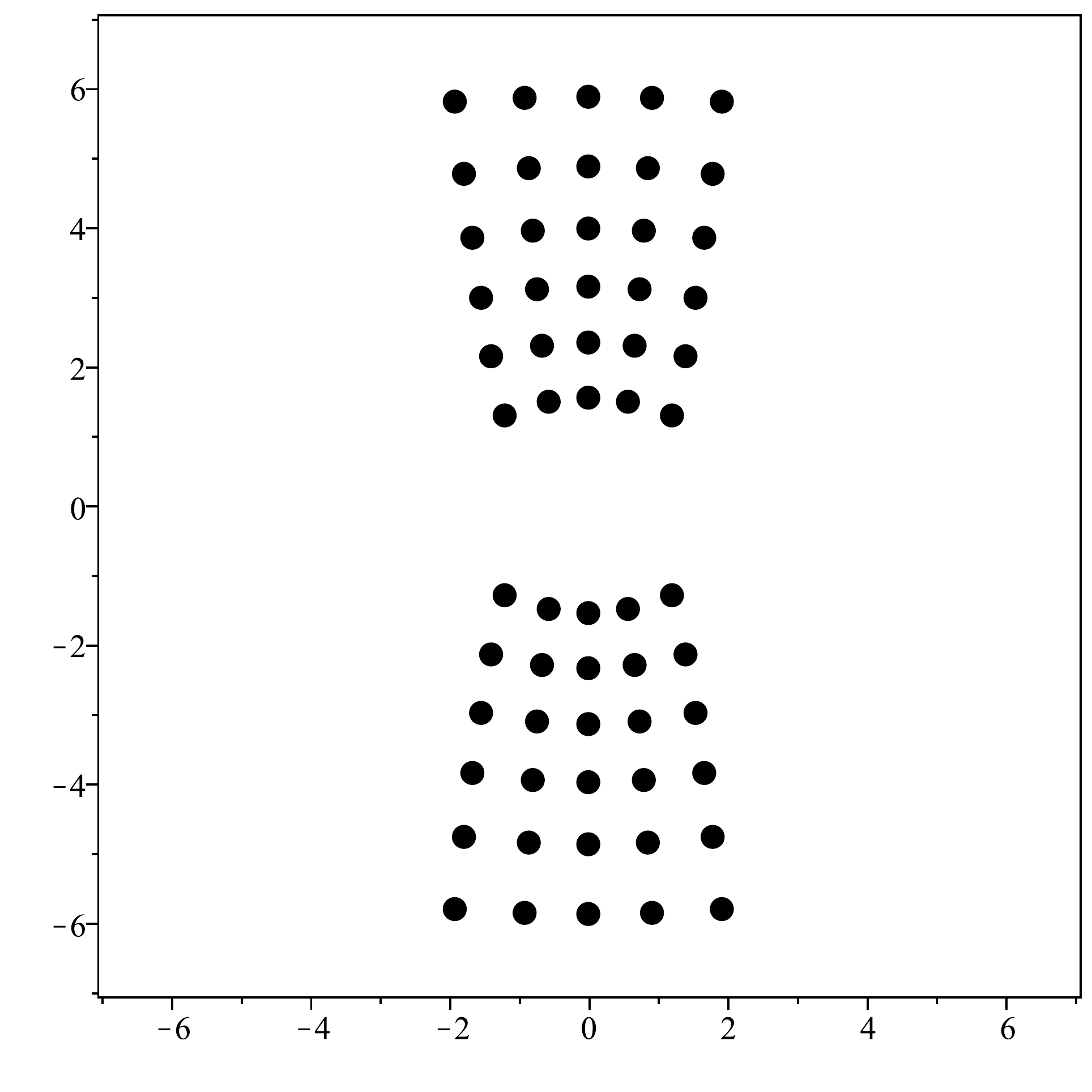}\\
\mu=-14 & \mu=-31/2 & \mu=-35/2\\ 
\end{array}\]
\caption{\label{fig:T64}The roots of
$T_{6,4}^{(\mu)}(\tfrac12z^2)$ for various $\mu$.}
\end{figure}

\end{example}

Further investigations suggest that 
 the roots of $T_{m,n}^{(\mu)}(\tfrac12z^2)$ that are away from the origin 
 form blocks in the form of approximate trapezoids and/or triangles
 near the origin and 
rectangles further away. We label such blocks E--G
as shown in Figure~\ref{fig:blocks}. 
 We 
 say a rectangle has size $d_1 \times d_2$ if it
has width $d_1$ and height $d_2$. 
A 
 trapezoid of size $d_1 \times d_2$ has long base $d_1$ and
 short base $d_2$. If $d_2=1$ then we call the resulting (degenerate)
 trapezoid a triangle.  {The blocks of roots centered on the real or imaginary axis in approximate rectangles are labelled blocks E and D respectively, and those forming approximate trapezoids are labelled G and F respectively.}
\begin{figure}
\centering 
\begin{subfloat}[All blocks] {
\begin{tikzpicture}[scale=0.75]
 
\draw[black, thick] (-3.5,-1.6) -- (-2,-1.6) -- (-2,1.6) -- (-3.5,1.6) --cycle;

\draw[black, thick] (3.5,-1.6) -- (2,-1.6) -- (2,1.6) -- (3.5,1.6) --cycle;

\draw[black, thick] (-1.3,2) -- (1.3,2) -- (1.3,3) -- (-1.3,3) --cycle; 

\draw[black, thick] (-1.3,-2) -- (1.3,-2) -- (1.3,-3) -- (-1.3,-3) --cycle; 

\draw (0,2.5) node {D} ;
\draw (0,-2.5) node {D} ;

\draw (-2.75,0) node {E} ;
\draw (2.75,0) node {E} ;

\draw[black, thick] (0.3,0.3) -- (1,1.6) -- (-1,1.6) --
(-0.3,0.3)--cycle; 

\draw[black, thick] (0.3,-0.3) -- (1,-1.6) -- (-1,-1.6) --
(-0.3,-0.3)--cycle; 

\draw (0,0.95) node {F} ;
\draw (0,-0.95) node {F} ;

\draw[black, thick] (0.6,0.3) -- (1.5,1.4) -- (1.5,-1.4) --
(0.6,-0.3)--cycle;

\draw[black, thick] (-0.6,0.3) -- (-1.5,1.4) -- (-1.5,-1.4) --
(-0.6,-0.3)--cycle;

\draw (1.05,0) node {G} ;
\draw (-1.05,0) node {G} ;

\end{tikzpicture}
}
\end{subfloat}
\begin{subfloat}[$\Tmn{5,8}{-57/5}(\tfrac12z^2) $\label{root_colour1}] {
\includegraphics[width=.3\linewidth]{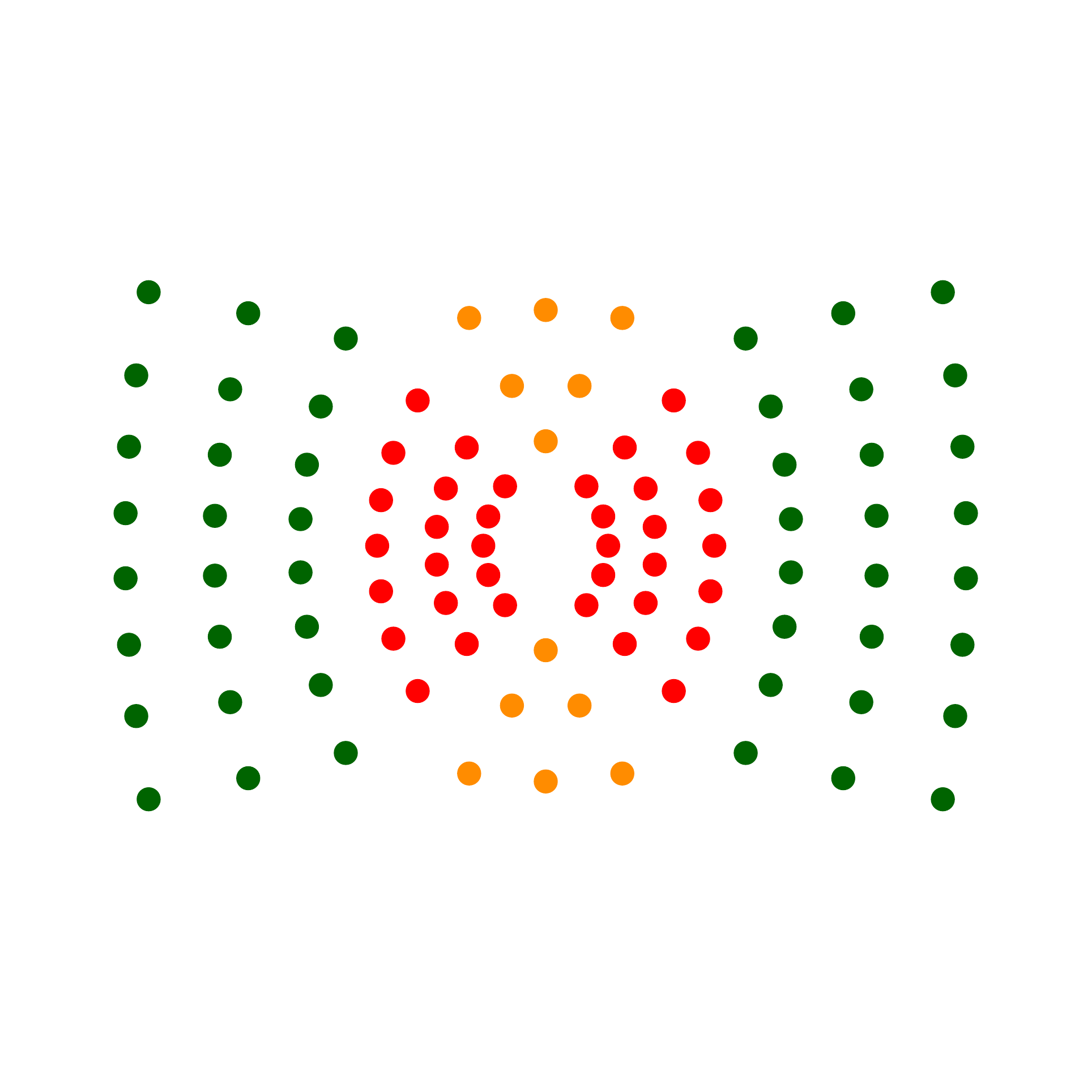}
%
%
}
\end{subfloat}
\begin{subfloat}[ $\Tmn{5,8}{-313/20}(\tfrac12z^2)$ \label{root_colour2}] {
%
\includegraphics[width=.3\linewidth]{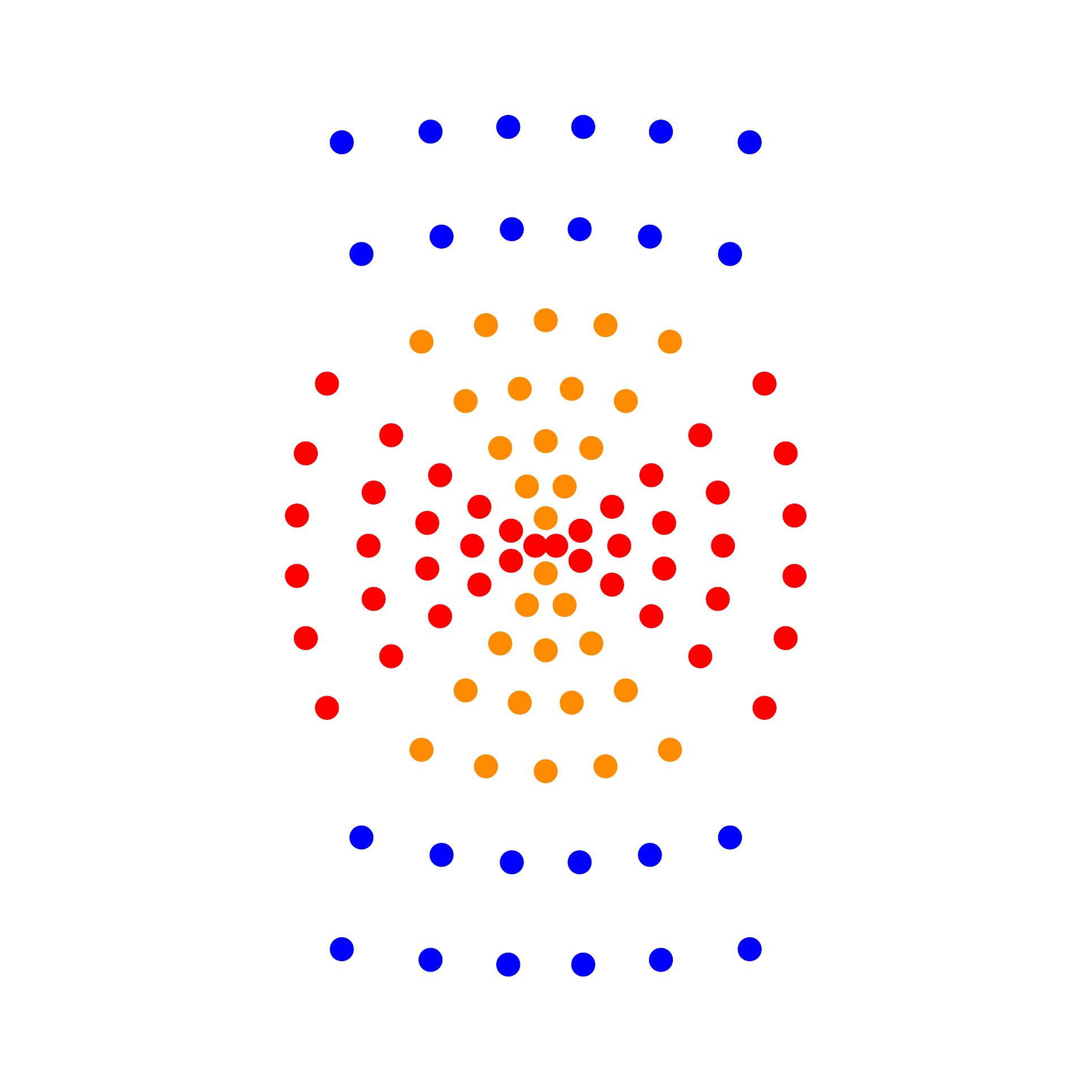}
%
}
\end{subfloat}
\caption{\label{fig:blocks}Blocks formed by the zeros of $T_{m,n}^{(\mu)}(\tfrac12z^2)$.}
\end{figure}
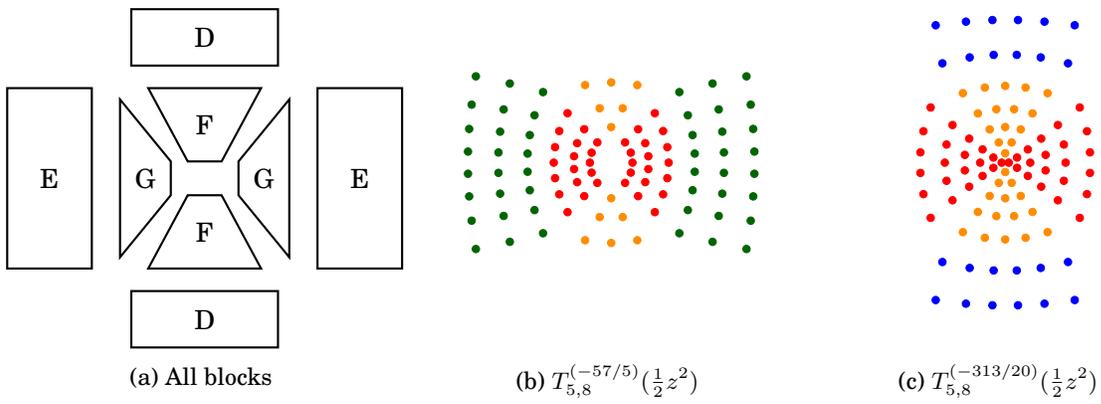
Figures~\ref{root_colour1} and \ref{root_colour2}
 show the zeros of $\Tmn{5,8}{-57/5}(\tfrac12z^2)$ and
$\Tmn{5,8}{-323/20}(\tfrac12z^2)$ with 
 block E zeros in green, block G in red, block F in orange and block D in blue. 

 We describe how the roots transition between blocks as a function of $\mu$ and determine 
the size of each root block for a given $\mu$ when $m=5$ and $n=3$,
 before stating the result for all $m,n$.
\begin{figure}[ht]\[ \begin{array}{c@{\quad}c@{\quad}c}
\fig{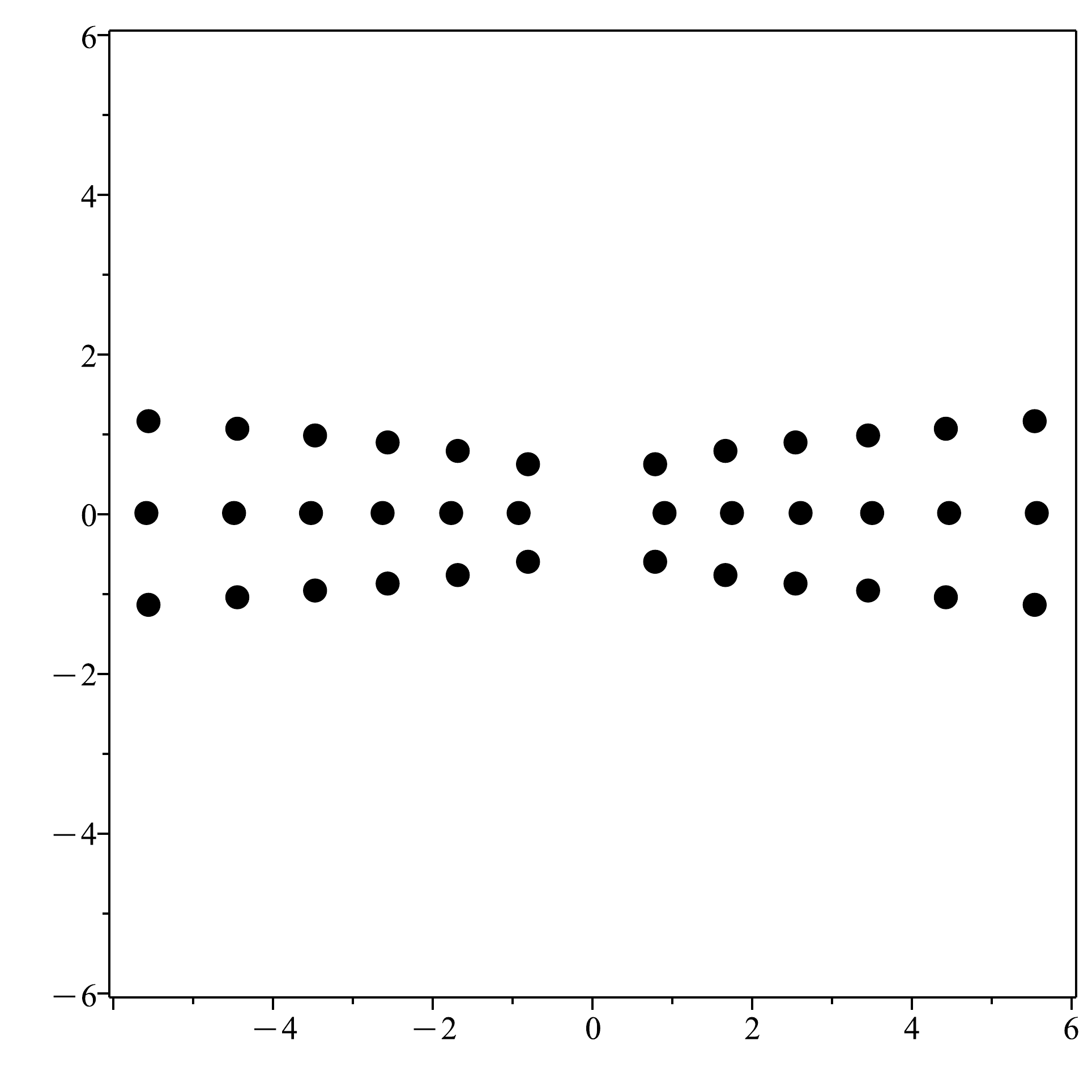} & \fig{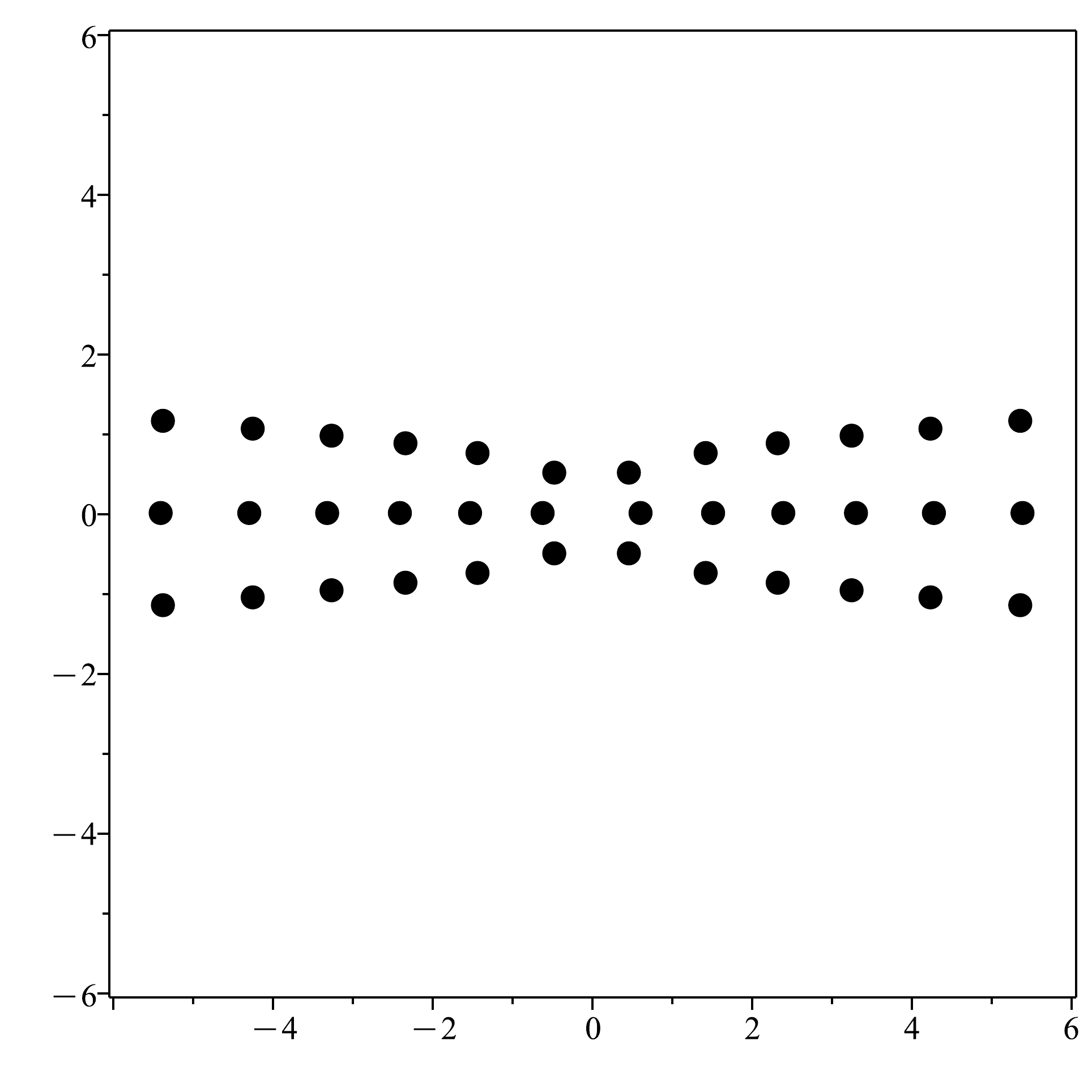} &\fig{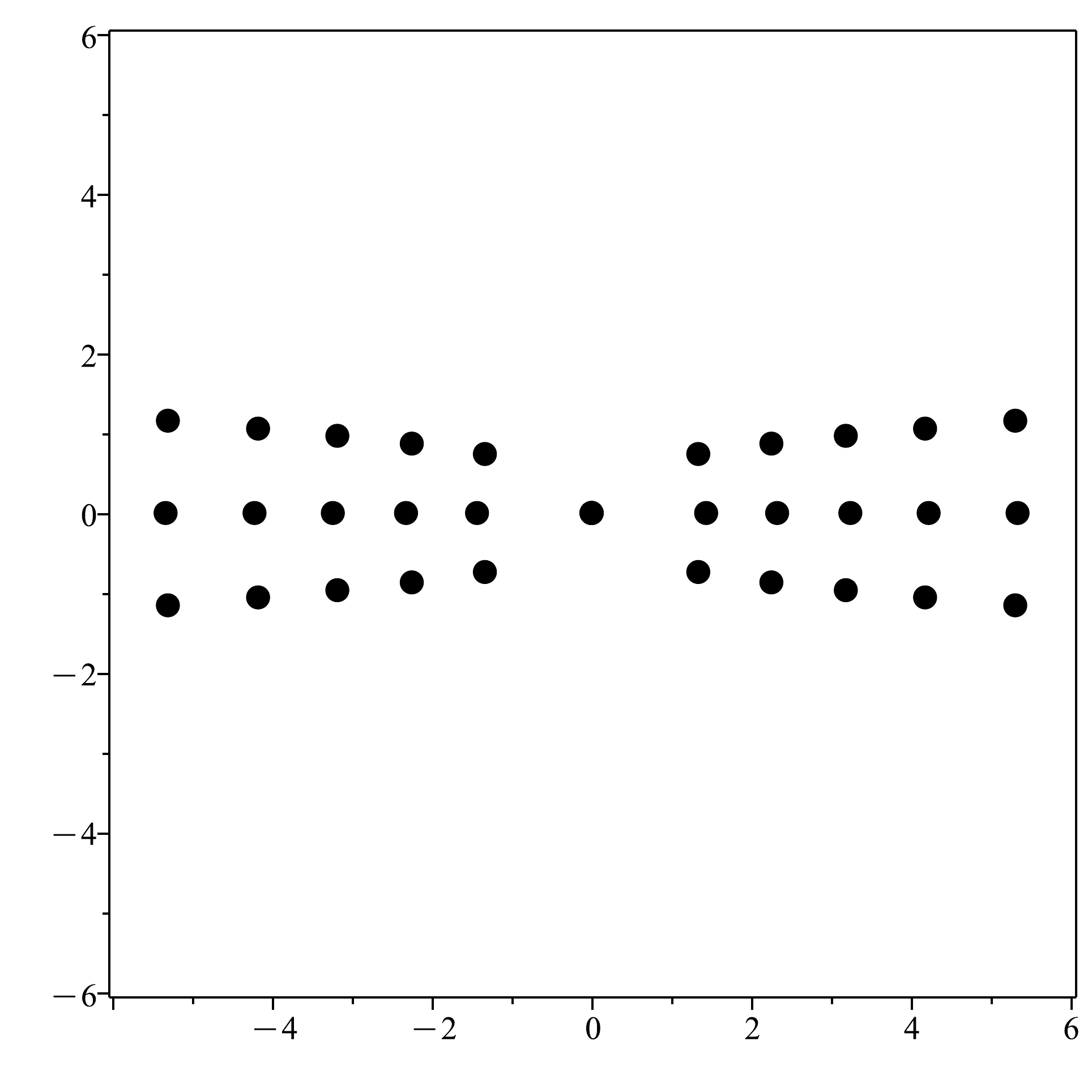}\\
\mu=-\fract{16}{5} & \mu=-\fract{19}{5} & \mu=-4\\
\fig{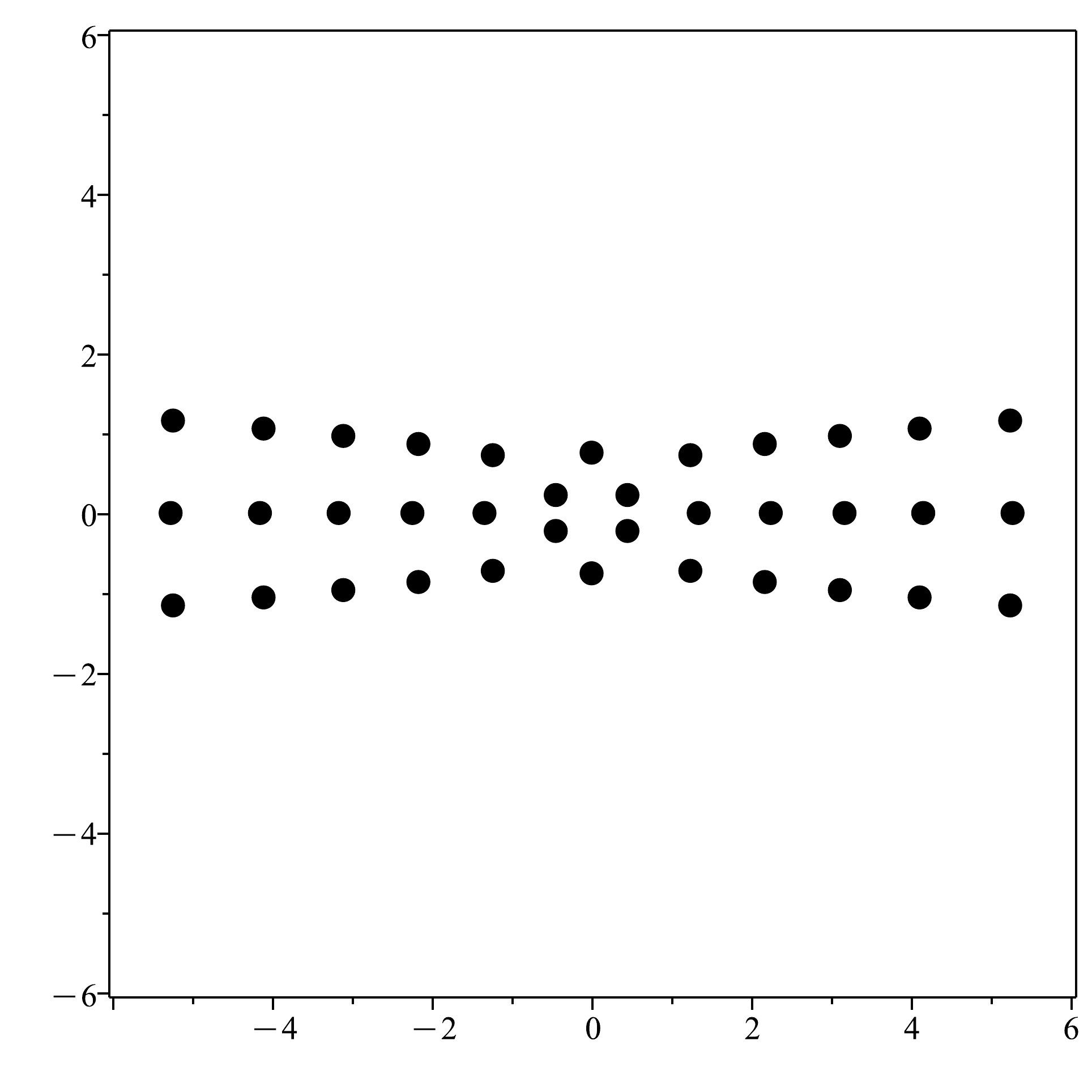} & \fig{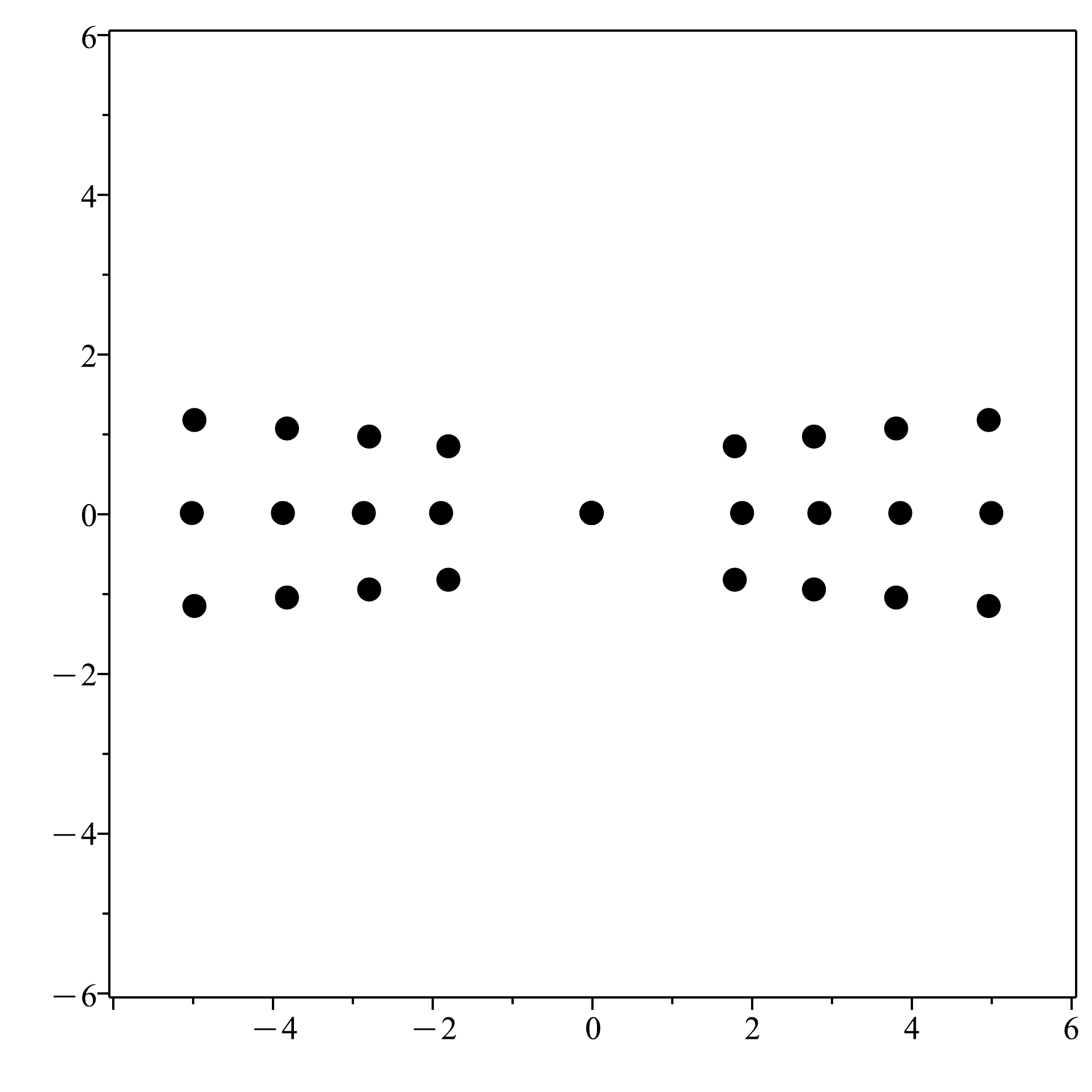} &\fig{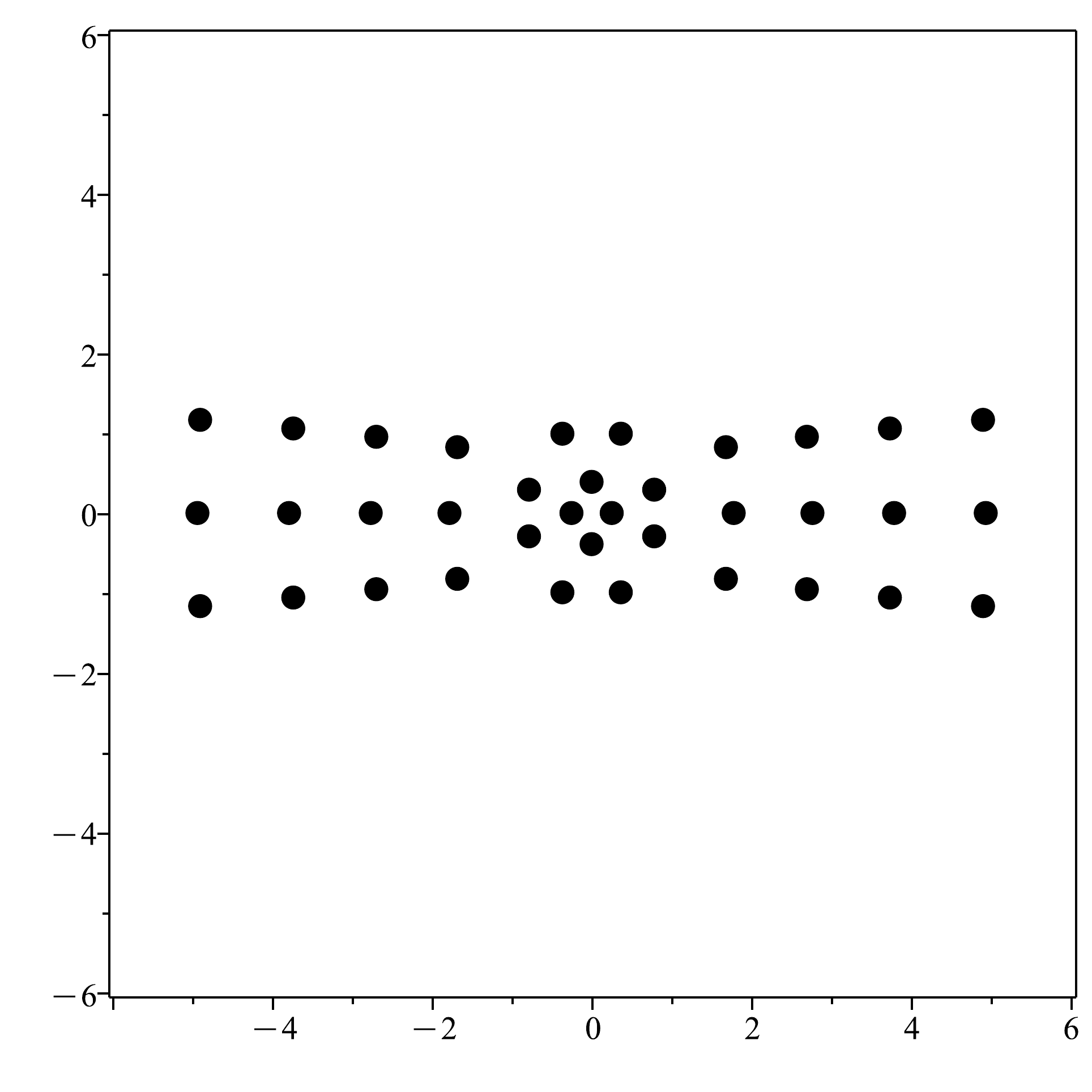}\\
\mu=-\fract{21}{5} & \mu=-5 & \mu=-\fract{26}{5}\\
\fig{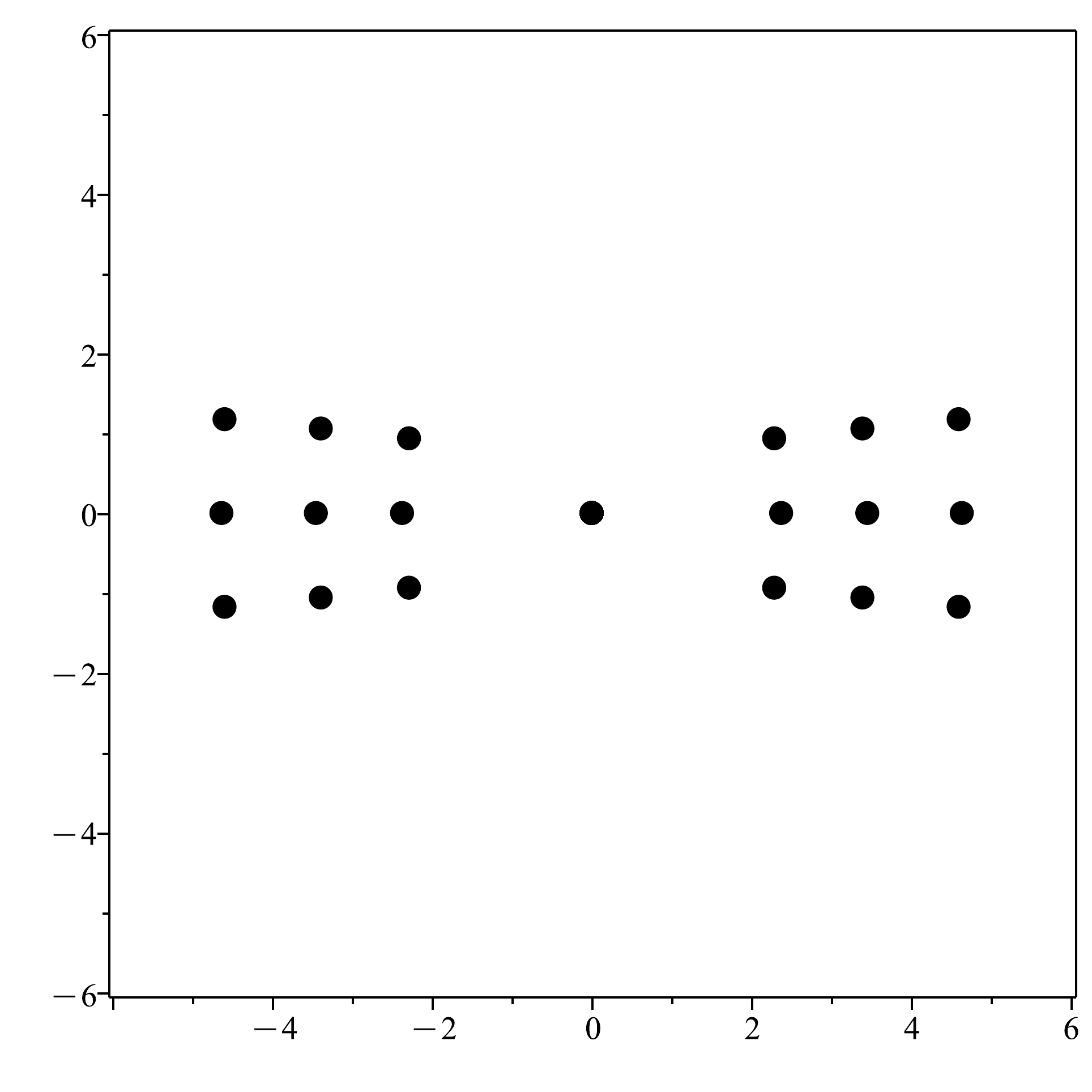} & \fig{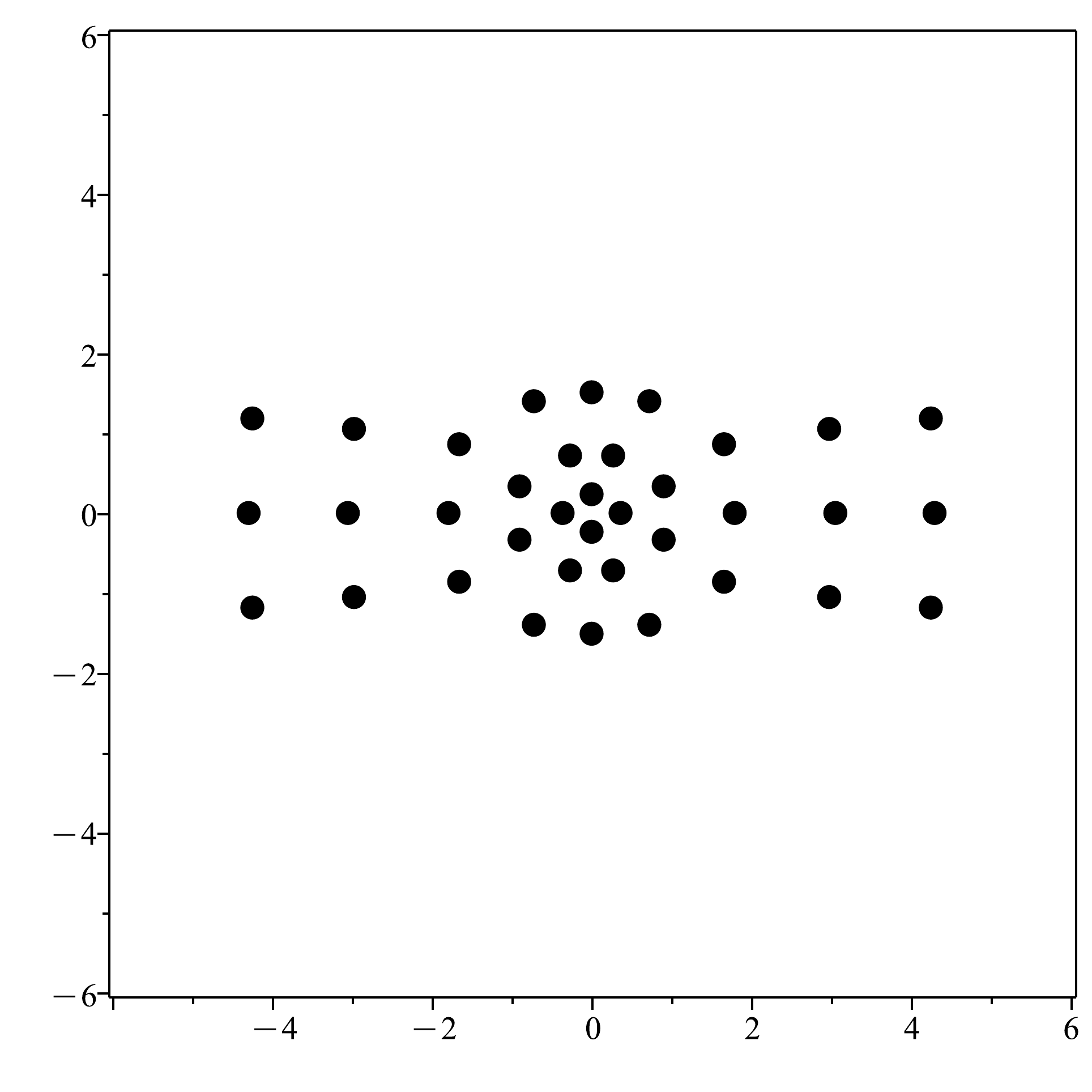} &\fig{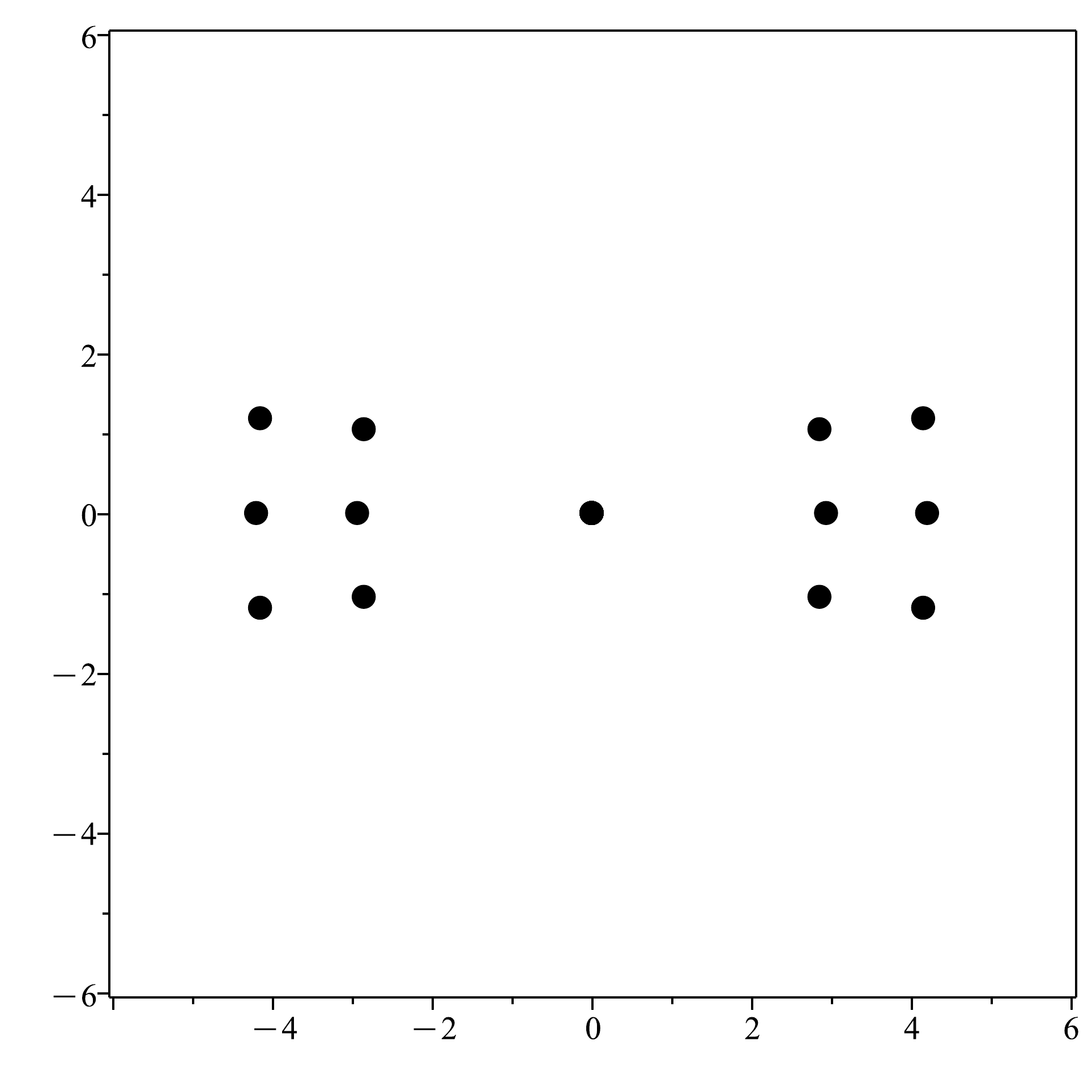}\\
\mu=-6 & \mu=-\fract{34}{5} & \mu=-7\\
\end{array}\]
\caption{\label{fig:T531}The roots of
$T_{5,3}^{(\mu)}(\tfrac12z^2)$ for $\mu \in [-7, -\fract{16}{5} ]. $}
\end{figure}

\begin{figure}[ht]\[ \begin{array}{c@{\quad}c@{\quad}c}
\fig{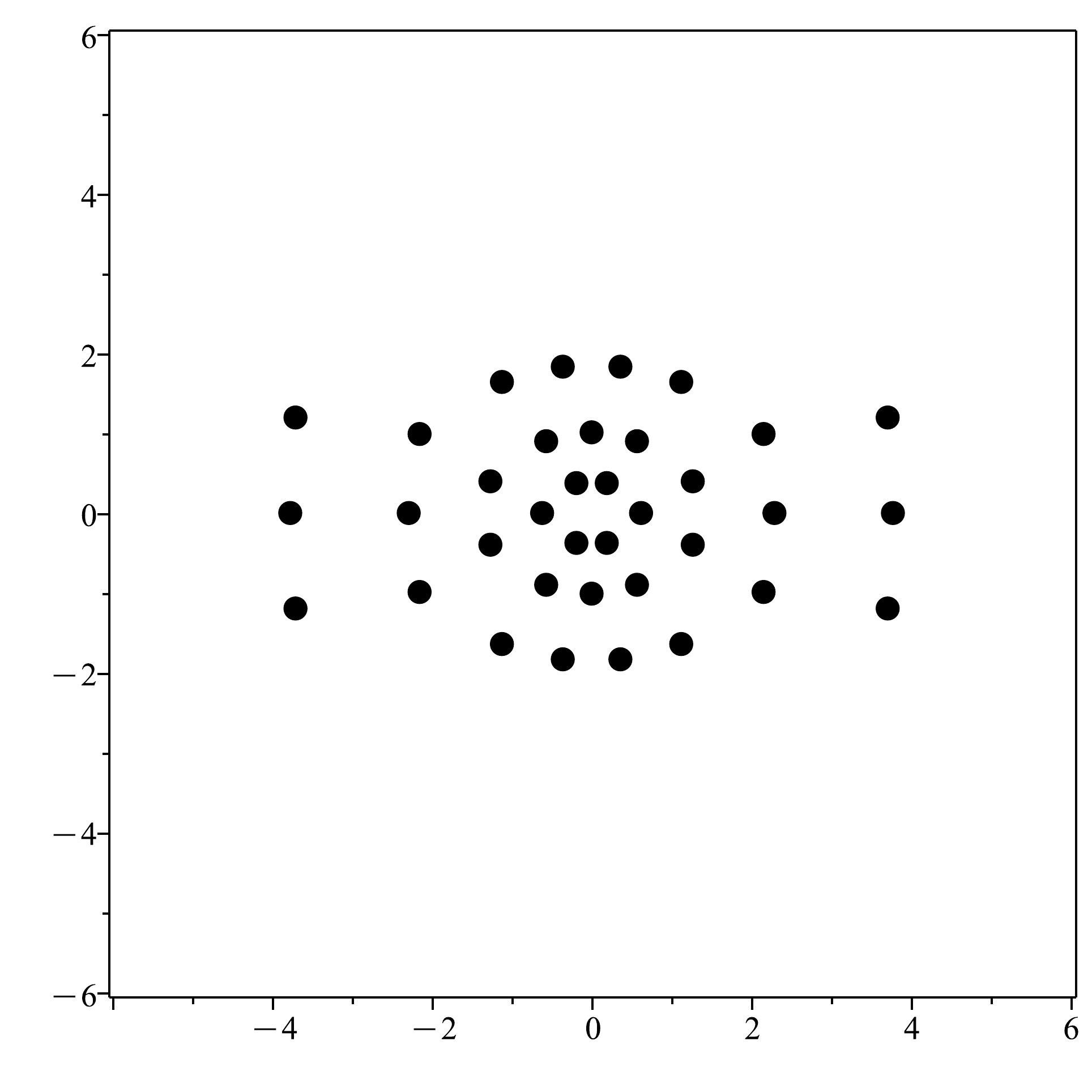} & \fig{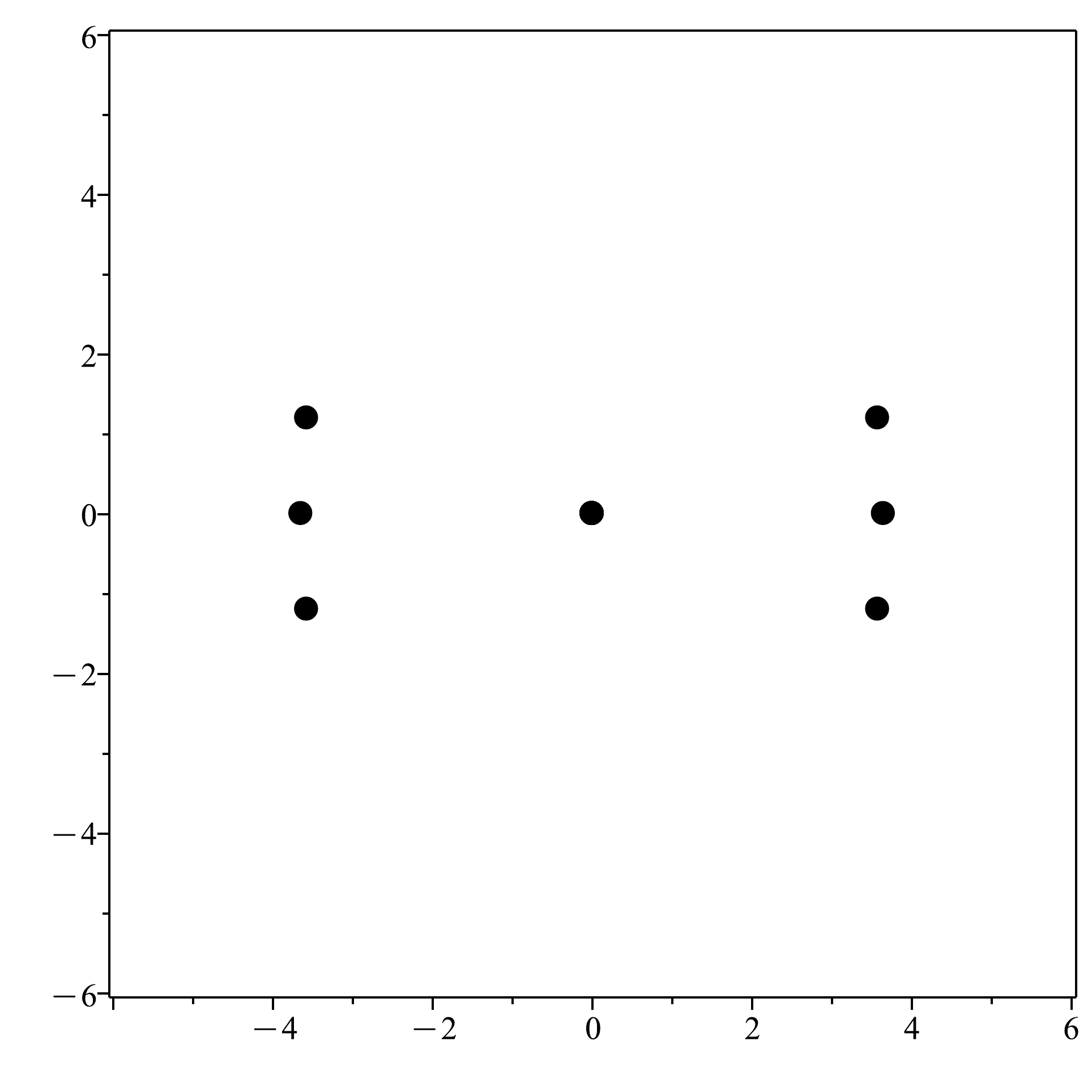}& \fig{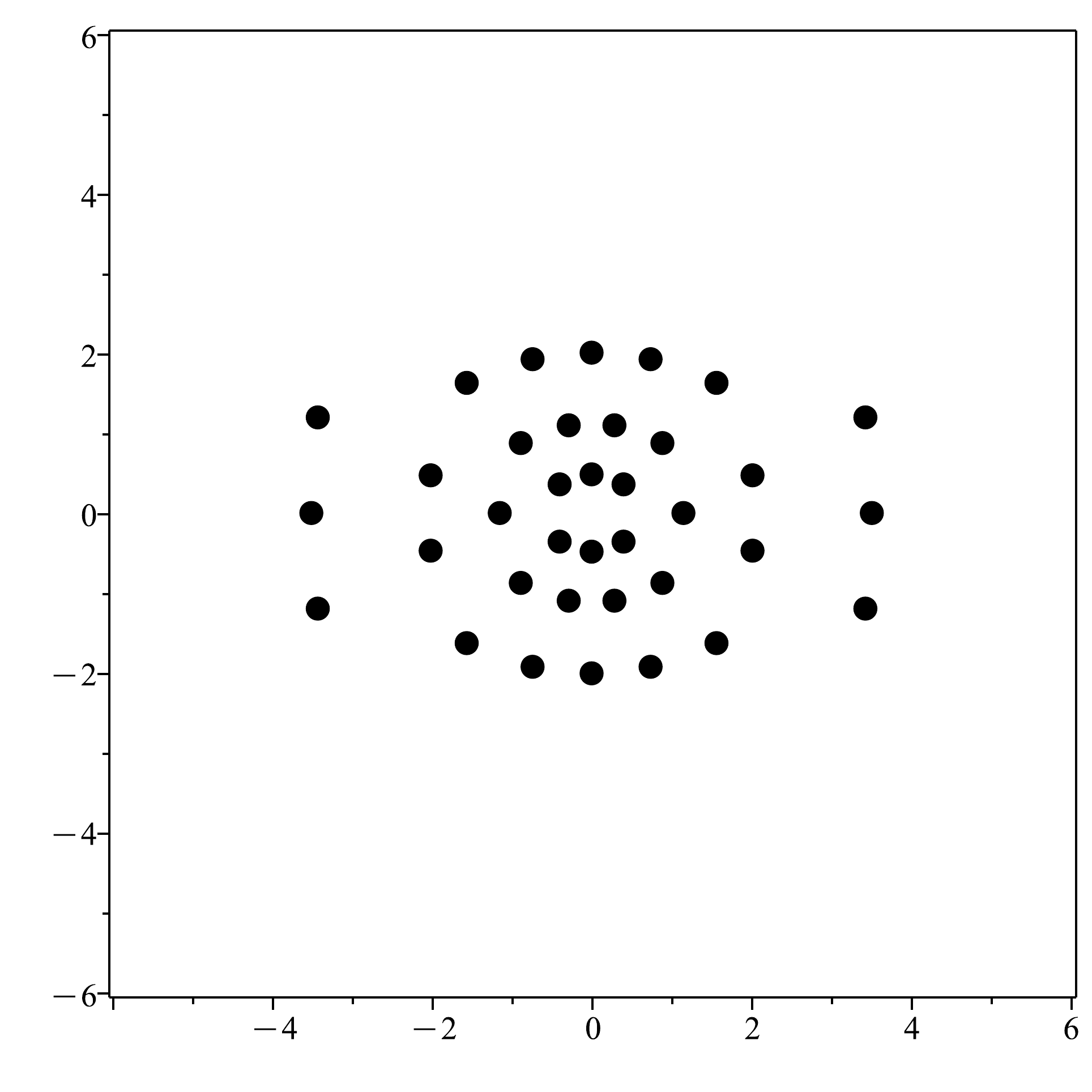}\\
\mu=-\fract{39}{5} & \mu=-8& \mu=-\fract{41}{5}\\
\fig{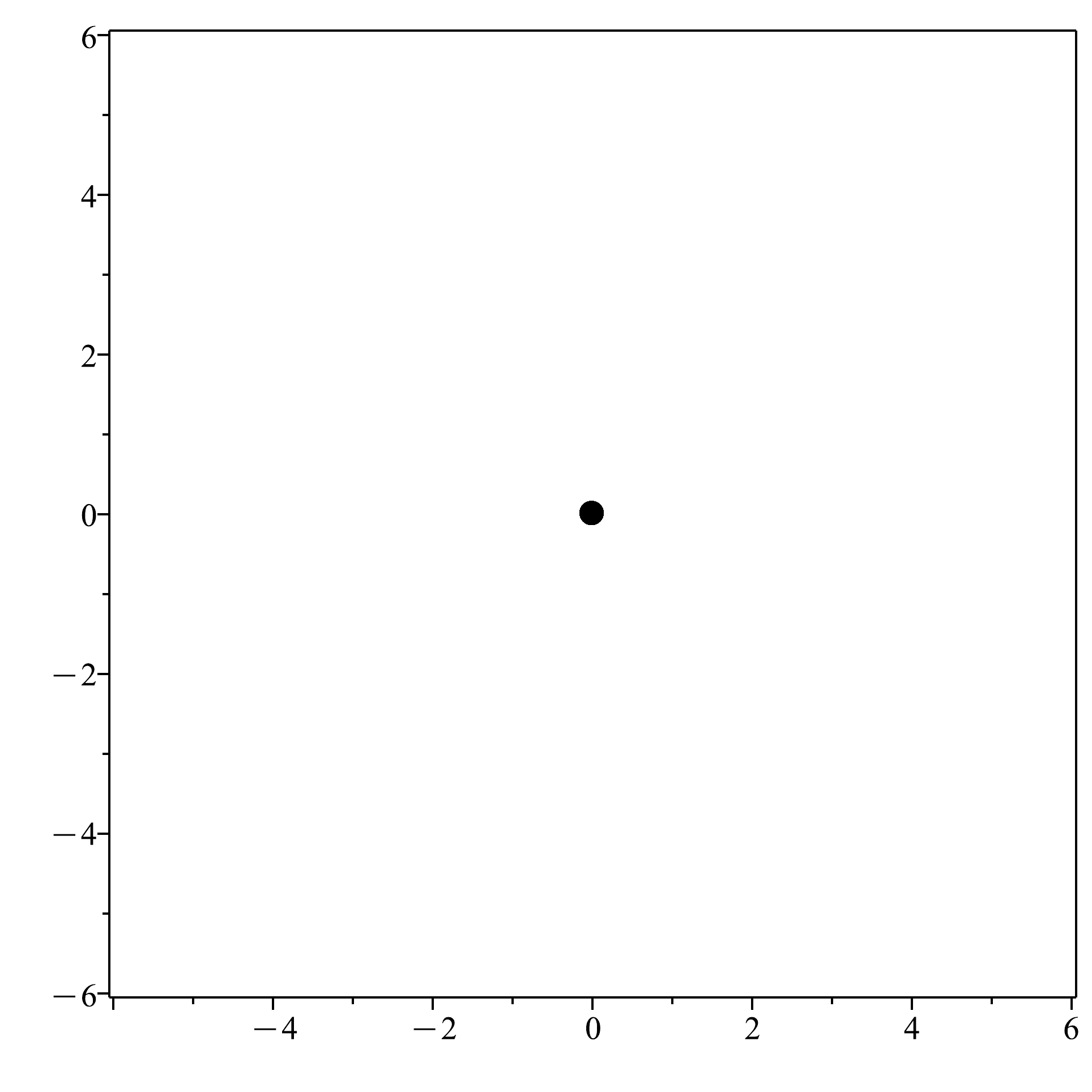} & \fig{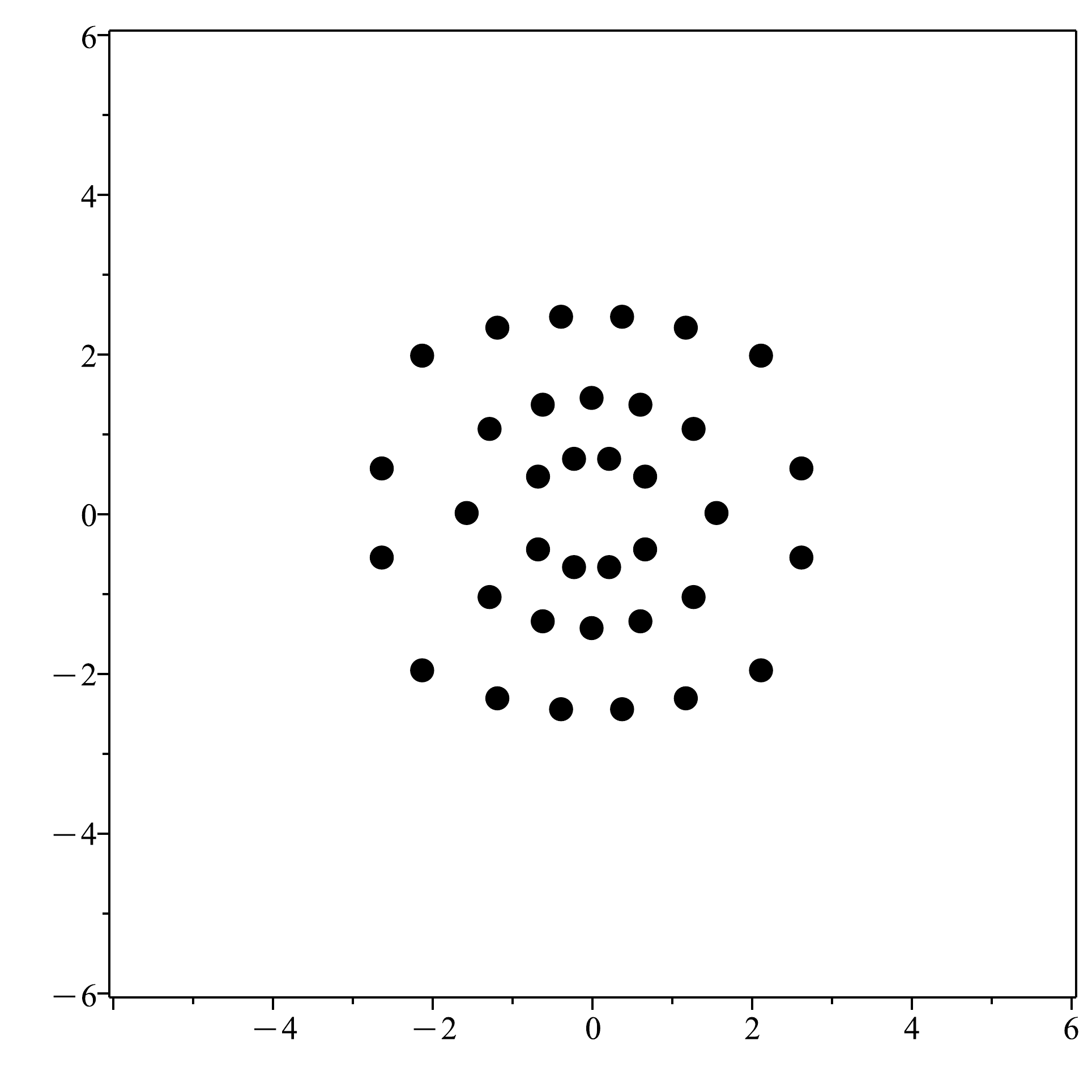}& \fig{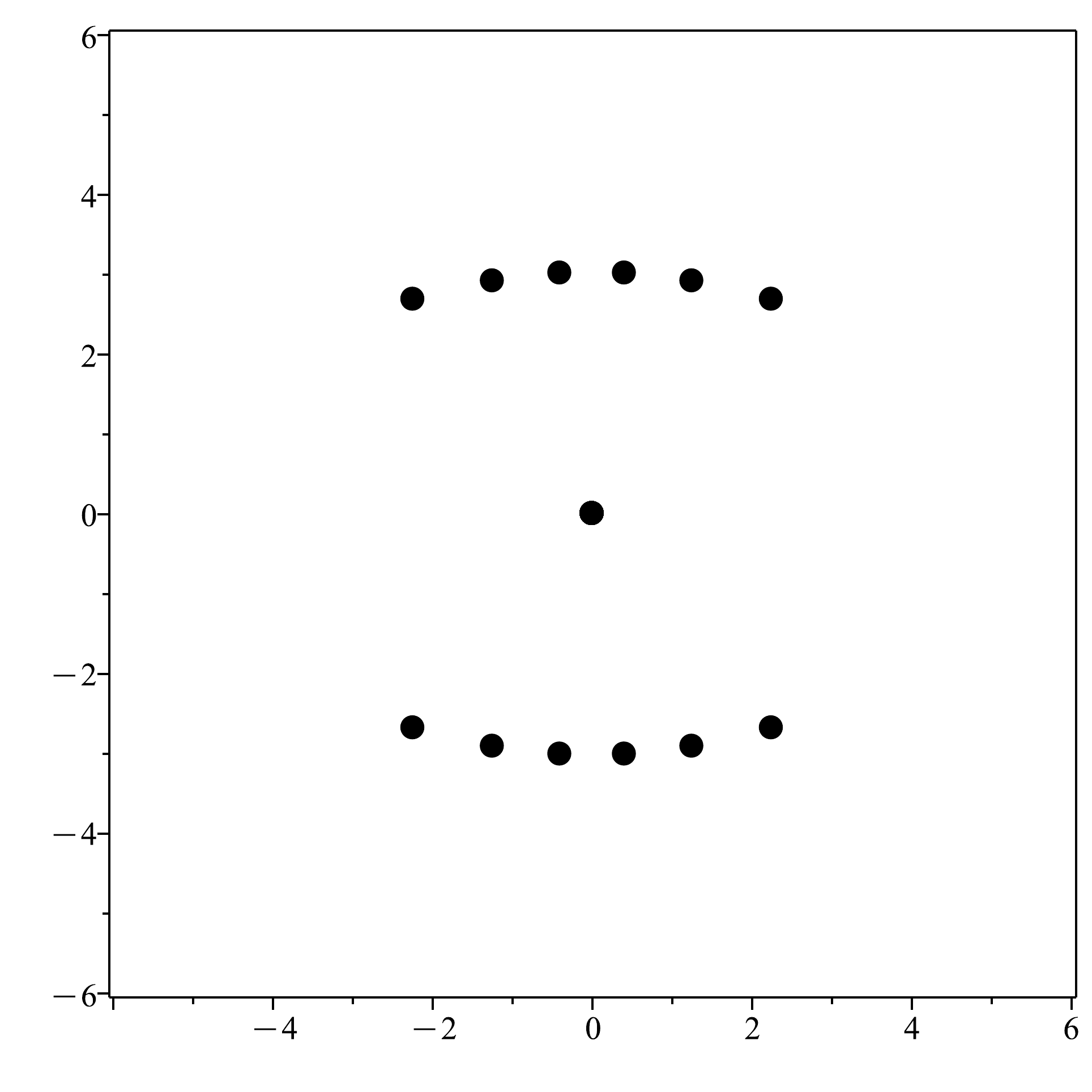}\\
\mu=-9 & \mu=-\fract{46}{5}& \mu=-10\\
\fig{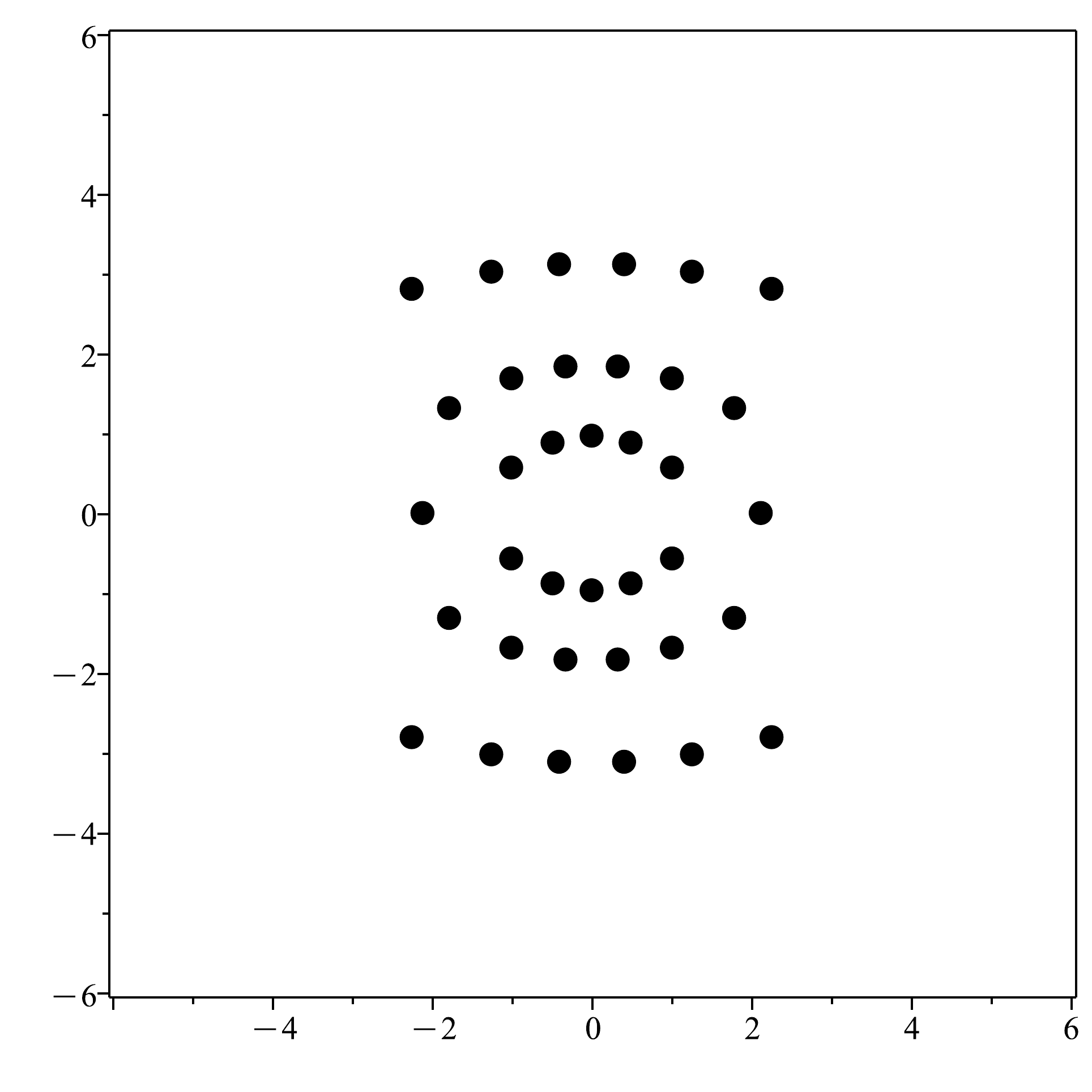} & \fig{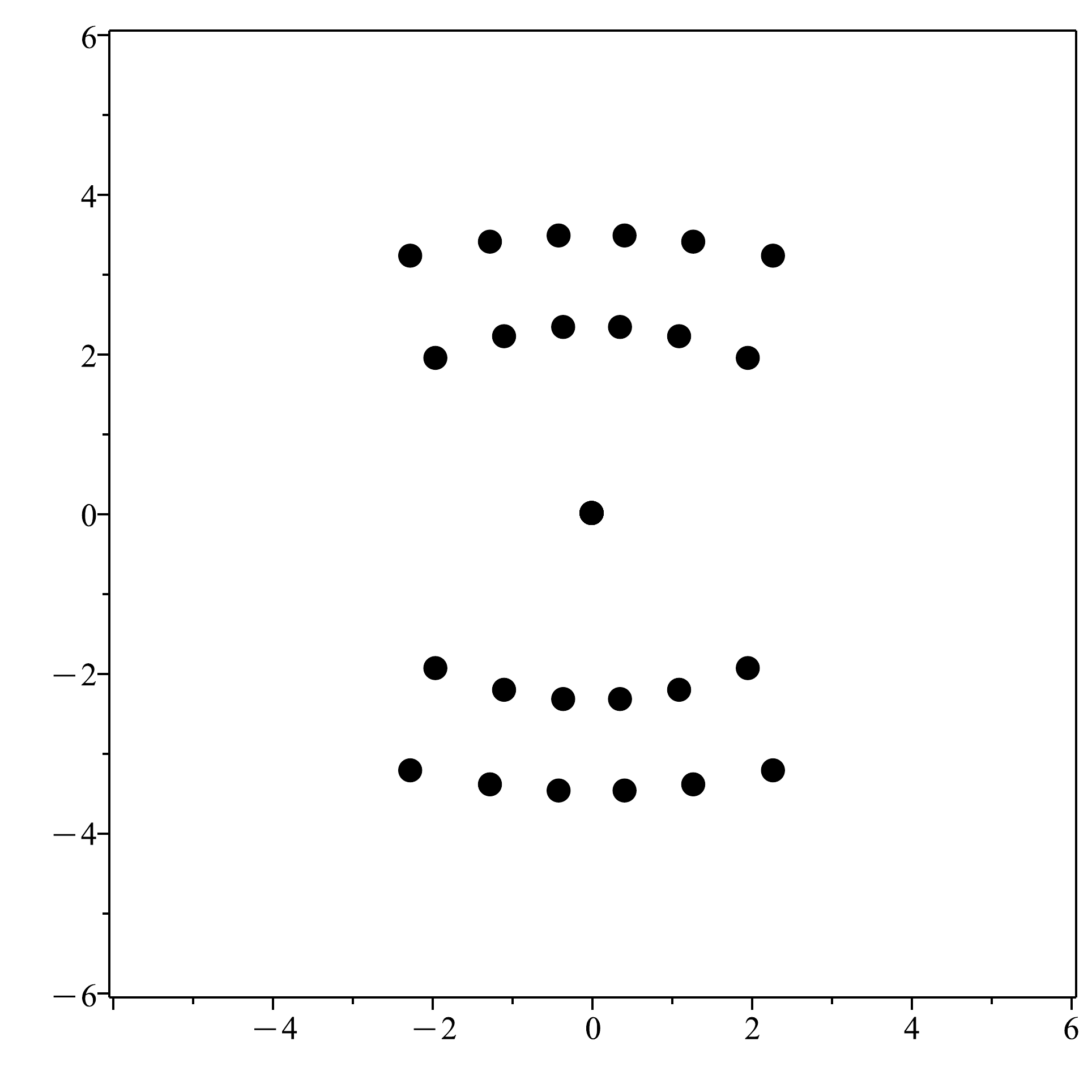} &\fig{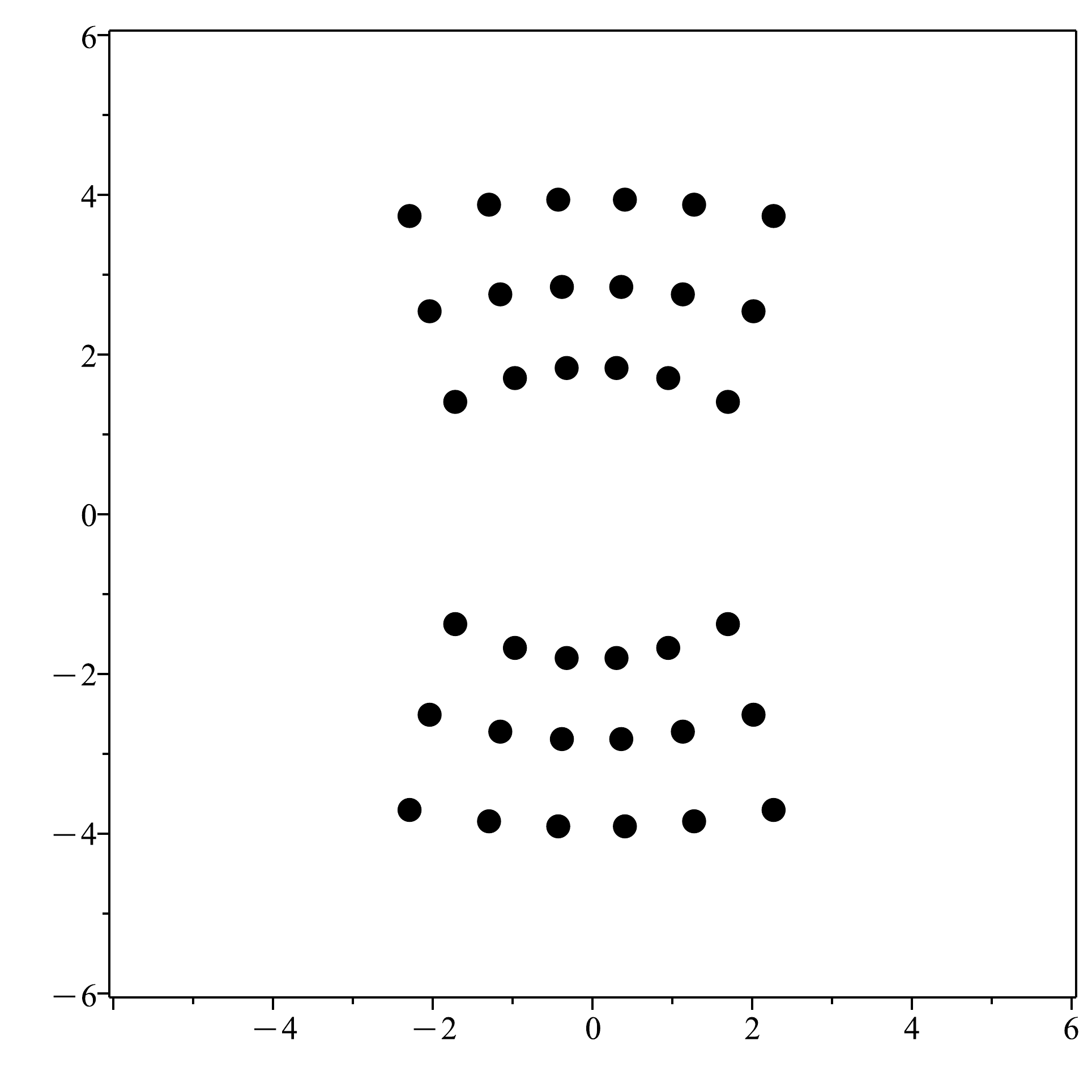}\\
\mu=-\fract{51}{5} & \mu=-11 & \mu=-\fract{61}{5}\\
\end{array}\]
\caption{\label{fig:T532}The roots of
$T_{5,3}^{(\mu)}(\tfrac12z^2) $ for $\mu \in [-\fract{61}{5}, -\fract{39}{5} ]. $}
\end{figure}
\begin{example}
Figures~\ref{fig:T531} and 
\ref{fig:T532} show the roots of $T_{5,3}^{(\mu)}(\tfrac12z^2) $ for various $\mu$. We describe 
the root blocks and transitions between the blocks 
 as $\mu$ varies from 
$-16/5$ to $ -61/5$. 
For $\mu > -4$ the roots form two E-type rectangles of size $6 \times 3$ 
as shown in the first two images in Figure~(\ref{fig:T531}). 
As $\mu \to -4$ all roots move 
 towards the imaginary axis. At $\mu=-4$ the 
 innermost column of three zeros from each rectangle have coalesced 
at the origin and the remaining roots form 
two rectangles of size $5 \times 3$. We discuss the detailed 
behaviour of the coalesecing zeros 
in the next section. 

As $ \mu$ decreases further, the zeros at the origin emerge as a pair of 
zeros on the imaginary axis and 
two complex zeros forming a pair of columns of height two. 
 The coalescing roots 
 move away from the origin, while the other roots
move towards the origin. {As $\mu $ continues to decrease}, 
the zeros that coalesced turn back towards the origin. At $\mu=-5$ 
these roots 
and the six roots in the 
column of the E-rectangle closest to the imaginary axis all coalesce at $z=0$. There are now 
twelve zeros at the origin and the remaining zeros 
form two rectangles of size $4 \times 3$. 
As $\mu$ decreases, the roots emerge from the origin 
as four $2$-triangles with the remaining roots forming two $4 \times 3 $ 
E-rectangles. The roots in the triangles 
 initially move away from the origin while the rectangles move towards the origin. 
For some $\mu \in (-6,-5)$ all the roots in the triangles have turned back towards the origin. 
At $\mu=-6$ the roots in the triangles and the next innermost 
column of zeros from each rectangle coalesce at the origin. 
After the next coalescence, we see the appearance of a 
a pair of F-trapezoids as well as G-triangles and E-rectangles.

Until all roots coalesce at $\mu=-m-n-1=-9$, 
the coalescing roots always 
consist of the roots that previously coalesced plus the 
innermost column of roots from each E-rectangle. These zeros 
re-configure and join new blocks as they emerge from the origin. 
The coalescing roots initially move away from the origin 
as $\mu$ decreases, and at various values of $\mu$ 
 return to the origin to re-coalesce. 
For $\mu < -m-n-1$, some of the roots start to form D-type rectangles. 
Such roots do not return to the origin as $\mu$ decreases, while all 
other roots return to the origin at each coalesence until they 
become part of a D-rectangle. 
The sizes of each root block of $T_{5,3}^{(\mu)}(\tfrac12z^2)$ 
 for $\mu$ between each coalescence point is given in Table~\ref{tab:T53}.




\begin{table}[ht]
\centering
\begin{tabular}{|c|c|c|c|c|} \hline
$\mu$ & E & G & F& D \\ 
 & rectangle & trapezoid/ & triangle/ & rectangle \\ 
& & triangle & trapezoid & \\\hline
$-4< \mu<\infty $ & $6 \times 3$ &&&\\ \hline
$-5<\mu<-4 $ & $5 \times 3$ & $2 \times 2$ & $1$ & \\ 
$-6<\mu<-5 $ & $4 \times 3$ & $2 \times 1$ & $2$ & \\\hline 
$-7<\mu<-6 $ & $3 \times 3$ & 2 & $3 \times 1 $& \\ 
$-8<\mu<-7 $ & $2 \times 3$ & 2 & $4 \times 2 $ & \\
$-9<\mu<-8 $ & $1 \times 3$ & 2 & $5 \times 3$ & \\ \hline
$-10<\mu<-9 $ & & 2 & $5 \times 4$ & $6 \times 1 $\\ 
$-11<\mu<-10 $ & & 1 & $5 \times 5$ & $6 \times 2 $\\\hline
$-\infty < \mu < -11 $ && $6 \times 3 $&&\\\hline 
\end{tabular}
\caption{\label{tab:T53}Size of the root blocks of $T_{5,3}^{(\mu)}(\tfrac12z^2)$. 
}
\end{table}

\end{example}

\begin{conjecture}
\label{conj:blocks}
The block structures when $\mu=-n-j $ for $j=1,\dots, m+n$ 
and there are roots at the origin are given in Table~\ref{tab:ngtm2}.
Our investigations suggest 
the root blocks of $T_{m,n}^{(\mu)}(\tfrac12z^2)$ 
are as per Table~\ref{tab:ngtm} for $n>m$ and Table~\ref{tab:nlm} for $n\le m$ for 
$\mu$ such that $\ceil{\mu}=-n-j$ where $j \in \ZZ$, excluding the points $\mu = -n-1,-n-2, \ldots, -2n-m$.
 \begin{table}[ht]
\centering
\begin{tabular}{|c|c|c|c|c|} \hline 
\multicolumn{2}{|c|}{Condition} &Number of zeros & E & D\\ 
\cline{1-2}$j$ & $\mu$ &at origin & rectangle &rectangle \\ \hline


$1,\ldots, m+1$ & $-n-j$ & $2nj$ & $m-j+1 \times n$ & \\ \hline

$2,\ldots, n$ & $-m-n-j$ & $2(m+1)(n+1-j) $ & & $m+1 \times j-1$ \\ 
 
\hline 
\end{tabular}
\caption{\label{tab:ngtm2}{Conjectured root blocks} of $T_{m,n}^{(\mu)}(\tfrac12z^2)$ 
at $\mu$ when there are zeros at the origin. }
\end{table}
\begin{table}[ht]
\centering
\begin{tabular}{|c|c|c|c|c|} \hline
Condition & E & G & F& D \\ 
 $j=-n-\ceil{\mu}$ & rectangle & trapezoid/ 
& triangle/
& rectangle \\ 
 & & triangle
& trapezoid 
 & \\\hline
$j\le 0$ & $m+1 \times n$ &&&\\ \hline
$1<j<m+1$ & $m+1-j \times n$ & $n-1 \times n-j$ & 
 $j $ &\\ \hline
$m+1< j<n$ & & $m+n-j \times n-j $ & $m $ & 
 $m+1 \times j-m$ \\\hline 
$n < j<m+n$ & & $m+n -j $ & $m \times j-n+1$ & 
 $m+1 \times j-m$ \\ \hline
$j > m+n$ & & & & $m+1 \times n$ \\ \hline 
\end{tabular}
\caption{\label{tab:ngtm}{Conjectured root blocks} of $T_{m,n}^{(\mu)}(\tfrac12z^2)$ when $n > m$ and 
$j=-n-\ceil{\mu}\in \ZZ$.}
\end{table}
\begin{table}[ht]
\centering 
\begin{tabular}{|c|c|c|c|c|} \hline
Condition & E & G & F& D \\ 
$j=-n-\ceil{\mu}$ & rectangle & trapezoid/ & trapezoid/ & rectangle \\ 
& & triangle & triangle & \\\hline
$j\le 0$ & $m+1 \times n$ &&&\\ \hline
$1<j<n$ & $m+1-j \times n$ & $n-1 \times n-j $ & 
 $j$ & \\ \hline
$n+1<j<m+1$ & $m+1-j \times n $ & 
 $n-1$ & $j \times j-n+1$ & \\ \hline 
$m+2<j<m+n$ & & $m+n-j$ & $m \times j-n+1$ & 
 $m+1 \times j-m$ \\ \hline
$j> m+n$ & & & & $m+1 \times n$ \\ \hline 
\end{tabular}
\caption{\label{tab:nlm}{Conjectured root blocks} of $T_{m,n}^{(\mu)}(\tfrac12z^2)$ when $n \le m$ and 
$j=-n-\ceil{\mu}\in\ZZ$.}
\end{table}

\end{conjecture}

The family of 
Wronskian Hermite polynomials with partitions $\bold{\Lambda}=(m^n)$ are known as the 
generalised Hermite polynomials $H_{m,n}(z)$. 
The roots 
 form $m \times n$ rectangles centered on the origin \cite{refPACpiv,refPACcmft}.

The appearance of rectangular blocks of width $m+1$ and height $n$ for large positive and negative $k$ 
in the root pictures for $\Tmn{m,n}{-2n-k-1/2}(\tfrac12z^2)$ 
is consistent with Theorem 9.6 and Remark 9.7 of \cite{refCM01}. The results therein imply for large $k$ the roots will, up to scaling, be those of a certain Wronskian Hermite polynomial shifted to the right along the real axis, plus the block reflected in the imaginary axis. The numerical investigations in \cite{refBDS} suggest that the relevant Wronskian Hermite polynomial is $H_{m+1,n}(z)$.

\subsection{Root coalescences}
We now zoom into the origin to investigate precisely how the zeros that coalesce behave as they approach
 and leave the origin. We start with the example 
of $T_{5,3}^{(\mu)}(\tfrac12z^2)$, for which the coalescences 
occur at $\mu=-11,-10, \ldots, -4$. 

\begin{example}
Recall that at $\mu \to -4^+$, the six roots of $T_{5,3}^{(\mu)}(\tfrac12z^2)$ 
that form the two innermost columns of the E-rectangles coalesce at $\mu=-4$.
The left-hand plot in Figure~\ref{fig:T53_first_smooth} 
shows the coalescence of these six zeros by overlaying the root plots for 
$\mu \in [-4,-16/5] $ near the origin. 
\begin{figure}[ht]\[ \begin{array}{c@{\quad}c}
 \includegraphics[width=0.4\linewidth]{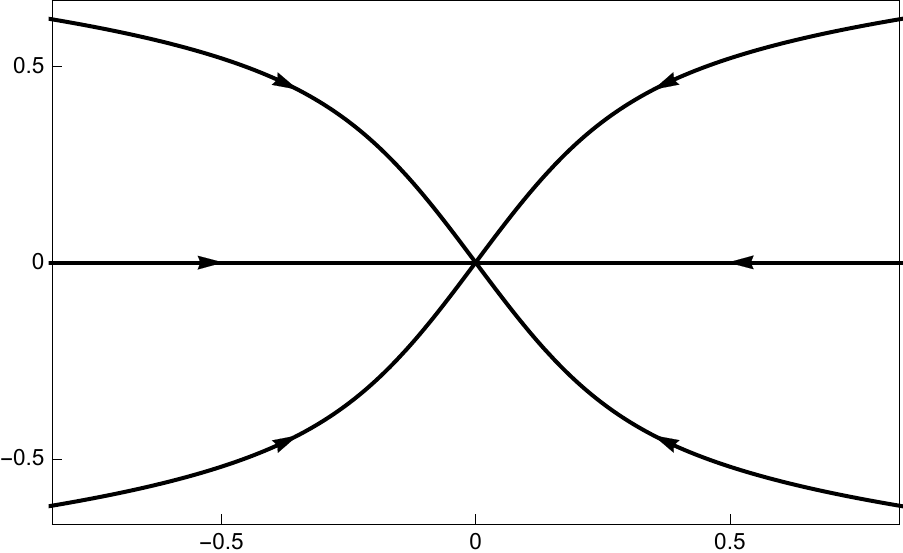}
& 
 \includegraphics[width=0.4\linewidth]{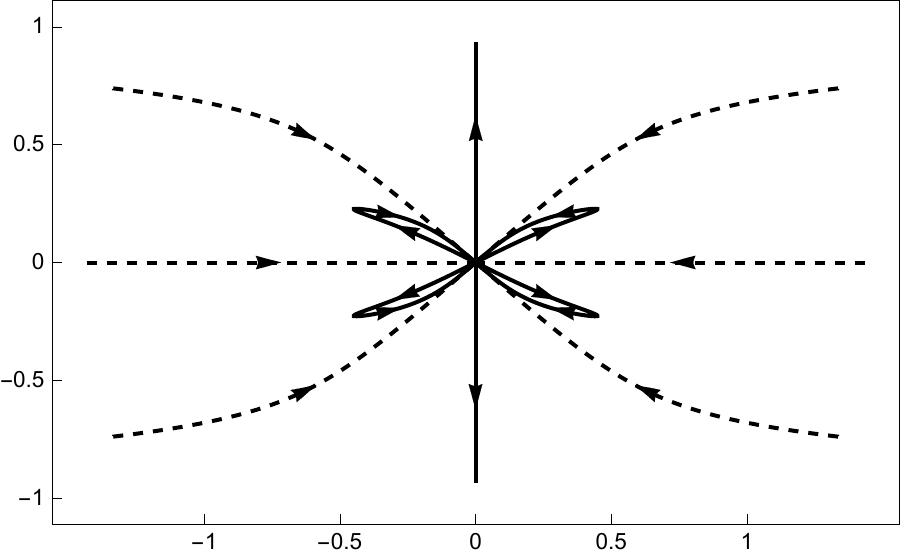} 
\\
\mu \in [-4, -\fract{16}{5}] & \mu \in [-5, -4] \\ 
\end{array}\]
\caption{\label{fig:T53_first_smooth} The coalescence of the 
zeros of $T_{5,3}^{(\mu)}(\tfrac12z^2) $ that are closest 
to the origin shown by overlaying the zero plots as $\mu$ tends to 
$\mu=-4$ (left) and $\mu=-5$ (right). The arrows show the direction in which $\mu $ decreases. 
The solid lines correspond to zeros that arise from the first 
column of the E-rectangles, and the dashed lines correspond to 
zeros that arise from the second column of the E-rectangles. }
\end{figure}
The bold lines in the 
right-hand plot of Figure~\ref{fig:T53_first_smooth} shows 
the reapparance of those zeros as $\mu$ decreases towards $\mu=-5$. 
 The 
previously-real zeros move onto the imaginary axis and the complex zeros 
return to the complex plane and move away from the origin. The 
arrows show the direction of decreasing $\mu$. 
At $\mu\approx 4.2105$, the complex zeros that coalesced 
turn back towards the origin. The lower solid line in the first quadrant shows the movement of the complex root for $\mu \in (4.2105, -4]. $ The upper line shows the root for $\mu \in [-5, 4.2105)$. 
 At $\mu \approx 4.32656$, the imaginary zeros also turn back to the origin. 
The dashed lines show the coalescence of the six zeros 
in the innermost columns of the E-rectangles for $\mu$ from $-4$ to $-5$. 
At $\mu=-5$ all twelve zeros are at the origin. 
The top right plot in Figure~\ref{fig:T53_transitions} shows 
the twelve zeros as they emerge from the origin as $\mu$ decreases from $4$. 
\begin{figure}[ht]\[ \begin{array}{c@{\quad}c@{\quad}c}
 \fig{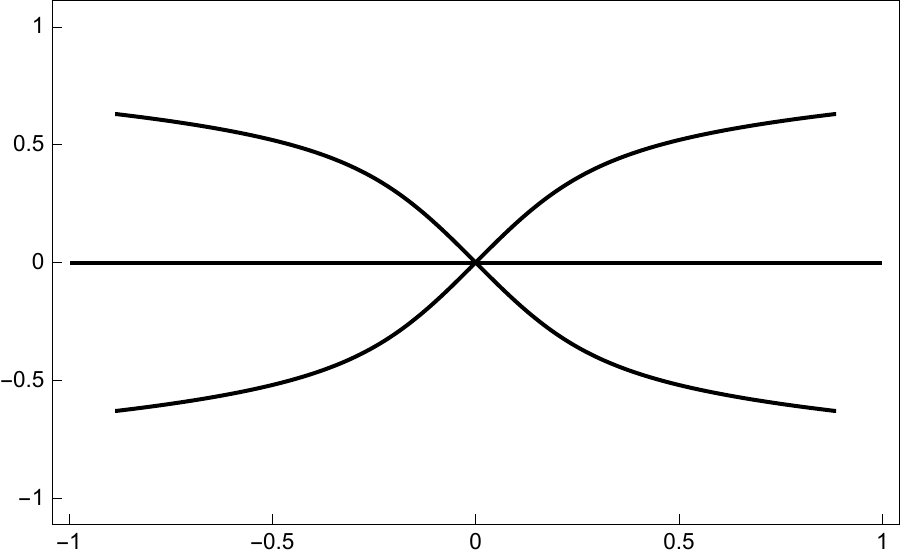} & \fig{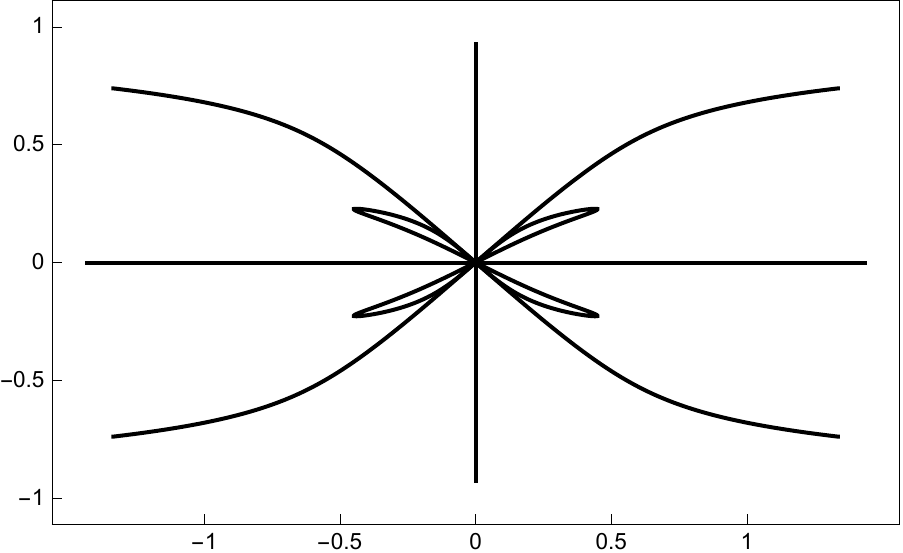} & \fig{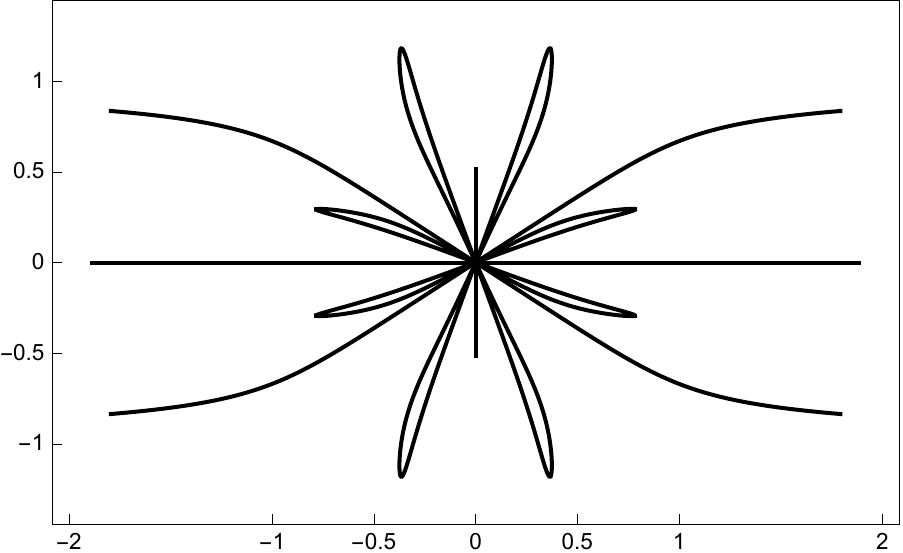} \\ 
\mu \in [-4,-3 ] & \mu\in [-5,-4]& \mu \in [-6,-5] \\ 
\fig{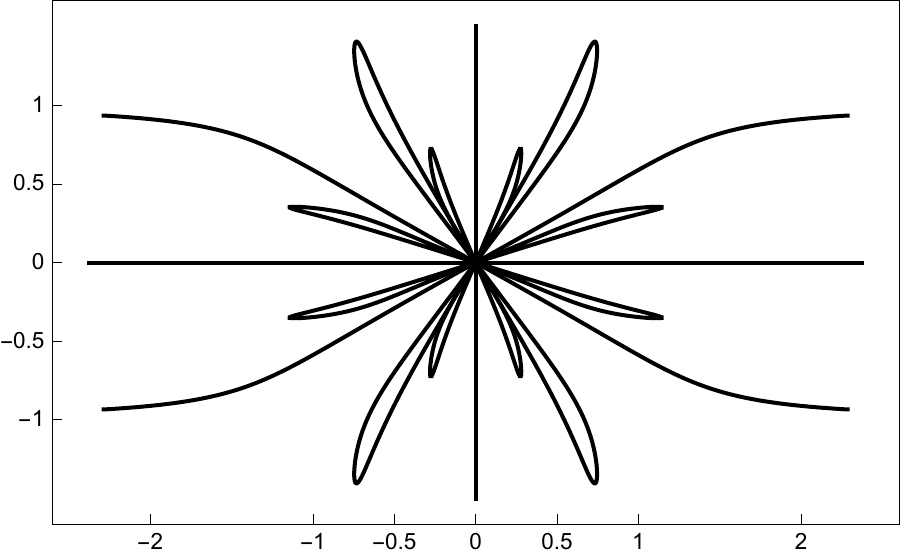} & 
 \fig{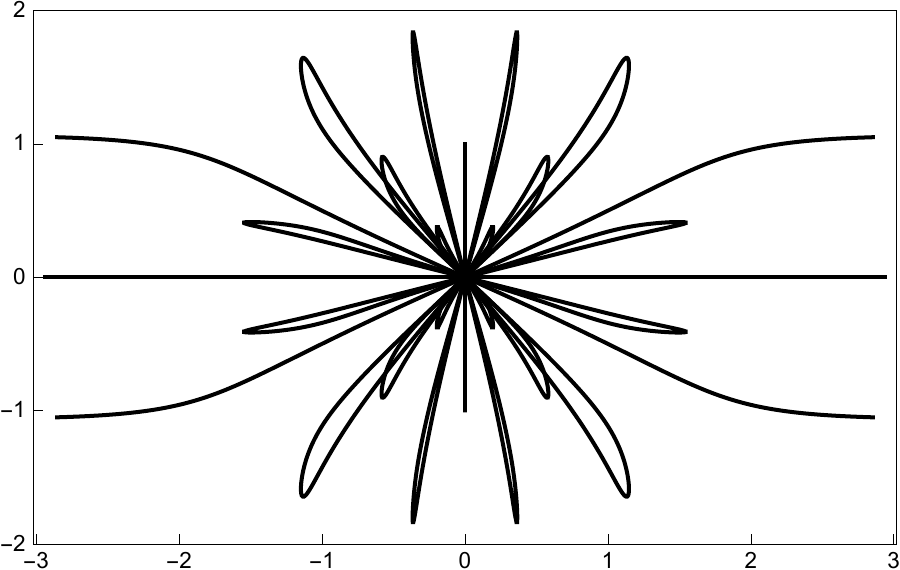} & \fig{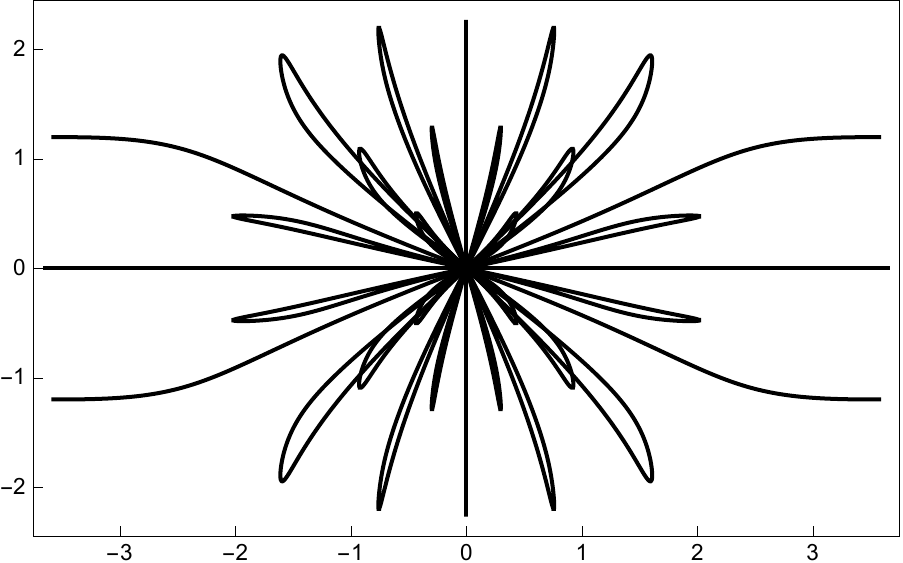} \\
 \mu \in[-7,-6] & \mu\in [-8,-7]& \mu \in [-9,-8] \\ 
 \fig{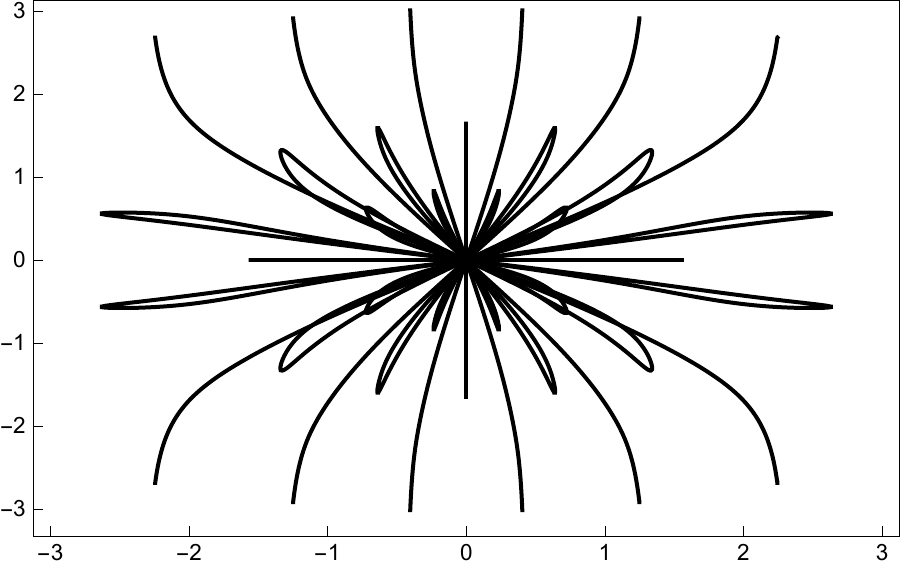} & \fig{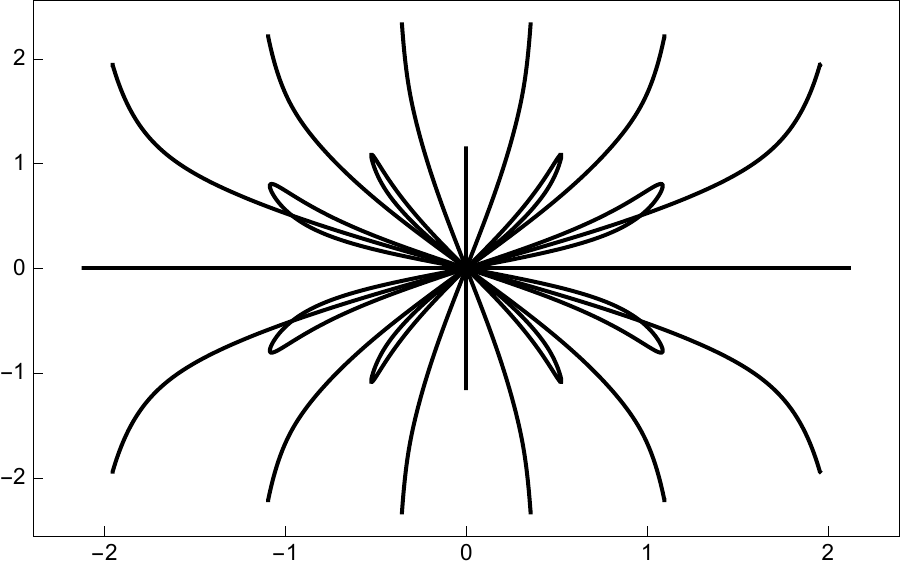} & \fig{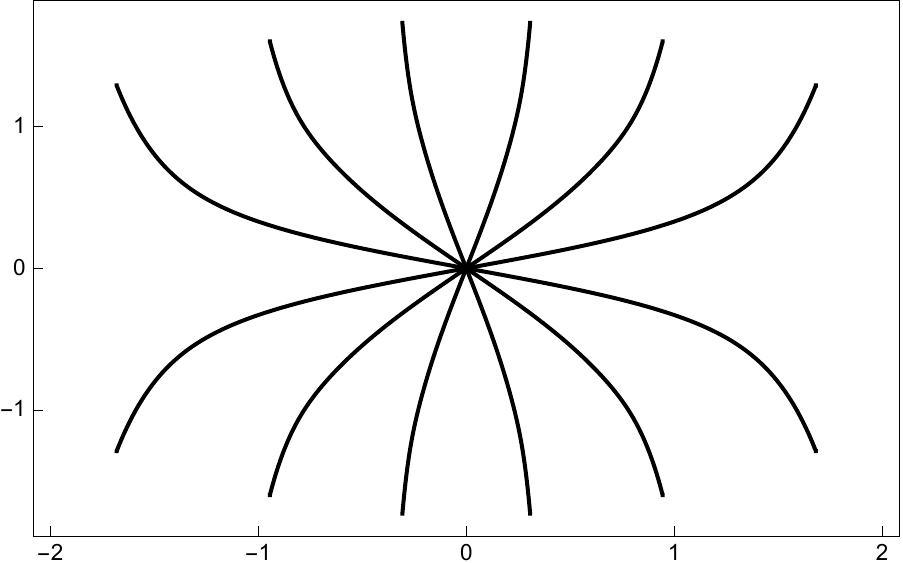}
\\
 \mu \in[-10,-9] & \mu\in [-11,-10]& \mu \in [-12,-11] \\ 
\end{array}\]
\caption{\label{fig:T53_transitions} The movement of the roots of
$T_{5,3}^{(\mu)}(\tfrac12z^2)$ closest to the origin overlaid for $\mu$ in each given interval. }
\end{figure}
There are two roots on the imaginary axis, two on the real axis and 
eight in the complex plane, all of which initially 
move away from the origin. All roots eventually turn around  and return 
to the origin, along with the next set of six zeros from 
the innermost column of the E-rectangles. We see the petal-like shapes traced 
out by the complex zeros as $\mu$ decreases from $-5$ to $-6$. {The values of $\mu$ at which each set of zeros turn around are different.}
The remaining plots 
in Figure~\ref{fig:T53_transitions} show the zeros emerging from the origin and 
 those that coalescence for each of the stated $\mu$. 
Some roots form F-rectangles when $\mu<-9$. 
\end{example}

Our numerical investigations reveal that the angles in the complex plane 
 at which the coalescing roots approach the origin and emerge from it 
can be determined for all $m,n,j$ where $\mu=-n-j$ and $j=1,2,\ldots, m+n$. Before giving the result for 
$\Tmn{m,n}{\mu}(z)$ as a function of $z$, we consider an example. 

\begin{example}\label{ex_T23}
The roots of $\Tmn{2,3}{\mu}$ that coalesce at $ \mu=-3-j-\ep$ for $j=1\dots,5$ 
behave as the $n^{\rm th}$ roots of one or minus one as follows:
$$
\begin{array}{c@{\quad}c@{\quad}c@{\quad}c}
 j & \mu & \mu\to \mu^+ & \mu \to \mu^-\\\hline
1&-4 & (z^3 -1) & (z^3+1)\\ 
2&-5 & (z^4 -1) (z^2+1) & (z^4+1) (z^2-1)\\ 
3&-6 & (z^5-1)(z^3+1)(z-1) & (z^5+1)(z^3-1)(z+1)\\
4&-7 & (z^4+1)(z^2-1) & (z^4-1)(z^2+1)\\
5&-8 & (z^3-1) & (z^3+1) 
\end{array}
$$
Figure~\ref{fig:T23_m4in} shows the roots of 
$\Tmn{2,3}{\mu}$ that converge to to the origin (left) 
 as $\mu \to -4$ 
and emerge (right) from the origin. 
The third roots of $1$ and $-1$ are shown in black and red respectively. 
\begin{figure}[ht]\[ \begin{array}{c@{\quad}c}
 \includegraphics[width=0.3\linewidth]{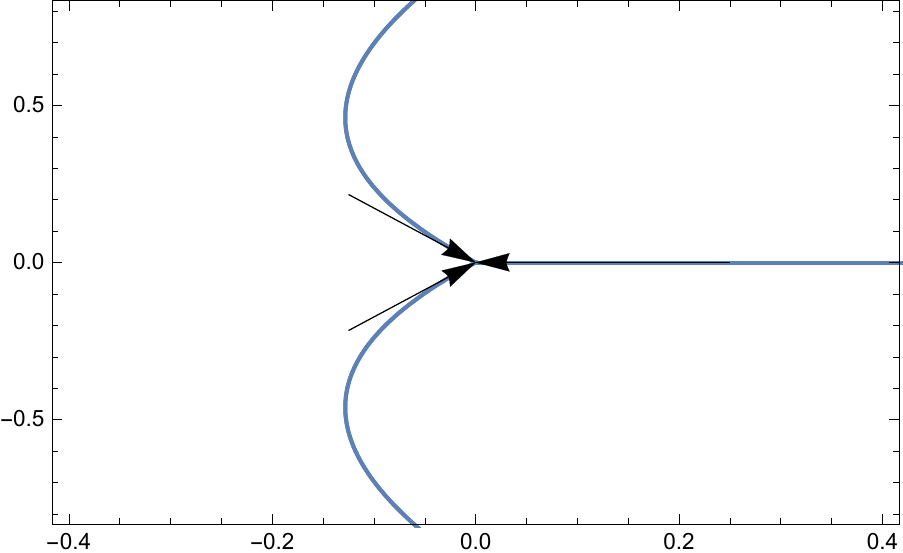}
& 
 \includegraphics[width=0.3\linewidth]{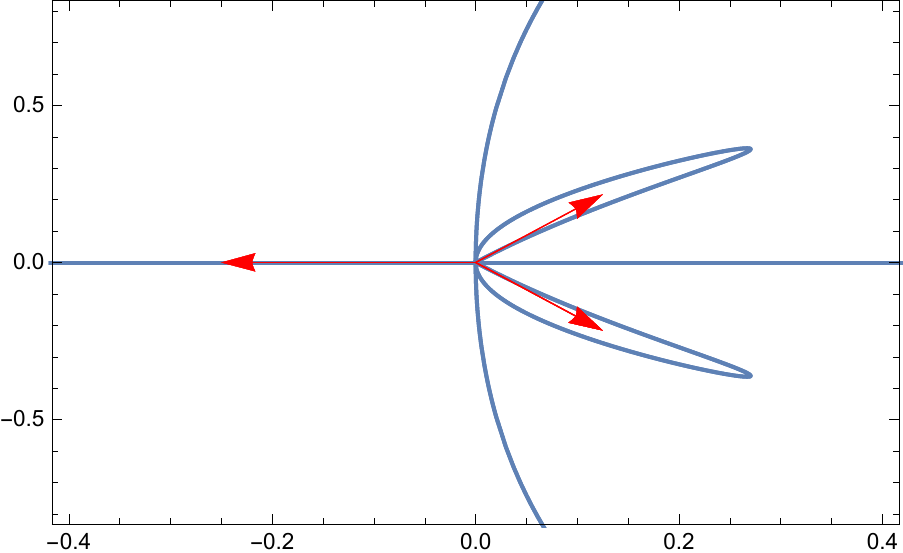} 
\\
\mu \in [-4, -3 ] & \mu \in [-5, -4] \\ 
\end{array}\]
\caption{\label{fig:T23_m4in} The coalescence of the 
zeros of $\Tmn{2,3}{\mu}$ that are closest 
to the origin shown by overlaying the zero plots as $\mu$ approaches 
$\mu=-4$ (left) and $\mu=-5$ (right) from the right. 
The black arrows (left) indicate the direction of the root movement as $\mu \to -4$ from the right and the red arrows (right) show the roots leaving the origin 
as $\mu$ decreases from $-4$. The black arrows 
show the third roots of unity and the red arrows (right) show the 
third roots of $-1$. The blue lines in the right figure without arrows 
correspond to the movement of the roots that 
approach the origin as $\mu \to -5^-$ at angles corresponding to the 
fourth roots of $1$ and the square roots of $-1$. 
 }
\end{figure}

\end{example}

\begin{conjecture}\label{angles}
Let $n>m$ and $\ep>0$. For $\mu = -n-j+\ep$ where $j=1,2,\ldots, m+1$ the $nj$ roots of $\Tmn{m,n}{\mu}(z)$ 
 that coalesce at the origin
 at $\ep =0$ approach the origin on the rays in the complex plane defined by 
certain roots of $+1$ and $-1$. We encode this behaviour 
in the polynomial 
 \eq
\prod_{k=1}^j \lf(z^{n+j+1-2k} - (-1)^{n+k}\ri),\qquad j=1,2,\ldots,m+1.
\label{rootangles1}
\en
Furthermore, when $\mu=-n-j+\ep$ for $j=m+2, \ldots, m+n$ the $(m+1) (m+n+1-j)$ roots that approach 
the origin behave as roots of $\pm 1$ according to 
\begin{subequations}\begin{align}
&\prod_{k=j-m}^j \lf(z^{n+j+1-2k} - (-1)^{n+k}\ri),&& j=m+2,m+3,\ldots,n, \\ 
&\prod_{k=j-m}^n \lf(z^{n+j+1-2k} - (-1)^{n+k}\ri),&& j=n+1,n+2,\ldots,m+n.
 \end{align}
 \label{rootangles2}
 \end{subequations}
The roots that coalesce leave the origin on rays that are rotated through 
$\tfrac12\pi$ compared to the coalescence rays. Thus the root behaviours as $\mu=-n-j-\ep$ for $j=1,2,\ldots, m+n$ are encoded in the polynomials
\begin{subequations}\begin{align}
&\;\;\prod_{k=1}^j \lf(z^{n+j+1-2k}+ (-1)^{n+k}\ri),&& j=1,2,\ldots,m+1, \\ 
&\prod_{k=j-m}^j \lf(z^{n+j+1-2k}+ (-1)^{n+k}\ri),&& j=m+2,m+3,\ldots,n, \\ 
&\prod_{k=j-m}^n \lf(z^{n+j+1-2k}+ (-1)^{n+k}\ri),&& j=n+1,n+2\ldots,m+n.
 \end{align}
 \label{rootangles3}
 \end{subequations}

Similarly, 
when $n\le m $ the roots coalesce at and emerge from the origin as $\mu=-n-j\pm\ep$ as roots of $\pm 1$ according to 
\begin{subequations}
\begin{align}
&\;\;\prod_{k=1}^j \lf(z^{n+j+1-2k} \mp (-1)^{n+k}\ri) ,&& j=1,2\ldots,n , \\ 
&\;\;\prod_{k=1}^n \lf(z^{n+j+1-2k} \mp (-1)^{n+k}\ri) ,&& j=n+1,n+2,\ldots,m+1, \\ 
&\prod_{k=j-m}^n \lf(z^{n+j+1-2k} \mp (-1)^{n+k}\ri) ,&& j=m+2,m+3,\ldots,m+n.
 \end{align}
\label{rootangles4}
\end{subequations}
\end{conjecture}

\subsection{The role of the partition}

In this section we remark that several features of the generalised Laguerre polynomials can be written in terms 
of partition data, particularly the hooks of the partition $\bfla=(m+1)^n$.

We first propose an expression for the coefficients of the Wronskian Laguerre 
polynomials $\Omega_{\bfla}^{(\alpha)}(z)$ for all partitions $\bfla$. The result generalises the expression given in 
 Theorem 3 and Proposition 2 in \cite{refBDS} for the coefficients of the Wronskian Hermite polynomials $H_{\bfLA}(z)$
for the subset of partitions $\bfLA$ with $2$-quotient $(\bfla,\bold{\emptyset})$. 
\begin{conjecture}
\label{conj:omega_coeffs}
Consider the Wronskian Laguerre polynomial $\Omega_{\bfla}^{(\alpha)}(z)$ defined in (\ref{bk_omega}). Set 
\eq
 \Omega_{\bfla}^{(\alpha)}(z) =c_{\bfla}\sum_{j=0}^{|\bfla|} r_{j}^{(\alpha)} \,z^{ |\bfla| -j},\
 \label{omega_exp}
 \en
 with $r_{0}^{(\alpha)}=1$. Then 
 \eq
 c_{\bfla}
 = \frac{\Delta_{\bfla} }
{
\prod_{h \in \bold{h}_{\bfla}} (-1)^h h!
}.
\label{r0}
\en
 and 
\eq
r_{j}^{(\alpha)} = \binom{ |\bfla|}{j} \sum_{\widetilde \bfla <_j \, \bfla }
\frac{F_{\widetilde \bfla} F_{ \bfla / \widetilde \bfla} }
{F_{\bfla}} \frac{\Psi_{\bfla}^{(\alpha)}} {\Psi_{\widetilde \bfla}
^{(\alpha+\ell(\bfla)-\ell(\widetilde \bfla) )}}\,,
\en
where the sum is over all partitions $\widetilde \bfla$ in the Young lattice obtained by removing $j$ 
boxes from the Young diagram of $\bfla$. 
Moreover, 
\begin{align}
\Psi_{\bfrh}^{(\alpha)} &= (-1)^{|\bfrh|+ 
\text{ht}(\bold{\Rho})} \prod_{j=1}^{\ell(\bfrh)}
\Biggl ( \ \prod_{k=\ell(\bfrh)}^{\bold{h}_{\bfrh_j}-1}\lf( 
\bold{h}_{\bfrh_j}-k +\alpha+\ell(\bfrh) \ri) 
\nn \\
& \quad \quad \quad \quad \quad \quad \quad \times 
\prod_{k \in \{0,1,\dots,\ell(\bfrh)-1\} \setminus \bold{h}_{\bfrh} } ^{j-1}\lf( j-1- k
-\alpha-\ell(\bfrh)
 \ri) \Biggr)
 \label{Psi}
 \end{align}
 where $\text{ht}(\bold{\Rho})$ is the number of vertical dominoes in the partition $\bold{\Rho}$ that has empty 2-core and 2-quotient $( \bfrh,\bold{\emptyset})$. 
 We remark that $\Psi_{\bfrh}^{(\alpha)}$ is a polynomial of degree $|\bfrh|$ in $\alpha$ with leading coefficient $(-1)^{|\bfrh|}$. A consequence is that all coefficients of the Wronskian Laguerre polynomial are written through (\ref{Psi}) in terms of the hooks of partitions. 
\end{conjecture}

\begin{figure}
\vspace{1cm}
\begin{center}
 \begin{subfigure}[b]{0.48\textwidth}
\centering
\hspace{-1.5cm}
\begin{tikzpicture}[scale=0.7]
\begin{Tableau}{{,,,,,,,,},{,,,,,,,},{,,,,,,},{,,},{,},{ }}
\draw[draw=black] (0.1,-1.9) rectangle ++(0.8,1.8);
\draw[draw=black] (0.1,-3.9) rectangle ++(0.8,1.8);
\draw[draw=black] (0.1,-5.9) rectangle ++(0.8,1.8);
\draw[draw=black] (1.1,-2.9) rectangle ++(0.8,1.8);
\draw[draw=black] (1.1,-4.9) rectangle ++(0.8,1.8);
\draw[draw=black] (2.1,-3.9) rectangle ++(0.8,1.8);
\draw[draw=black] (1.1,-0.9) rectangle ++(1.8,0.8);
\draw[draw=black] (3.1,-0.9) rectangle ++(1.8,0.8);
\draw[draw=black] (5.1,-0.9) rectangle ++(1.8,0.8);
\draw[draw=black] (7.1,-0.9) rectangle ++(1.8,0.8);
\draw[draw=black] (2.1,-1.9) rectangle ++(1.8,0.8);
\draw[draw=black] (4.1,-1.9) rectangle ++(1.8,0.8);
\draw[draw=black] (6.1,-1.9) rectangle ++(1.8,0.8);
\draw[draw=black] (3.1,-2.9) rectangle ++(1.8,0.8);
\draw[draw=black] (5.1,-2.9) rectangle ++(1.8,0.8);
\draw [decorate,decoration={brace,amplitude=10pt,raise=1pt
},xshift=0pt,yshift=0pt]
(0.1,0.1) -- (8.9,0.1) node [black,midway,above=10pt]{$2m+1$}; 
\draw [decorate,decoration={brace,amplitude=10pt,raise=1pt
},xshift=0pt,yshift=0pt]
 (-0.1,-2.9) -- (-0.1,-0.1) node [black,midway,left=10pt]{$n$}; 
\draw [decorate,decoration={brace,amplitude=10pt,raise=1pt
},xshift=0pt,yshift=0pt]
 (-0.1,-5.9) -- (-0.1,-3.1) node [black,midway,left=10pt]{$n$}; 
\end{Tableau} 
\end{tikzpicture}
\vspace{-1cm}
 \caption{
The Young diagram of $\bold{\Lambda}_{4,3} = (9,8,7,3,2,1)$. }
 \label{fig:youngtiling43}
 \end{subfigure}
\hfill
 \begin{subfigure}[b]{0.48\textwidth}
\hspace{-1.5cm}
\begin{tikzpicture}[scale=0.7]
\begin{Tableau}{{,,},{,},{,},{,},{,},{ }}
\draw[draw=black] (0.1,-1.9) rectangle ++(0.8,1.8);
\draw[draw=black] (0.1,-3.9) rectangle ++(0.8,1.8);
\draw[draw=black] (0.1,-5.9) rectangle ++(0.8,1.8);
\draw[draw=black] (1.1,-2.9) rectangle ++(0.8,1.8);
\draw[draw=black] (1.1,-4.9) rectangle ++(0.8,1.8);
%
\draw[draw=black] (1.1,-0.9) rectangle ++(1.8,0.8);
\draw [decorate,decoration={brace,amplitude=10pt,raise=1pt
},xshift=0pt,yshift=0pt]
(0.1,0.1) -- (2.9,0.1) node [black,midway,above=10pt]{$2m+1$}; 
\draw [decorate,decoration={brace,amplitude=10pt,raise=1pt
},xshift=0pt,yshift=0pt]
 (-0.1,-1.9) -- (-0.1,-0.1) node [black,midway,left=10pt]{$m+1$}; 
\draw [decorate,decoration={brace,amplitude=10pt,raise=1pt
},xshift=0pt,yshift=0pt]
 (-0.1,-3.9) -- (-0.1,-2.1) node [black,midway,left=10pt]{$2(n-m-1)$};
\draw [decorate,decoration={brace,amplitude=10pt,raise=1pt
},xshift=0pt,yshift=0pt]
 (-0.1,-5.9) -- (-0.1,-4.1) node [black,midway,left=10pt]{$m+1$}; 
\end{Tableau}
\end{tikzpicture}
\vspace{-1cm}
 \caption{The Young diagram of $\bold{\Lambda}_{1,3} = (3,2^4,1)$. 
}
 \label{fig:youngtiling13}
 \end{subfigure}
\end{center}
\caption{Examples of Young diagrams of 
$\bold{\Lambda}_{m,n}$ for $m>n-2$ (left) 
and $m \le n-2$ (right). 
The domino tiling is shown. 
 The number of vertical dominoes 
is $\text{ht}( \bold{\Lambda}_{4,3})=6$
 and 
$\text{ht}(\bold{\Lambda}_{1,3})=5$ respectively. 
\label{fig:youngtiling}}
\end{figure}
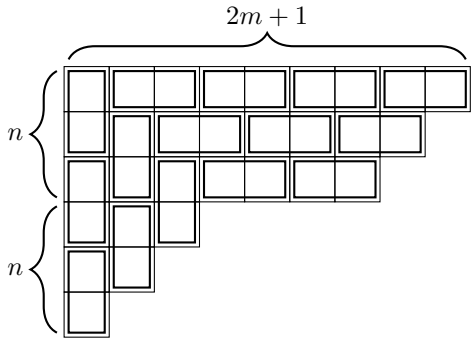
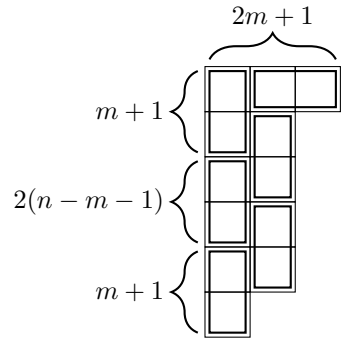

\begin{remark}
We have also generalised Conjecture \ref{conj:omega_coeffs} 
 to determinants of Laguerre polynomials of universal character type \cite{refKoike}. Such polynomials are defined in terms of two partitions and are generalisations of Wronskian Hermite polynomials $H_{\bfLA}(z)$
with $2$-quotient $(\bfla_1,\bfla_2)$. 
Examples include the generalised Umemura polynomials \cite{refMOK} and the 
 Wronskian Laguerre polynomials arising in \cite{refBK18, refDuran14b,refDP15,refGUGM18}. 
A proof of the more general result is under consideration. 
\end{remark}

We now record some information about the partitions ${\bfla}=((m+1)^n)$ of the generalised Laguerre polynomial $\Tmn{m,n}{\mu}(z)$ and the corresponding partition $\bfLA_{m,n}$ with empty 2-core and 2-quotient $(\bfla, \bold{\emptyset})$.
The Young diagram of $((m+1)^n)$
is a rectangle of width $m+1$ and height $n$.
Since the
degree vector of $\bfla$ is 
\[\bold{h}_{\bfla}=(m+n, m+n-1, \ldots, m+1),\] 
 the Vandermonde determinant is 
\[
\Delta(\bold{h}_{\bfla})=
(-1)^{n(n-1)/2}\prod_{j=2}^n (j-1)!,\]
Since ${\bfla}^*=( n^{m+1})$, the 
 multiset of hooks $\mathcal{H}_{m,n}$ of ${\bfla}$ following from (\ref{hook_length}) is
\eq
\mathcal{H}_{m,n}= \{\{ m+n+2-j-k\}_{k=1}^{m+1}\}_{j=1}^n\,.
\label{hook1}
\en
The multiset can also be written as 
\eq
\mathcal{H}_{m,n}= \{k^k\}_{k=1}^{{\rm min} (m+1,n)-1}
\cup \{k^{{\rm min}(m+1,n)}\}_{k={\rm min}(m+1,n) }^{{\rm max} (m+1,n)}
\cup 
 \{k^{m+n+1-k}\}_{k={\rm max}(m+1,n)+1 }^{m+n}.
 \label{hookmulti}
\en

We now describe the Young diagram of $\bold{\Lambda}_{m,n}$ and determine its $2$-height. The shape of the Young diagram depends on the 
relative values of $m$ and $n$. 
When $m >n-2$, the Young diagram consists of 
 the top $n$ rows of a staircase partition of size $2m+1$ with 
 a complete staircase of size $n$ below. When $m\le n-2$ 
the Young diagram consists of the top $m+1$ rows of a 
$2m+1$ staircase, then $2(n-m-1)$ rows of length $m+1$ and finally a
complete $m+1$ staircase. The two cases are illustrated in Figure~\ref{fig:youngtiling}.

All Young diagrams corresponding to partitions $\bold{\Lambda}(0,\bfnu)$ with empty 2-core and 2-quotient $(\bfnu, \bold{\emptyset})$ 
have a unique tiling with $|\bfnu|$ 
 dominoes: 
tile the boxes of the Young
diagram to the right and above the main diagonal with horizontal dominoes and tile the boxes 
on and below the main diagonal with vertical 
dominoes. The tiling is illustrated in Figure~\ref{fig:youngtiling}. The number of vertical dominoes and, therefore, the $2$-height of $\bold{\Lambda}(0,\bfnu)$ is 
$$
\text{ht}(\bold{\Lambda}(0,\bfnu)) =\sum_{j=1}^d (\lambda_j^* -j)/2,
$$
where $d$ is the number of boxes in the main diagonal or, equivalently, the size of the Durfee square. 
 The $2$-heights of the 
 Young diagrams of $\bold{\Lambda}_{m,n}$ are therefore 
 \eq
\text{ht}(\bold{\Lambda}_{m,n}) = 
\begin{cases}
n(n+1)/2 & m >n-2 \\ 
(2n-m)(m+1)/2 & m\le n-2. 
\end{cases}
\label{height}
\en

\begin{lemma}
Recall the expansion (\ref{Texp}) of the generalised Laguerre polynomial 
$$
T^{(\mu)}_{\bfla}(z) = 
c_{m,n} \lf ( z^{n(m+1)} + d_1^{(\mu)} z^{n(m+1)-1} + 
\dots + (-1)^{n(m+1)} d_{n(m+1)}^{(\mu)})\ri)\,.
$$
The overall constant is 
\eq
c_{m,n}=(-1)^{n(m+1)} 
\frac{ \Delta(h_{\bfla} )
}
{\prod_{h \in \bold{h}_{\bfla} } (-1)^h \,h!
} 
\en
where 
\eq
\Delta(\bold{h}_{\bfla})=
(-1)^{n(n-1)/2}\prod_{j=1}^n (j-1)!,
\en
and 
\eq
d_{1}^{(\mu)} = -n(m+1) (\mu+m+n+1)\,.
\en
The constant $d_{n(m+1)}^{(\mu)}$ can be written 
in terms of the hooks of the Young diagram of $\bfla$: 
\eq
d_{n(m+1)}^{(\mu)} = \prod_{h \in \mathcal{H}_{m,n}} \mu+n+h\,.
\en
\end{lemma}
\begin{proof}
Set $\bfla = ((m+1)^n). $ 
Then $ \ell(\bfla)=n $ and $|\bfla|=n(m+1)$. 
Using the relation (\ref{TOmega}) between $T_{m,n}^{(\mu)}(z)$ and 
$\Omega_{\bfla}^{(\alpha)}(z)$ and comparing the expansions 
 (\ref{Texp}) and (\ref{omega_exp}), we have 
$$
c_{m,n}=(-1)^{n(n-1)/2}c_{\bfla},
$$
$$
d_{1}^{(\mu)} = r_1^{(\mu+n)}=n(m+1) \frac{\Psi^{(\mu+n)}_{\bfla}}{\Psi^{(\mu+n)}_{\widetilde \bfla}},
$$
and 
\beq
d_{n(m+1)}^{(\mu)} =(-1)^{n(m+1)} r_{n(m+1)}^{(\mu+n)} = 
(-1)^{n(m+1)} \Psi^{(\mu+n)}_{\bfla}\,. 
\label{eq:dpsi}
\eeq
The expression for $c_{m,n}$ follows from (\ref{r0}) using the 
degree vector $\bold{h}_{\bfla}$.

We now determine $\Psi^{(\alpha)}_{\bfla}$ from (\ref{Psi}). We need 
 (\ref{height}) 
and 
$$
\{0,1,\dots n -1\} \setminus \bold{h}_{\bfla} 
= \begin{cases}
\{0,1,\dots n -1\},\qquad & m>n-2, \\
 \{0,1,\dots m\}, & m\le n-2.
\end{cases}
$$
We deduce that when $m>n-2$ then 
\begin{align}
\Psi_{\bfla}^{(\alpha)} &= (-1)^{n(m+1) + n(n+1)/2} 
\prod_{j=1}^{n}\lf( 
\prod_{k=n}^{m+n-j}\lf( 
m+2n+1-j-k +\alpha \ri) 
\prod_{k=0}^{j-1}\lf( 
j-1-k -\alpha-n \ri) \ri)
\nn \\ 
 &= (-1)^{n(m+1) } 
\prod_{j=1}^{n} \lf ( 
\prod_{k=1}^{m+1-j}\lf( 
m+n+2-j-k +\alpha \ri) \prod_{k=m+2-j}^{m+1}\lf( 
m+n+2-j-k +\alpha \ri) \ri), 
\end{align}
where the second line follows after changing variables and taking 
a minus sign out of each entry in the second set of products. 
If $m<n-2$ then 
\begin{align}
\Psi_{\bfla}^{(\alpha)} &= (-1)^{n(m+1) +(2n-m)(m+1)/2 } 
\prod_{j=1}^{m} 
\prod_{k=n}^{m+n-j}( 
m+2n+1-j-k +\alpha ) 
 \prod_{j=1}^{n} 
\prod_{k=0}^{\text{min}(j-1,m)}( 
j-1-k-\alpha-n ) 
\nn \\ 
 &= (-1)^{n(m+1) } 
\prod_{j=1}^{m} \lf ( 
\prod_{k=1}^{m+1-j}( 
m+n+2-j-k +\alpha ) \prod_{k=m+2-j}^{m+1}( 
m+n+2-j-k +\alpha ) \ri ) \nn \\
& \quad \quad \quad \times \prod_{j=m+1}^n\prod_{k=1}^{m+1} (m+n+2-j-k+\alpha)\,. 
\end{align}
Recalling that the hook in box $(j,k)$ of the Young diagram of $\bfla$ is 
$h_{j,k}= m+n+2-j-k$, we deduce for all $m,n$ that
\eq
\Psi_{\bfla}^{(\alpha)} = (-1)^{n(m+1)} \prod_{j=1}^n\prod_{k=1}^{m+1}
(h_{j,k}+\alpha)\,.
\en
Therefore from (\ref{eq:dpsi}) we conclude that 
\eq
d_{n(m+1)}^{(\mu)} =
\prod_{j=1}^n\prod_{k=1}^{m+1}
(h_{j,k}+\mu+n)\,.
\en

To determine the coefficient $r_1^{(\alpha)}$ we find 
all partitions $\widetilde \bfla$ obtained from $\bfla$ by removing 
one box from the Young diagram of $\bfla$ such that the result is a valid Young diagram. Since the Young diagram of $\bfla$ is a rectangle, the only
possibility is to remove box in position $(n,m+1)$ . Hence 
\eq
\widetilde \bfla = ((m+1)^{n-1} ,m),\qquad 
\bold{h}_{\widetilde \bfla } =( m+n, m+n-1, \dots, m+2,m),
\en
and $\ell(\widetilde \bfla)=n$ and $|\widetilde \bfla | =n(m+1)-1$.
Clearly $F_{\bfla} = F_{\widetilde \bfla} $ 
and $F_{{\bfla}/{\widetilde \bfla}}=1$.
We also need the 2-height of 
the partition $\widetilde \bfLA$ with empty 2-core and quotient $(\widetilde \bfla, \bold{\emptyset})$.
The partition is 
\eq
\widetilde \bfLA=
\begin{cases}
{\lf( \{2m-j+1\}_{j=0}^{m} , \{m+1\}_{j=1}^{2(n-m-1)-1} ,m
, \{m-j\}_{j=0}^{m-1} \ri ),}\quad 
& m \le n-2, \\[3pt] 
\lf ( \{2m-j+1\}_{j=0}^{m-1} , m,m ,\{m- j\}_{j=0}^{m-1} \ri ), & m=n-1, \\[3pt]
\lf ( \{2m-j+1\}_{j=0}^{n-2} , \{2m-n\} ,\{n- j\}_{j=0}^{n-1} \ri ), & m>n-2,
\end{cases}
\label{Lamt}
\en
which is obtained from $\bfLA_{m,n}$ 
by removing one vertical domino from the Young diagram 
 if $m>n-1$ and 
one horizontal domino if $m\le n-1$. 
Hence the 2-height is 
\eq
\text{ht}(\widetilde {\bold{\Lambda}}) = 
\begin{cases}
\tfrac12n(n+1) -1, & m >n-2, \\ 
\tfrac12(2n-m)(m+1),\qquad & m\le n-2.
\end{cases}
\label{height2}
\en

Carefully evaluating (\ref{Psi}), we deduce that when $m=n-1$ then 
\begin{align}
\Psi_{\widetilde \bfla}^{(\alpha)} &=
-(-1)^{m } 
\prod_{j=1}^{m} \lf ( 
\prod_{k=1}^{m+1-j}\lf( 
2m+3-j-k +\alpha \ri) 
\prod_{k=m+2-j}^{m+1}\lf( 
2m+3-j-k +\alpha \ri) \ri) 
\nn \\ 
&\quad \quad \quad \times \prod_{j=m+1}^{m+1}\prod_{k=2}^{m+1} (2m+3-j-k+\alpha) \,. 
\end{align}
When $m>n-2$ then 
\begin{align}
\Psi_{\widetilde \bfla}^{(\alpha)} &=
-(-1)^{n(m+1)}\prod_{j=1}^{n-1} 
\prod_{k=1}^{m+1-j} ( m+n+2-j-k+\alpha )
\prod_{k=2}^{m+1-n} ( m+n+2-(n)-k+\alpha )
\nn\\ 
& \quad \quad \times \prod_{j=1}^n \prod_{k=m+2-j}^{m+1} 
( m+n+2-j-k+\alpha),
\end{align}
and when $m\le n-2$ then 
\begin{align}
\Psi_{\widetilde \bfla}^{(\alpha)} &= 
-(-1)^{n(m+1)} 
\prod_{j=1}^{n-1}\prod_{k=1}^{m+1-j}
( m+n+2-j-k+\alpha)
\prod_{j=1}^m \prod_{k=m+2-j}^{m+1} ( m+n+2-j-k+\alpha)
\nn \\ 
& \quad \quad \times 
\prod_{j=m}^{n-1} \prod_{k=3}^{m+2}( m+n+2-j-k+\alpha ) 
\prod_{j=m+1}^{n-1}\prod_{k=1}^1 ( m+n+2-j-k+\alpha).
\end{align}

We notice that in each case $\Psi_{\widetilde \bfla}^{(\alpha)}$ includes all 
terms of the form $h_{j,k}+\alpha$ where $h_{j,k}$ are the hooks of the
 Young diagram of $\bfla$ except for the term $m+1+\alpha$. 
Therefore 
\eq
(m+1+\alpha) \Psi_{\widetilde \bfla}^{(\alpha)} = 
-(-1)^{n(m+1)} \prod_{j=1}^n\prod_{k=1}^{m+1}
(h_{j,k}+\alpha)= -\Psi_{\bfla}^{(\alpha)}\,.
\en
We conclude that 
\eq
r_1^{(\alpha)} = n(m+1) \frac {\Psi_{\bfla}^{(\alpha)}}
{\Psi_{\widetilde \bfla}^{(\alpha)}}
=-n(m+1) (\alpha+m+1)\,. 
\en
and 
\eq
d_{1}^{(\alpha)}=-n(m+1) (\mu+m+n+1)\,. 
\en
\end{proof}

\begin{conjecture}
The hook multiset $\mathcal{H}_{m,n}$ (\ref{hookmulti}) has the form 
\eq
\mathcal{H}_{m,n}= 
\begin{cases}
 \{k^{p_1}\}_{k=1}^{m}
\cup \{k^{p_2}\}_{k=m+1}^{n}
\cup 
 \{k^{p_3}\}_{k=n+1}^{m+n},\qquad & n>m, \\[3pt] 
 \{k^{p_1}\}_{k=1}^{n}
\cup \{k^{\widetilde{p}_2}\}_{k=n+1}^{m+1}
\cup 
 \{k^{p_3}\}_{k=m+2}^{m+n},\qquad & n \le m,
 \end{cases}
 \label{hooks_mult}
\en
where 
$$ p_1=k\quad,\quad p_2=m+1\quad,\quad \widetilde{p}_2=n \quad,\quad p_3=m+n+1-k,
$$ are the multiplicities of the hooks in each respective set.
The discriminant of $T_{m,n}^{(\mu)}(z)$ for $n>m$ in terms of partition data is 
\begin{align}
\text{Dis}_{m,n}(\mu)&=
(-1)^{(m+1) \floor{n/2} } c_{m,n}^{n(m+1) -1} \nn \\
& \qquad \times \prod_{k=1}^{m} k^{2k(n-k)(k-1-m)} 
\prod_{k=1}^{m} k^{k p_1^2} \,(\mu+n+k)^{f(n-1,p_1) }
\nn \\
& \qquad\times 
\prod_{k=m+1}^{n} k^{k p_2^2} \,(\mu+n+k)^{f(m+n-k, p_2)} 
\prod_{k=n+1}^{m+n} k^{ k p_3 ^2} \,(\mu+n+k)^{f(m,p_3 )},
\label{dis1}
\end{align}
where $f(k,p) =kp^2 -p(p-1)(p-2)/3 $. Similarly the discriminant 
when $n \le m$ is 
\begin{align}
\text{Dis}_{m,n}(\mu)&=
(-1)^{(m+1) \floor{n/2} } c_{m,n}^{2(n(m+1) -1)} 
\prod_{k=1}^{m} k^{2k(n-k)(k-1-m)} 
\prod_{k=1}^{n} k^{k p_1^2} \,(\mu+n+k)^{f(n-1,p_1) }
\nn \\
& \qquad\times 
\prod_{k=n+1}^{m} k^{k \widetilde{p}_2^2} \,(\mu+n+k)^{f(k-1, \widetilde{p}_2)} 
\prod_{k=m+1}^{m+n} k^{ k p_3 ^2} \,(\mu+n+k)^{f(m,p_3 )}.
\label{dis22}
\end{align}
\end{conjecture}
The discriminant representations (\ref{dis1}) and (\ref{dis22}) 
follow directly from rewriting (\ref{dis}) and (\ref{dis2}) in 
terms of the hooks and their multiplicities as defined by (\ref{hooks_mult}). 

As already mentioned, the E- and F-type blocks seen for large positive 
and negative values of $\mu$ are of size $m+1 \times n$ and therefore 
resemble the rectangular Young diagram of $\lambda$. 
Moreover, the three allowed sets of block structures corresponding 
to intermediate values of $\mu$, as given in table~\ref{tab:ngtm}, 
appear at $\mu+n+k=0$ where the multiplicity of the first column hook $k$ in $\bold{h}_{\bfla}$ changes its 
multiplicity type from type $p_1$ to $p_2$ to $p_3$.

\begin{conjecture}
\label{conj:rootangles}
Finally, the set of 
integers encoding the $n^{\rm th}$ roots of $\pm 1$ 
via the polynomials in Conjecture~\ref{angles} 
are 
 the hooks on the diagonals parallel to the main diagonal of the Young diagram of $\bfla$. 
 Specifically, as $\ep \to 0$
 for $\mu=-n-j-\ep$, 
 hook $h_{jk}$ in column $j$ contributes an 
$h_{jk}^{\rm th}$ root of unity 
 if $k$ is odd and an $h_{jk}^{\rm th}$ root of $-1$ if
 $k$ is even. For $\mu=-n-j\mp\ep$ 
 the polynomials in Conjecture~\ref{angles} are 
\begin{align*}
& \ \ \prod_{k=1}^j \ \ z^{h_{j,k}}\mp(-1)^{n+k}, && j =1,2, \dots, m+1, \\
& \prod_{k=j-m}^n z^{h_{j,k}}\mp(-1)^{n+k}, && j =m+2 ,m+3,\dots, n, \\
& \prod_{k=j-m}^n z^{h_{j,k}}\mp(-1)^{n+k}, && j=n+1 , n+2,\dots, m+n ,
\end{align*}
when $n>m$ where $h_{j,k}\in\mathcal{H}_{m,n}$. For $n\le m$ the result is 
\begin{align*}
&\ \ \prod_{k=1}^j z^{h_{j,k}}\mp(-1)^{n+k}, && j=1, 2,\dots, n, \\
&\ \ \prod_{k=1}^n z^{h_{j,k}}\mp(-1)^{n+k}, && j=n+1,n+2 ,\dots, m+1, \\
&\prod_{k=j-m}^n z^{h_{j,k}}\mp(-1)^{n+k}, && j=m+2,m+3,\dots ,m+n. 
\end{align*}
\end{conjecture}
\begin{remark}
 The result follows from Conjecture \ref{angles} by rewriting 
 the hook multiset (\ref{hookmulti}) as 
\eq
\mathcal{H}_{m,n}=
\begin{cases}
\{\{n{+}j{+}1{-}2k\}_{k=1}^j \}_{j=1}^{m+1} \cup 
\{\{n{+}j{+}1{-}2k\}_{k=j-m}^j \}_{j=m+2}^{n} \cup \{\{n{+}j{+}1{-}2k\}_{k=j-m}^n \}_{j=n+1}^{m+n}, &n>m, \\ 
\{\{n{+}j{+}1{-}2k\}_{k=1}^j \}_{j=1}^{n} \cup 
\{\{n{+}j{+}1{-}2k\}_{k=1}^n \}_{j=n+1}^{m+1} \cup \{\{n{+}j{+}1{-}2k\}_{k=j-m}^n \}_{j=m+2}^{m+n}, &n\le m. 
\end{cases}
\label{hookms}
\en
We illustrate how to determine the root angle polynomials from a Young diagram 
in Figure~\ref{fig:rootangles} for the example~\ref{ex_T23}
of $\Tmn{2,3}{\mu}(z)$. 
\end{remark}

\begin{remark}
We have found other families of Wronskian Hermite and Wronskian Laguerre polynomials for which properties can be written compactly in terms of partition data. Combinatorial concepts also appeared in the studies of special polynomials associated with \peqs\ in \cite{RefUm98,RefUm01,RefUm20,refNoumi2,refBDS,refBon}. We are currently investigating this curious appearance of partition combinatorics in various aspects of Wronskian polynomials. 
\end{remark}
\begin{figure}
 \centering 
\begin{tikzpicture}[scale=0.6
,draw/.append style={thick,black},
 baseline=(current bounding box.center)]
 \

\node[] at (1.25,0.75) {$\color{gray!90}{1}$};
\node[] at (0.25,0.75) {$2$};

\node[] at (-0.75,0.75) {$\color{gray!90}{3}$};

\node[] at (-0.25,0.25) {$\searrow$};
\node[] at (0.75,0.25) {$\searrow$};
\node[] at (1.75,0.25) {$\searrow$};

\node[] at (-0.75,-0.25) {$4$};
\node[] at (-0.75,-1.25) {$\color{gray!90}{5}$};

\node[] at (-0.25,-.75) {$\searrow$};
\node[] at (-0.25,-1.75) {$\searrow$};

 \tableauRow=-1.5
 \foreach \Row in {
{\color{gray!90}{5},\color{black}{4},\color{gray!90}{3}},
{4,\color{gray!90}{3},\color{black}{2}},
{\color{gray!90}{3},2,\color{gray!90}{1}}} {
 \tableauCol=0.5
 \foreach\k in \Row {
 \draw[thin](\the\tableauCol,\the\tableauRow)rectangle++(1,1);
 \draw[thin](\the\tableauCol,\the\tableauRow)+(0.5,0.5)node{$\k$};
 \global\advance\tableauCol by 1
 }
 \global\advance\tableauRow by -1
 }
\end{tikzpicture}
 \caption{The hooks on the $j^{\rm th}$ diagonal 
of the Young diagram of $\Tmn{2,3}{\mu}$ 
encode the behaviour of the roots that coalesce at the origin at $\mu=-n-j-\ep$
through the polynomials in Conjecture \ref{conj:rootangles}.
When $j=3$ the polynomial is $(z^5-1)(z^3+1)(z-1)$ and when $j=3$ or $j=5$ the polynomial is $z^3-1$.
}
 \label{fig:rootangles}
\end{figure}
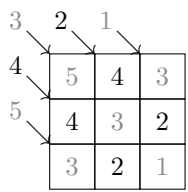



 


 

\subsection*{Acknowledgements}
We thank David G\'omez-Ullate, Davide Masoero and Bryn Thomas
for helpful comments and illuminating discussions. {We also thank the reviewers for their constructive comments and suggestions.}

\comment{\subsection*{Conflict of interest statement}
The authors declare no conflict of interest.
\subsection*{Data availability statement}
This paper has no associated data.
\subsection*{Orcid}
Peter Clarkson: 0000-0002-8777-5284 \\
Clare Dunning: 0000-0003-0535-9891 }

\def\cpam{Commun. Pure Appl. Math.}
\def\CPAM{Commun. Pure Appl. Math.}
\def\funk{Funkcial. Ekvac.}
\def\FUNK{Funkcial. Ekvac.}
\def\CUP{Cambridge University Press}

\def\refpp#1#2#3#4{\vspace{-0.2cm}
\bibitem{#1} \textrm{\frenchspacing#2}, \textrm{#3}, #4.}

\def\refjltoap#1#2#3#4#5#6#7{\vspace{-0.2cm}
\bibitem{#1}\textrm{\frenchspacing#2}, \textrm{#3},
\textit{\frenchspacing#4}\ (#6)\ #7.}

\def\refjl#1#2#3#4#5#6#7{\vspace{-0.2cm}
\bibitem{#1}\textrm{\frenchspacing#2}, \textrm{#3},
\textit{\frenchspacing#4}, \textbf{#5}\ (#6)\ #7.}

\def\refbk#1#2#3#4#5{\vspace{-0.2cm}
\bibitem{#1} \textrm{\frenchspacing#2}, \textit{#3}, #4, #5.}

\def\refcf#1#2#3#4#5#6#7{\vspace{-0.2cm}
\bibitem{#1} \textrm{\frenchspacing#2}, \textrm{#3},
in: \textit{#4}, {\frenchspacing#5}, pp.\ #6, #7.}

\def\ams{American Mathematical Society}
\def\DE{Diff. Eqns.}
\def\JPA{J. Phys. A}
\def\NMJ{Nagoya Math. J.}
\def\JCAM{J. Comput. Appl. Math.}

\end{document}